\documentclass[letterpaper,11pt,UKenglish]{article}

\usepackage[margin = 1in]{geometry}
\usepackage{amssymb, amsthm}
\usepackage{mathtools}
\usepackage{pgfplots}
\usepackage{hyperref}
\usepackage{subcaption}
\usepackage{float}

\DeclareMathOperator{\WL}{WL}
\newcommand\doubleBrackets[1]{\{\!\!\{ #1 \}\!\!\}}
\newcommand{\N}{\mathbb{N}}
\newcommand{\curlC}{\mathcal{C}}

\newtheorem{theorem}{Theorem}[section]
\newtheorem{lemma}[theorem]{Lemma}
\newtheorem{corollary}[theorem]{Corollary}
\newtheorem{observation}[theorem]{Observation}
\newtheorem{notation}[theorem]{Notation}
\newtheorem{proposition}[theorem]{Proposition}
\theoremstyle{definition}
\newtheorem{definition}[theorem]{Definition}
\newtheorem{assumption}[theorem]{Assumption}
\theoremstyle{remark}

\title{A Classification of Long-Refinement Graphs for Colour Refinement}

\author{Sandra Kiefer\\ 
University of Oxford\\
\tt{sandra.kiefer@cs.ox.ac.uk}
\and
T.\ Devini de Mel\\ 
University of Oxford\\
\tt{devini.demel@gmail.com}
}

\pgfplotsset{compat=1.18}
\begin{document}

\maketitle

\begin{abstract}
The Colour Refinement algorithm is a classical procedure to detect symmetries in graphs, whose most prominent application is in graph-isomorphism tests. The algorithm and its generalisation, the Weisfeiler--Leman algorithm, evaluate local information to compute a colouring for the vertices in an iterative fashion. Different final colours of two vertices certify that no isomorphism can map one onto the other. 

The number of iterations that the algorithm takes to terminate is its central complexity parameter. For a long time, it was open whether graphs that take the maximum theoretically possible number of Colour Refinement iterations actually exist. Starting from an exhaustive search on graphs of low degrees, Kiefer and McKay proved the existence of infinite families of such long-refinement graphs with degrees $2$ and $3$, thereby showing that the trivial upper bound on the iteration number of Colour Refinement is tight. 

In this work, we provide a complete characterisation of the long-refinement graphs with low (or, equivalently, high) degrees. We show that, with one exception, the aforementioned families are the only long-refinement graphs with maximum degree at most $3$, and we fully classify the long-refinement graphs with maximum degree $4$.  To this end, via a reverse-engineering approach, we show that all low-degree long-refinement graphs can be represented as compact strings, and we derive multiple structural insights from this surprising fact. Since long-refinement graphs are closed under taking edge complements, this also yields a classification of long-refinement graphs with high degrees.

Kiefer and McKay initiated a search for long-refinement graphs that are only distinguished in the last iteration of Colour Refinement before termination. We conclude it in this submission by showing that such graphs cannot exist.
\end{abstract}

\section{Introduction}\label{section:intro}

Colour Refinement, also known as Naïve Vertex Classification or the $1$-dimensional Weisfeiler--Leman algorithm, is a classical combinatorial procedure to detect symmetries in graphs. The algorithm computes in iterations a colouring of the vertices that reflects how the vertices are connected to the rest of the graph. By assessing local structure, it refines in each iteration the colouring from the previous one, until it reaches a stable state, which is when it terminates and returns the current colouring. The colour of each vertex in an iteration consists of (an encoding of) the colours from the previous iteration of the vertex and the vertices in its neighbourhood. The colourings are isomorphism-invariant: if two vertices obtain different colours, then there is no isomorphism between the graphs mapping one vertex to the other. This way, the colourings can decrease the size of the search space for isomorphisms drastically, and, in fact, the most prominent application of the algorithm is in graph-isomorphism tests and graph-canonisation approaches. Colour Refinement can be implemented to run in time $O((m + n) \log(n))$ \cite{carcro82,McKay81}, where $n$ and $m$ are the number of vertices and edges, respectively, of the input graph, and efficient implementations are used in all competitive isomorphism solvers, see e.g.\ \cite{AndersS21,DargaLSM04,JunttilaK07,JunttilaK11,McKay81,mckaypip14}.

Besides the practical implementations, the algorithm has also fascinated theoreticians over decades. The main reason for this is that it keeps revealing more and more links between various areas in computer science. For example, the distinguishing power of the algorithm corresponds to the expressivity of graph neural networks \cite{morritfey+19}, and there are machine-learning kernel functions based on the algorithm \cite{LiSSSS16,ShervashidzeSLMB11}. Two graphs have the same final colourings with respect to the algorithm if and only if there is a fractional isomorphism between them \cite{god97,ramscheiull94,tin91}. This is equivalent to no sentence in the $2$-variable fragment of the counting logic $\textsf{C}$ distinguishing the graphs \cite{cfi}, and also to the graphs agreeing in their homomorphism counts for all trees \cite{DellGR18}. 
Remarkably, these last two results generalise naturally to higher-dimensional versions of the Weisfeiler--Leman algorithm (\cite{Dvorak10,immlan90}, see also \cite{WeisfeilerL68} for the original $2$-WL algorithm), whose expressive power is a very active area of research. Most notably, the algorithm serves as a central subroutine in Babai's celebrated quasipolynomial-time graph-isomorphism algorithm \cite{Babai16}. 

The connections that the algorithm exhibits to other areas such as logic, homomorphism counting, and machine learning extend the toolbox of proof methods to understand the algorithm better. While the expressive power of Colour Refinement has been completely characterised \cite{ArvindKRV17,kieschweisel22}, its computational complexity still sparks research. The central complexity measure that has tight links in the related fields is the number of iterations that the algorithm uses to reach its output. This number of colour recomputations is crucial for the parallel running time of Colour Refinement, since $\ell$ iterations can be simulated with $O(\ell)$ steps on a PRAM with $O(n)$ processors \cite{groverb06,KoblerV08}. Also, when run on two graphs in parallel, the iteration number corresponds to a bound on the descriptive complexity of the difference between the graphs. More precisely, the number of iterations until the graphs obtain different colourings equals the minimum quantifier depth of a distinguishing sentence in the counting logic $\textsf{C}^2$ \cite{cfi,immlan90}. There are also interesting connections to proof complexity \cite{AtseriasM13,AtseriasO18} and database theory \cite{GR24}.

The algorithm proceeds by iteratively enriching the colouring from the previous iteration with more information, yielding with each iteration a finer partition of the vertices into colour classes. A graph that requires the maximum theoretically possible number of Colour Refinement iterations needs to provide as little “progress” as possible in each iteration, implying that the initial partition into colour classes needs to be the trivial one with all vertices of the same colour, whereas in the final partition, every vertex has to have its own distinct colour. Hence, the trivial upper bound on the iteration number on $n$-vertex graphs is $n-1$, with the number of colour classes growing by $1$ in each iteration. We call graphs that achieve this bound \emph{long-refinement graphs}. On random graphs, the iteration number is asymptotically almost surely $2$ \cite{BabErdSelSta80}. On the other hand, Krebs and Verbitsky managed to construct graphs with an iteration number of $n - O(\sqrt{n})$ \cite{krever15}. For a long time, no progress on the trivial upper bound $n-1$ was made, and it was open whether long-refinement graphs actually exist. Inspired by work due to Gödicke \cite{Goedicke} and starting from an exhaustive search on graphs of low degrees, Kiefer and McKay proved the existence of infinite families of long-refinement graphs with degrees $2$ and $3$, thereby showing that the trivial upper bound on the iteration number of Colour Refinement on $n$-vertex graphs is tight \cite{KMcK20}. 

The main motivation for the works \cite{Goedicke} and \cite{KMcK20} was to find out whether long-refinement graphs exist. With the discovered families, various distinctive properties of long-refinement graphs with degrees 2 and 3 became apparent. Those allowed the authors to find compact representations for the graph families, which ultimately made it possible to show that the families are infinite. The hope was to use these insights to construct new families of long-refinement graphs of degrees $2$ and $3$ that would cover all possible graph sizes (starting from the smallest possible $n$, which is $10$). However, infinitely many gaps have persisted since then, and in particular, it has remained open for which of the remaining values of $n$ there exist long-refinement $n$-vertex graphs at all.

\paragraph*{Our Contribution}

In this work, we provide a complete characterisation of the long-refinement graphs with small (or, equivalently, high) degrees. More concretely, we show that, with one exception of a graph with minimum degree $1$, the aforementioned families are the only long-refinement graphs with maximum degree at most $3$, and we fully characterise the long-refinement graphs with maximum degree at most $4$. Since complements of long-refinement graphs are long-refinement graphs as well, this also yields a characterisation of the long-refinement graphs with high degrees (i.e.\ with minimum degree at least $n-5$). While it clearly strengthens and generalises the results from \cite{Goedicke} and \cite{KMcK20}, our approach does not use experimental results and is entirely theoretical.

By the definition of Colour Refinement, it is easy to see that every long-refinement graph has exactly two different vertex degrees. From the facts that the algorithm iteratively refines a partition of the vertex set and that long-refinement graphs need to end up with the finest partition possible, we deduce that when executing Colour Refinement on a long-refinement graph, there will be a stage in which each colour class has size at most $2$. We analyse the possible connections between these pairs and singletons, depending on the indices of the iterations in which the pairs split into singletons. This generalises the analyses from \cite{KMcK20}, since our results hold for any degree pairs. Our structural study gives us a clear understanding of the order and edges between splitting pairs in a long-refinement graph. We then head off to reverse-engineer the larger colour classes from earlier iterations step by step. Again, our results hold for long-refinement graphs with no degree restrictions, but the analysis becomes more and more complex with larger colour classes.

Restricting ourselves to small degrees, we obtain our classification of long-refinement graphs of maximum degree at most $4$. With this, we manage to show that all even-order long-refinement graphs $G$ with vertex degrees $2$ and $3$ have a stage in which each colour class has size $2$, and all odd-order long-refinement graphs with these degree restrictions have a stage in which each colour class has size at most $2$ and there is only one singleton colour class. This is a surprising insight because the corresponding statement was phrased as a requirement in \cite{KMcK20}, while we show that it is in fact a general property. It implies that \emph{every} long-refinement graph with degrees $2$ and $3$ has a simple compact string representation. It also allows us to capture more restrictions on symmetries in the strings that represent long-refinement graphs, ultimately yielding that the families from \cite{KMcK20} are exhaustive and that the gaps in sizes that the work could not cover occurred because there are no long-refinement graphs with those sizes and degrees simultaneously. From the characterisation via string representations, we can then easily deduce that there is exactly one long-refinement graphs with minimum degree $1$ and maximum degree at most $3$.

Since Colour Refinement is usually applied to two graphs with the purpose of distinguishing them through their colourings, a natural question to ask -- also in the light of the lower bound from \cite{krever15} -- is whether there are actually \emph{pairs of graphs} that take as many refinements as theoretically possible to be distinguished. This question was initiated in \cite{KMcK20} and has been open since then. Our aforementioned results give us a complete understanding of the structure and splitting dynamics of long-refinement graphs with low (or high) degrees. No two of those long-refinement graphs are distinguished only in the last iteration of Colour Refinement on them, but there could in principle be higher-degree long-refinement graphs with this property. In this submission, we conclude the search by proving that there cannot be a ``long-refinement pair'' of graphs, i.e.\ a pair of long-refinement graphs (of any degrees) that is only distinguished in the last iteration of Colour Refinement before termination.

\paragraph*{Further Related Work}

It is worth mentioning that just as for Colour Refinement, for general $k$, 
the iteration number $\WL_k(n)$ of the $k$-dimensional Weisfeiler--Leman algorithm ($k$-WL) on $n$-vertex graphs is being studied. For $2$-WL, the paper \cite{kieschw19} used a simple combinatorial game to show that the trivial upper bound of $\WL_2(n) \leq n^2-1$ on the iteration number on $n$-vertex graphs (as well as the improved bound one obtains merely accounting for the fact that the initial partition for $k$-WL with $k \geq 2$ always has multiple colour classes) is asymptotically not tight. In \cite{lichponschwei19}, algebraic machinery was employed to show that $\WL_2(n) \in O(n \log(n))$. This was recently generalised to $O(kn^{k-1}\log(n))$ for $k \geq 3$ in \cite{GroheLN23}. Concerning lower bounds, Fürer's $\WL_k(n) \in \Omega(n)$ remained unbeaten over decades \cite{Furer01}. A few years ago, significant progress was made when proceeding to relational structures of higher arity \cite{BerkholzN16} (see also \cite{GroheLN23}). Only very recently, Fürer's bound was also improved for graphs, with $\Omega(n^{k/2})$ being the currently best bound \cite{GroheLNS23com}.

Concerning the number of iterations until two graphs are distinguished, it has been shown that on graphs of bounded treewidth, on planar graphs, and on interval graphs, there is a constant dimension of the Weisfeiler--Leman algorithm that distinguishes a graph from such a class from every second one in $O(\log(n))$ iterations \cite{ GroheK21,groverb06,vBGKO23}.

\section{Overview of Results and Methods}
This section serves as an overview of our main results and the methods that we use. All results formally stated in this section as well as the short proofs for Lemmas \ref{lem:deg-13-singleton-intro} and \ref{lem:deg-13-pair-intro} are collected from the later sections. 

Given a (vertex-coloured or monochromatic) input graph, Colour Refinement uses a simple local criterion to iteratively refine the partition of the vertices induced by the current colouring. As soon as a stable state with respect to the criterion is reached, the algorithm returns the graph together with the current vertex colouring.

\begin{definition}[Colour Refinement]
	Let $G$ be a graph with vertex set $V(G)$, edge set $E(G)$, and vertex-colouring $\lambda \colon V(G)\rightarrow\curlC$. 
	On input $(G,\lambda)$, the Colour Refinement algorithm computes its output recursively with the following update rules. Letting $\chi^0 \coloneqq \lambda$, for $i \in \N$, the \emph{colouring after $i$ iterations of Colour Refinement} is
		\[\chi^i(v) \coloneqq (\chi^{i-1}(v),\doubleBrackets{\chi^{i-1}(u) | \{u,v\} \in E(G)}).\]
\end{definition}

Thus, two vertices obtain different colours in iteration $i$ precisely if they already had different colours in iteration $i-1$ or if they differ in the $\chi^{i-1}$-colour multiplicities in iteration $i-1$ among their neighbours.

Letting $\pi^i \coloneqq \pi_{G,\lambda}(\chi^i)$ be the vertex partition induced by the colouring after $i$ iterations of Colour Refinement, it is immediate that $\pi^{i+1}$ is at least as fine as $\pi^{i}$ (we write $\pi^{i+1} \preceq \pi^{i}$), for every $i \in \N$, and that, hence, there is a minimal $j \leq |V(G)|-1$ such that $\pi^{j+1} = \pi^{j}$. With this choice of $j$, the output of Colour Refinement on input $(G,\lambda)$ is $\chi_{G,\lambda} \coloneqq \chi^j$.

This paper studies long-refinement graphs, i.e.\ graphs $G$ on which the Colour Refinement procedure takes $n-1$ iterations to terminate, where $n = |V(G)|$. It is easy to see that this is the maximum theoretically possible number of iterations on an $n$-vertex graph. To analyse long-refinement graphs, we first recall some very basic insights about the refinement iterations on them. Here, since Colour Refinement refines the colour of each vertex based on the multiplicity of colour occurrences in its neighbourhood, the degrees of vertices (i.e.\ their numbers of neighbours) are crucial. 

Colour Refinement terminates as soon as the vertex partition induced by the colours is not strictly finer than the one in the previous iteration. Hence, to obtain $n-1$ iterations until termination, every iteration on a long-refinement graph must increase the number of colour classes by exactly $1$, starting with a single colour class containing all vertices and ending with a partition in which each vertex forms its own colour class. This implies in particular that there are precisely two vertex degrees in the graph and that in each iteration exactly one colour class is split into two colour classes, which then induces the next split. In Section \ref{sec:colref}, with Lemma \ref{rmk:work-backwards}, we formalise the notion of \emph{inducing a split} by introducing \emph{unbalanced} pairs of colour classes. The lemma is the central ingredient for the structural arguments in Sections \ref{sec:rev-eng} and \ref{sec:max4}.

In Section \ref{sec:rev-eng}, we first observe that every long-refinement graph must have a stage in which every colour class has size at most $2$, call that iteration $p$ and the induced vertex partition $\pi^p$. From then on, there is a linear splitting order among the classes of size $2$, the \emph{pairs}. The first part of Section \ref{sec:rev-eng} analyses the connections between these classes, completely dropping the degree requirements from \cite{KMcK20} and sharpening the insights about the partitions. The first split of a pair must be induced by the splitting of a class of size $4$ into two pairs, which we denote $P_a$ and $P_b$. See Figure \ref{h=0-intro} for the situation when $P_a$ and $P_b$ are successive in the linear order of pairs, and when $P_b$ is the successor of the successor of $P_a$.

    \begin{figure}[htpb]
    \centering
        \begin{subfigure}[htpb]{0.24\textwidth}
    	\centering
    	\includegraphics[scale=0.8]{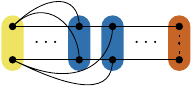}
    	\caption{$\pi^{p-1}$ if $b=a+1$}
    \end{subfigure}
    \begin{subfigure}[htpb]{0.24\textwidth}
    	\centering
    	\includegraphics[scale=0.8]{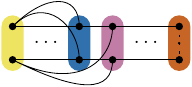}
    	\caption{$\pi^p$ if $b=a+1$}
    \end{subfigure}
    \hspace{0.1cm}
        \begin{subfigure}[htpb]{0.24\textwidth}
    	\centering
    	\includegraphics[scale=0.8]{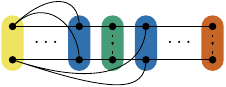}
    	\caption{$\pi^{p-1}$ if $b=a+2$}
    \end{subfigure}
    \begin{subfigure}[htpb]{0.24\textwidth}
    	\centering
    	\includegraphics[scale=0.8]{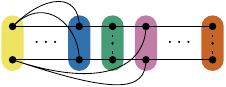}
    	\caption{$\pi^p$ if $b=a+2$}
    \end{subfigure}
    \caption{The partitions $\pi^{p-1}$ and $\pi^{p}$, where the splitting order of pairs is from left to right. $P_a$ is blue, $P_b$ is first blue and then pink, respectively. For clarity, we omit all edges between non-consecutive pairs apart from the ones connecting $P_a$ and $P_b$ to the minimum of the linear order.}\label{h=0-intro}
\end{figure}

The iterations prior to this see the pairs symmetrically preceding $P_a$ and succeeding $P_b$ resulting from splits of classes of $4$ themselves. We thus reverse-engineer the first partition to contain no classes larger than $4$. This partition may or may not contain a class of size $3$, and the  analysis of partitions in the preceding iterations -- wherein we reverse-engineer the classes of size $6$ -- considers both cases. The statements in the second part of Section \ref{sec:rev-eng} read more technical, due in part to the case distinction required for if there is a class of size $3$, but a closer look at them lets us understand a similar argument and a similar symmetry between the iterations where classes of size $4$ split into two pairs, and those where classes of size $6$ split into a class of size $4$ and a pair. Throughout Section \ref{sec:rev-eng}, arguments about the classes allow us to make observations about the edges incident to the vertices within them, and we gain general insights about the structure of all long-refinement graphs, without any assumptions on degree.

We employ these insights in Section \ref{sec:max4}. Here, we first revisit the compact string notation for long-refinement graphs of degrees $2$ and $3$ that was introduced in \cite{KMcK20}.

\begin{notation}\label{notation:23-intro}
	Let $\Sigma \coloneqq \{\mathrm{S},\mathrm{X},\mathrm{X}_2,0,0_2,1,1_2\}$, and let $n_\mathcal{P}$ be the number of colour classes of size $2$ in the minimal iteration $p$ in which all colour classes have size at most $2$. By $\pi^p_{\mathcal{S}}$, denote the set of colour singletons in iteration $p$. 
    
    We define $\sigma \colon [n_\mathcal{P}] \rightarrow \Sigma$ as the function that maps $i$ to the element in $\Sigma$ that matches the type of $P_i$. More precisely, $\sigma(i)$ is defined as follows:
	\begin{itemize}
		\setlength\itemsep{0em}
		\item $\mathrm{S}$, \ if $P_i = \min(\prec) = P_1$.
		\item $\mathrm{X}$, \ if $P_i \in \{P_a,P_b\}$ and $P_i$ is not adjacent to any singleton $S \in \pi^p_{\mathcal{S}}$. 
		\item $\mathrm{X}_2$, \ if $P_i \in \{P_a,P_b\}$ and $P_i$ is adjacent to a singleton $S_i \in \pi^p_{\mathcal{S}}$.
		\item $0$, \ if $P_i \notin \{P_1,P_a,P_b\}$, $\deg(P_i)=2$ and $P_i$  is not adjacent to any singleton $S \in \pi^p_{\mathcal{S}}$.
		\item $0_2$, \ if $P_i \notin \{P_1,P_a,P_b\}$, $\deg(P_i)=2$, and $P_i$ is adjacent to a singleton $S_i  \in \pi^p_{\mathcal{S}}$.
		\item $1$, \ if $P_i \notin \{P_1,P_a,P_b\}$ with $\deg(P_i)=3$, such that $P_i$ is not adjacent to any singleton $S \in \pi^p_{\mathcal{S}}$.
		\item $1_2$, \ if $P_i \notin \{P_1,P_a,P_b\}$, $\deg(P_i)=3$, and $P_i$ is adjacent to a singleton $S_i  \in \pi^p_{\mathcal{S}}$.
	\end{itemize}
	Since $\prec$ is a linear order, we can define a string representation based on a long-refinement graph $G$, where the $i$-th letter in the string is $\sigma(i)$ and corresponds to $P_i$, the $i$-th element of $\prec$. We call a string defined in this way a \emph{long-refinement string}.  
\end{notation}

Since the paper \cite{KMcK20} was concerned with proving the existence of long-refinement graphs, it sufficed to use the string representations to find infinitely many long-refinement graphs. Here, we are concerned with a \emph{complete characterisation}, and to this end, we use the results from Section~\ref{sec:rev-eng} to detect precisely which strings over the alphabet yield long-refinement graphs. We find, for example, that every long-refinement graph with degrees $2$ and $3$ has a stage in which all classes have size $2$ except for one singleton in the case where the order of the graph is odd. This yields the helpful insight that we can use the above alphabet to represent all long-refinement graphs with degrees $2$ and $3$ unambiguously, since the strings describe the exact neighbourhood of each pair and each singleton, as long as there is at most one singleton. Thus, up to isomorphism, every string $\Xi$ represents exactly one graph, which we denote by $G(\Xi)$.

Using the technical insights from Section \ref{sec:rev-eng}, Section \ref{sec:max4} generalises the results beyond degrees $2$ and $3$, namely to maximum degree at most $4$. This yields the first comprehensive classification of such long-refinement graphs. Through the symmetries defined in Section \ref{sec:rev-eng}, we are able to divide this argument into cases depending on the number of pairs that remain in the partition, and not merged into larger classes. 

In Section \ref{sec:classification}, with Theorems \ref{thm:classification23} -- \ref{thm:classification3}, we collect all results about the structure of long-refinement graphs of low degrees to obtain a compact classification, ordered by degree pairs.  Moreover, we provide an alternative and simple proof that there is exactly one long-refinement graph with minimum degree $1$ and maximum degree at most $4$, which is the graph in Figure \ref{fig:d=l+3l+4-graphs-intro}. We append these results in the following subsection (Subsection \ref{subsec:classification}), with references to the figures and tables that represent the graphs. Note that, since a graph is a long-refinement graph if and only if the graph obtained by complementing the edge relation is a long-refinement graph, the results directly yield a classification of long-refinement graphs with high degrees.

Finally, in Section \ref{sec:long-ref-distinguish}, we turn towards the distinguishability of long-refinement graphs. In its applications, Colour Refinement is used to distinguish two graphs. Therefore, given our insights about individual long-refinement graphs, it is natural to ask for graphs that take $n-1$ iterations to be distinguished by Colour Refinement, i.e.\ to obtain distinct colourings (which implies that they are not isomorphic). The strongest known lower bound on the number of iterations to distinguish two graphs with Colour Refinement is $n - O(\sqrt{n})$ \cite{krever15} and it is from a decade ago \cite{krever15}. With the results from \cite{KMcK20}, it is plausible to search for long-refinement pairs of graphs, and hence this was posed as an open question in the same paper. Again using the fact that the last split in both graphs must be a pair splitting into singletons, we show that such pairs of graphs do not exist. 

\begin{theorem}\label{thm:nolongdist-intro}
    For $n \geq 3$, there is no pair of graphs $G$, $H$ with $|V(H)| \leq |V(G)| = n$ that Colour Refinement distinguishes only after $n - 1$ iterations.
\end{theorem}

\subsection{Our Classification of Long-Refinement Graphs with Small Degrees}\label{subsec:classification}

This subsection lists our main classification results, see also Theorems \ref{thm:classification23} -- \ref{thm:classification3}. Since our entire analysis is up to isomorphism, we will not always distinguish explicitly between graphs from the same isomorphism class.

We write $N_G(v)$ for the neighbourhood of vertex $v$ in the graph $G$, and $\deg(G)$ for the set of vertex degrees in $G$. 

Recall Notation \ref{notation:23-intro}. The work \cite{KMcK20} identified families of long-refinement graphs $G$ with $\deg(G) = \{2,3\}$, but left as an open problem whether they are comprehensive. We answer that question affirmatively.

\begin{theorem}\label{thm:classification23-intro}
    The even-order long-refinement graphs $G$ with $\deg(G) = \{2,3\}$ are exactly the ones represented by strings contained in the following sets:
	\begin{itemize}
		\setlength\itemsep{0em}
		\item $\{\mathrm{S011XX}\}$
		\item $\{\mathrm{S1^{\mathit{k}}001^{\mathit{k}}X1X1^{\mathit{k}}0} \mid k \in \N_0\}$
		\item $\{\mathrm{S1^{\mathit{k}}11001^{\mathit{k}}XX1^{\mathit{k}}0} \mid k \in \N_0\}$
		\item $\{\mathrm{S1^{\mathit{k}}0011^{\mathit{k}}XX1^{\mathit{k}}10} \mid k \in \N_0\}$
		\item $\{\mathrm{S011(011)^{\mathit{k}}00(110)^{\mathit{k}}XX(011)^{\mathit{k}}0} \mid k \in \N_0\}$
		\item $\{\mathrm{S(011)^{\mathit{k}}00(110)^{\mathit{k}}1X0X1(011)^{\mathit{k}}0} \mid k \in \N_0\}$
	\end{itemize}

    The odd-order long-refinement graphs $G$ with $\deg(G) = \{2,3\}$ are exactly the ones represented by strings contained in the following sets:
\begin{itemize}
	\setlength\itemsep{0em}
        \item $\{\mathrm{S1_211XX}\}$
	\item $\{\mathrm{S0X1X_2}\} \cup \{\mathrm{S1^{\mathit{k}}1011^{\mathit{k}}X1X1^{\mathit{k}}1_2} \mid k \in \N_0\}$
	\item $\{\mathrm{S110XX_2}\} \cup \{\mathrm{S111^{\mathit{k}}1011^{\mathit{k}}XX1^{\mathit{k}}1_2} \mid k \in \N_0\}$
	\item $\{\mathrm{S1^{\mathit{k}}01^{\mathit{k}}1XX1^{\mathit{k}}1_2} \mid k \in \N_0\}$
	\item $\{\mathrm{S(011)^{\mathit{k}}00(110)^{\mathit{k}}X1_2X(011)^{\mathit{k}}0} \mid k \in \N_0\}$
\end{itemize}
\end{theorem}

In particular, the theorem implies that the gaps that the graphs found in \cite{kiefer:phd} leave open cannot be covered with degrees $2$ and $3$.

\begin{corollary}\label{no-LRgraphs-n-intro}
There are no long-refinement graphs $G$ with $\deg(G) = \{2,3\}$ and order $n \neq 12$ such that $n \bmod 18\in\{6,12\}$.
\end{corollary}

Theorem \ref{thm:classification23-intro} and Corollary \ref{no-LRgraphs-n-intro} completely classify the long-refinement graphs $G$ with $\deg(G)=\{2,3\}$. Expanding to $\max\deg(G) \leq 4$, we obtain the following further parts of our classification of long-refinement graphs with small degrees (using the results in Section \ref{sec:max4}).

\begin{theorem}\label{thm:deg-12-or-14-intro}
    There is no long-refinement graph $G$ with $\deg(G) = \{1,2\}$ or $\deg(G) = \{1,4\}$. 
\end{theorem}
\begin{theorem}\label{thm:deg-24-intro}
    The long-refinement graphs $G$ with $\deg(G) = \{2,4\}$ are precisely the ones in Figures \ref{fig:d=l-deg24-17-bigfig}, \ref{fig:d=l-deg24-14-bigfig}, \ref{fig:d=l-deg24-12-bigfig}, \ref{fig:d=l-deg24-13-bigfig}, \ref{fig:d=l-deg24-16-bigfig},  \ref{fig:d=l-deg24-18-bigfig}, \ref{fig:d=l-deg24-21-bigfig}, \ref{fig:d=l-deg24-27-bigfig}, \ref{fig:d=l+1-singleton-adjXX-1-bigfig},
    \ref{fig:d=l+1-singleton-adjXX-2-bigfig},
 \ref{fig:deg24-14a-bigfig}, and \ref{fig:deg24-24a-bigfig}. 
\end{theorem}

\begin{theorem}\label{thm:deg-34-intro}
    The long-refinement graphs $G$ with $\deg(G) = \{3,4\}$ are precisely the ones in Figures \ref{fig:d=l-deg34-11-bigfig}, \ref{fig:d=l-deg34-10-bigfig},  \ref{fig:d=l-deg34-12-bigfig}, \ref{fig:d=l-deg34-13A-bigfig}, \ref{fig:d=l-deg34-16-bigfig}, \ref{fig:d=l-deg34-21-bigfig}, \ref{fig:d=l-deg34-27-bigfig}, 
        \ref{fig:d=l+2-deg34-15-bigfig},
    as well as Tables \ref{tab:adj-list-d=l-34-intro}, \ref{tab:adj-list2-d=l-34-intro}, and \ref{tab:adj-list3-d=l-34-intro}.
\end{theorem}

\begin{theorem}\label{thm:classification3-intro}
    The graph in Figure \ref{fig:d=l+3l+4-graphs-intro} is the only long-refinement graph $G$ with $\deg (G) = \{1,3\}$. 
\end{theorem}

\begin{figure}[tpb]
    \centering
        \includegraphics{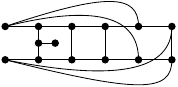}
    \caption{The unique long-refinement graph with degrees $\{1,3\}$.}
    \label{fig:d=l+3l+4-graphs-intro}
\end{figure}

While Theorems \ref{thm:deg-12-or-14-intro} -- \ref{thm:classification3-intro} are proved through an analysis of the refinement iterations in Section \ref{sec:max4}, we now present a simple alternative proof for Theorem \ref{thm:classification3-intro}, which relies merely on Theorem \ref{thm:classification23-intro}. By \cite[Lemma 13]{KMcK20}, a long-refinement graph $G$ with $\deg(G) = \{1,3\}$ has at most two vertices of degree $1$. We start by treating the case where $G$ has exactly one vertex of degree $1$.

\begin{lemma}\label{lem:deg-13-singleton-intro}
    The graph in Figure \ref{fig:d=l+3l+4-graphs-intro} is the only long-refinement graph $G$ with $\deg(G) = \{1,3\}$ in which $|\{v \in V(G) : \deg(v) = 1\}| = 1$.
\end{lemma}

\begin{proof}
     Let $G$ be any long-refinement graph with $\deg(G)=\{1,3\}$ and the set of vertices of degree $1$ being a singleton $\{v_1\}$. Then the partition after the first Colour Refinement iteration on $G$ is $\pi_G^1 = \{\{v_1\}, V(G) \setminus \{v_1\}\}$. The graph $G' \coloneqq G[V(G) \setminus \{v_1\}]$, i.e.\ the subgraph of $G$ induced by $V(G) \setminus \{v_1\}$, satisfies $\deg(G') = \{2,3\}$ and $\{v \in V(G') \mid \deg(v) = 2\} = \{v_2\}$, where $\{v_2\} = N_G(v_1)$. Also, Colour Refinement takes $n-2$ iterations to terminate on $G'$, otherwise $G$ would not be a long-refinement graph. Hence, $G'$ is a long-refinement graph as well, and since $|V(G)|$ is even, $|V(G')|$ is odd. By our classification in Theorem \ref{thm:classification23-intro}, $G'$ must be $G(\mathrm{S}1_211\mathrm{XX})$, since that is the only long-refinement graph with degrees $2$ and $3$ in which there is only one vertex of degree $2$. So, by degree arguments, $G$ must be the graph in Figure \ref{fig:d=l+3l+4-graphs-intro}.
\end{proof}

Now we consider the situation in which there are exactly two vertices of degree $1$.

\begin{lemma}\label{lem:deg-13-pair-intro}
    There is no long-refinement graph $G$ with $\deg(G) = \{1,3\}$ in which $|\{v \in V(G) : \deg(v) = 1\}| = 2$. 
\end{lemma}

\begin{proof}
     Suppose $G$ is a long-refinement graph with $\deg(G) = \{1,3\}$ for which $\{v \in V(G) : \deg(v) = 1\} = \{v_1, v'_1\}$ with $v_1 \neq v'_1$. We know that $\{v_1,v'_1\} \notin E(G)$, since $G$ is connected. However, now the graph $G' \coloneqq (V(G), E(G) \cup \{\{v_1,v'_1\}\}$ is a long-refinement graph with $\deg(G') = \{2,3\}$ and $\{v \in V(G') : \deg(v) = 2\} = \{v_1, v'_1\}$. The degree property is clear. To see that $G'$ is a long-refinement graph, note that inserting the new edge does not alter the partition into degrees, i.e.\ $\pi^1$. Also, since in both graphs, it holds that $|N(\{v_1,v'_1\})| \leq 2$, the class $\{v_1,v'_1\}$ must be the last class that splits (in both graphs). Furthermore, $N_G(v_1)\setminus\{v'_1\} = N_{G'}(v_1)\setminus\{v'_1\}$ and $N_G(v'_1)\setminus\{v_1\} = N_{G'}(v'_1)\setminus\{v_1\}$. This implies that $\pi_G^i = \pi_{G'}^i$ for every $i \in \N$, so $G'$ is a long-refinement graph if and only if $G$ is one. By double counting edges, $G$ (and hence also $G'$) has an even number of vertices and hence $G'$ is one of the graphs from the first list in Theorem \ref{thm:classification23-intro}. The only graph from that list with exactly two vertices of degree $2$ is $G(\mathrm{S011XX})$, but these vertices are not adjacent. Hence, $G'$ does not exist, and therefore $G$ does not exist.
\end{proof}

This concludes the proof of Theorem \ref{thm:classification3-intro}, which together with Theorem \ref{thm:deg-12-or-14-intro} yields that the graph in Figure \ref{fig:d=l+3l+4-graphs-intro} is the only long-refinement graph $G$ with $\min\deg (G) = 1$ and $\max\deg(G) \leq 4$. 

Figure \ref{fig:all-diagrams} lists all long-refinement graphs from our classification that are not part of infinite families, in the order in which they appear in our analysis in the remainder of this paper. Tables \ref{tab:adj-list-d=l-34-intro} -- \ref{tab:adj-list3-d=l-34-intro} list the infinite families.

\begin{figure}[htpb]
    \centering
    \begin{subfigure}{0.24\linewidth}
            \centering 
            \includegraphics[scale=0.8]{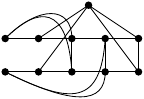}
         \caption{}\label{fig:d=l-deg34-11-bigfig}
    \end{subfigure}
    \begin{subfigure}{0.24\linewidth}
            \centering 
            \includegraphics[scale=0.8]{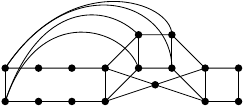}
         \caption{}\label{fig:d=l-deg24-17-bigfig}
    \end{subfigure}
    \begin{subfigure}{0.24\linewidth}
            \centering 
            \includegraphics[scale=0.8]{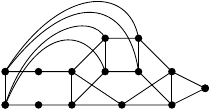}
         \caption{}\label{fig:d=l-deg24-14-bigfig}
    \end{subfigure} 
    \begin{subfigure}{0.19\linewidth}
            \centering 
            \includegraphics[scale=0.8]{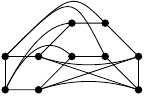}
         \caption{}\label{fig:d=l-deg34-10-bigfig}
    \end{subfigure}
    \begin{subfigure}{0.19\linewidth}
            \centering 
            \includegraphics[scale=0.8]{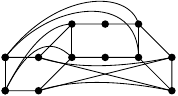}
         \caption{}\label{fig:d=l-deg24-12-bigfig}
    \end{subfigure}
    \begin{subfigure}{0.19\linewidth}
            \centering 
            \includegraphics[scale=0.8]{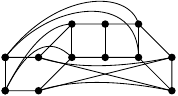}
         \caption{}\label{fig:d=l-deg34-12-bigfig}
    \end{subfigure}
    \begin{subfigure}{0.19\linewidth}
            \centering 
            \includegraphics[scale=0.8]{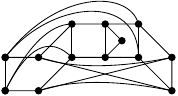}
         \caption{}\label{fig:d=l-deg34-13A-bigfig}
    \end{subfigure}
    \begin{subfigure}{0.19\linewidth}
            \centering 
            \includegraphics[scale=0.8]{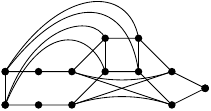}
         \caption{}\label{fig:d=l-deg24-13-bigfig}
    \end{subfigure}
    \begin{subfigure}{0.24\linewidth}
            \centering 
            \includegraphics[scale=0.8]{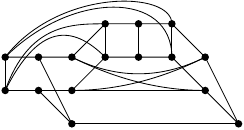}
         \caption{}\label{fig:d=l-deg34-16-bigfig}
    \end{subfigure}
    \begin{subfigure}{0.24\linewidth}
            \centering 
            \includegraphics[scale=0.8]{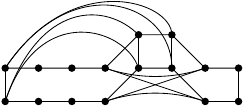}
         \caption{}\label{fig:d=l-deg24-16-bigfig}
    \end{subfigure}
    \begin{subfigure}{0.33\linewidth}
            \centering 
            \includegraphics[scale=0.8]{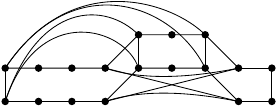}
         \caption{}\label{fig:d=l-deg24-18-bigfig}
    \end{subfigure}
    \begin{subfigure}{0.32\linewidth}
            \centering 
            \includegraphics[scale=0.8]{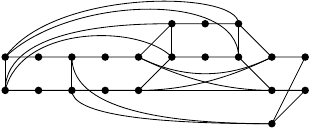}
         \caption{}\label{fig:d=l-deg24-21-bigfig}
    \end{subfigure}
    \begin{subfigure}{0.32\linewidth}
            \centering 
            \includegraphics[scale=0.8]{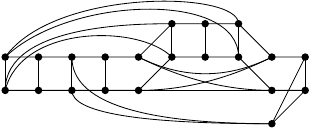}
         \caption{}\label{fig:d=l-deg34-21-bigfig}
    \end{subfigure}
    \begin{subfigure}{0.32\linewidth}
            \centering 
            \includegraphics[scale=0.8]{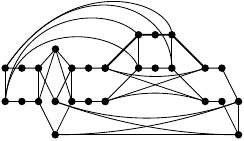}
         \caption{}\label{fig:d=l-deg24-27-bigfig}
    \end{subfigure}
    \begin{subfigure}{0.3\linewidth}
            \centering 
            \includegraphics[scale=0.8]{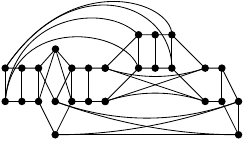}
            \caption{}\label{fig:d=l-deg34-27-bigfig}
    \end{subfigure}
    \begin{subfigure}{0.2\linewidth}
        \centering \includegraphics[scale=0.8]{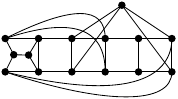}
        \caption{}
        \label{fig:d=l+2-deg34-15-bigfig}
    \end{subfigure}
    \begin{subfigure}{0.2\linewidth}
        \centering
        \includegraphics[scale=0.8]{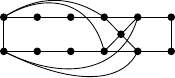}
        \caption{}\label{fig:d=l+1-singleton-adjXX-1-bigfig}
    \end{subfigure}
    \begin{subfigure}{0.2\linewidth}
        \centering
        \includegraphics[scale=0.8]{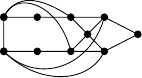}
        \caption{}\label{fig:d=l+1-singleton-adjXX-2-bigfig}
    \end{subfigure}
\begin{subfigure}{0.25\linewidth}
    \centering
    \includegraphics[scale=0.8]{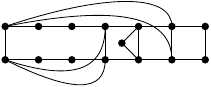}
    \caption{}\label{fig:deg24-14a-bigfig}
\end{subfigure}
\begin{subfigure}{0.32\linewidth}
    \centering
    \includegraphics[scale=0.8]{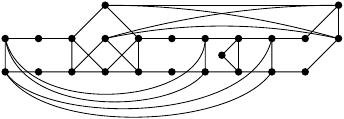}
    \caption{}\label{fig:deg24-24a-bigfig}
\end{subfigure}
\begin{subfigure}{0.32\linewidth}
    \centering
    \includegraphics[scale=0.8]{d=l+3l+4-graphs-4.pdf}
    \caption{}
\end{subfigure}
        \caption{The long-refinement graphs $G$ with $\deg(G) \neq \{2,3\}$, omitting the infinite families from Tables \ref{tab:adj-list-d=l-34} -- \ref{tab:adj-list3-d=l-34}.}
    \label{fig:all-diagrams}
\end{figure}

\begin{table}[htpb]
    \centering
\begin{tabular}{c|c|c}
    vertex & adjacency 1 with $n \in \N$ & adjacency 2 with $n \in \N_0$\\ \hline
    0 & $2n+1, 2n+2, 2n+5, 2n+6 $ & $2n+1, 2n+2, 2n+5, 2n+6 $\\
    1 & $3,4n+7, 4n+8$ & $3,4n+7, 4n+8$\\
    2 & $4, 4n+9, 4n+10$&  $4, 4n+9, 4n+10$\\
    odd  $i \in  I$& $i-2,i+1,i+2$& $i-2,i+1,i+2$\\
    even $i \in I$& $i-2,i-1, i+2$& $i-2,i-1, i+2$\\
    $2n+1^*$& $0, 2n-1, 2n+3$& $2n+1, 2n+2, 2n+5$\\
    $2n+2^*$& $0, 2n, 2n+4$& $2n, 2n+1 2n+4$\\
    $2n+3$& $2n+1, 2n+5, 6n+11, 6n+12$&$0, 2n+1, 2n+4, 2n+5$\\
    $2n+4$& $2n+4, 2n+8, 6n+17, 6n+18$&$0, 2n+2, 2n+3, 2n+6$\\
    $2n+5$& $0, 2n+3, 2n+7$&$2n+3,2n+6,2n+7$\\
    $2n+6$& $0, 2n+4, 2n+8$&$2n+4,2n+5,2n+8$\\
    $4n+7$& $1, 4n+5, 4n+9$&$1, 4n+5, 4n+9$\\
    $4n+8$&  $1, 4n+6, 4n+10$& $1, 4n+6, 4n+10$\\
    $4n+9$& $2,4n+7, 4n+11$&$2,4n+7, 4n+11$\\
    $4n+10$& $2, 4n+8, 4n+12$&$2, 4n+8, 4n+12$\\
    $6n+11$& $2n+3, 2n+4, 6n+9$&$0,6n+12,6n+10$\\
    $6n+12$&$2n+3, 2n+4, 6n+10$&$0,6n+11,6n+10$\\
\end{tabular}
    \caption{The two infinite families of long-refinement graphs $G$ with order $6n+13$ for $\deg(G) = \{3,4\}$ where $I= [3,2n]\cup[2n+7,4n+6] \cup [4n+11, 6n+10]$. Note also that $2n+1, 2n+2$ are omitted when $n=0$.} \label{tab:adj-list-d=l-34-intro}
\end{table}
\vspace{2ex}

\begin{table}[htpb]
    \centering
\begin{tabular}{c|c|c}
    vertex & adjacency 1 & adjacency 2 \\ \hline
    0 & $2k+3, 2k+4, 2k+7, 2k+8 $&$2k+5,2k+6,6k+13,6k+14$\\
    1 & $3,4k+9, 4k+10$&$3,4k+9,4k+10$ \\
    2 & $4, 4k+11, 4k+12$&$4, 4k+11, 4k+12$\\
    odd  $i \in I$ & $i-2,i+1,i+2$& $i-2,i+1,i+2$\\
    even $i \in I$ & $i-2,i-1, i+2$& $i-2,i-1, i+2$\\
    $2k+3$ & $0, 2k+1, 2k+5$&$2k+1,2k+4,2k+5$\\
    $2k+4$ & $0, 2k+2, 2k+6$&$2k+2,2k+3,2k+6$\\
    $2k+5$& $2k+3, 2k+7, 6k+13, 6k+14$&$0,2k+3,2k+6,2k+7$\\
    $2k+6$& $2k+4, 2k+8, 6k+13, 6k+14$&$0,2k+4,2k+5,2k+8$\\
    $2k+7$& $0, 2k+5, 2k+9$&$2k+5,2k+8,2k+9$\\
    $2k+8$& $0, 2k+6, 2k+10$&$2k+6,2k+7,2k+10$\\
    $4k+9$& $1, 4k+7, 4k+11$&$1, 4k+7, 4k+11$\\
    $4k+10$&  $1, 4k+8, 4k+12$&$1, 4k+8, 4k+12$\\
    $4k+11$& $2,4k+9, 4k+13$&$2,4k+9, 4k+13$\\
    $4k+12$& $2, 4k+10, 4k+14$&$2, 4k+10, 4k+14$\\
    $6k+13$& $2k+5, 2k+6, 6k+11$&$0, 6k+14, 6k+11$\\
    $6k+14$&$2k+5, 2k+6, 6k+12$&$0, 6k+13, 6k+12$\\
\end{tabular}
    \caption{All infinite families of long-refinement graphs $G$ with order $6k+15$ and $\deg(G) = \{3,4\}$, where $I= [3,2k+2]\cup[2k+9,4k+8] \cup [4k+13, 6k+12]$, defined for $k\in \N_0$.}\label{tab:adj-list2-d=l-34-intro}
\end{table}

\vspace{1ex}

\begin{table}[htpb]
    \centering
\begin{tabular}{c|c|c}
    vertex  &  adjacency 1  &  adjacency 2 \\ \hline
    0  &  $2k+7, 2k+8, 2k+11,2k+12$  & $2k+5,2k+6,6k+17,6k+18$\\
    1  &  $3,4k+13, 4k+14$  &  $3,4k+13, 4k+14$\\
    2  &  $4, 4k+15, 4k+16$  &  $4, 4k+15, 4k+16$\\
    odd  $i \in I$  &  $i-2,i+1,i+2$ &  $i-2,i+1,i+2$\\
    even $i \in I$  &  $i-2,i-1, i+2$ &  $i-2,i-1, i+2$\\
    $2k+7$  &  $0, 2k+5, 2k+9$ &  $2k+5,2k+8,2k+9$\\
    $2k+8$  &  $0, 2k+6, 2k+10$ & $2k+6,2k+7,2k+10$\\
    $2k+9$ &  $2k+7, 2k+11, 6k+17, 6k+18$ & $0, 2k+7, 2k+10, 2k+11$\\
    $2k+10$ &  $2k+8, 2k+12, 6k+17,6k+18$ & $0, 2k+8 ,2k+9 ,2k+12$\\
    $2k+11$ &  $0, 2k+9, 2k+13$ & $2k+9, 2k+12, 2k+13$\\
    $2k+12$ &  $0, 2k+10, 2k+14$ & $2k+10,2k+11,2k+14$\\
    $4k+13$ &  $1, 4k+11, 4k+15$ & $1, 4k+11, 4k+15$\\
    $4k+14$ &   $1, 4k+12, 4k+16$ & $1, 4k+12, 4k+16$\\
    $4k+15$ &  $2,4k+13, 4k+17$ & $2, 4k+13, 4k+17$\\
    $4k+16$ &  $2, 4k+14, 4k+18$ & $2, 4k+14, 4k+18$\\
    $6k+17$ &  $2k+9, 2k+10, 6k+15$ & $0, 6k+18, 6k+15$\\
    $6k+18$ & $2k+9, 2k+10, 6k+16$ & $0, 6k+17, 6k+16$\\
\end{tabular}
    \caption{All infinite families of long-refinement graphs $G$ with order $6k+19$ and $\deg(G) = \{3,4\}$, where $I= [3,2k+6]\cup[2k+13,4k+12] \cup [4k+17, 6k+16]$.}\label{tab:adj-list3-d=l-34-intro}
\end{table}

\newpage

\section{Preliminaries}

With $\N$ being the set of natural numbers, we set $\N_0 \coloneqq \N \cup \{0\}$, and for $a, b \in \N_0$, let $[a,b] \coloneqq \{n \in \N_0 \mid a \leq n \leq b\}$ and $[b] \coloneqq [1,b]$. Multisets are a generalisation of sets in which repeated elements are allowed, i.e.\ elements may occur more than once. To distinguish them from sets, we denote multisets using double braces
 $\doubleBrackets{,}$. Given a set $S$, a \emph{partition} of $S$ is a set of pairwise disjoint non-empty sets $\Pi$ such that $\bigcup_{M\in \Pi} M = S$. Given two partitions $\Pi,\Pi'$ of $S$, we say that $\Pi$ \emph{refines} $\Pi'$ and write $\Pi \preceq \Pi'$ if every element of $\Pi$ is a (not necessarily strict) subset of some element of $\Pi'$. We call a set with a single element a \emph{singleton} and a set with two elements a \emph{pair}. A partition that consists of singletons only is called \emph{discrete}.

In our setting, all graphs are undirected, finite, and simple. In particular, they have no self-loops and between every pair of vertices only at most one edge. For a graph $G$, we use $V(G)$ and $E(G)$ to denote the set of vertices and the set of edges of $G$, respectively. We drop parameters if they are clear from or irrelevant in the context and just use $V$ and $E$ or write $G=(V,E)$. The \emph{order} of $G$ is $|G| \coloneqq |V(G)|$. The \emph{neighbourhood} of a vertex $v \in V(G)$ is $N(v)  \coloneqq \{u \in V(G) \mid \{u,v\} \in E\}$ and we extend the notion to sets of vertices $S$ by setting $N(S) \coloneqq \bigcup_{v \in S} (N(v)) \setminus S$. We consider $v$ \emph{adjacent} to $S$ if $N(v)\cap S \neq \emptyset$ and say that $C'$ is adjacent to $C$ if $N(C)\cap C' \neq \emptyset$. 
The \emph{degree} of $v$ is $\deg(v) \coloneqq |N(v)|$ and we set $\deg(G) \coloneqq \{\deg(v) \mid v \in V(G)\}$. For $d \in \N_0$, the graph $G$ is \emph{$d$-regular} if $\deg(G) = \{d\}$, and \emph{regular} if it is $d$-regular for some $d$. A \emph{matching} is a $1$-regular graph. 

For non-empty $S, S' \subseteq V(G)$, the \emph{induced subgraph} $G[S]$ is the graph $(S,E)$ where $E = \{\{u,v\} \mid \{u,v\} \in E(G), u,v \in S\}$, and the subgraph $G[S,S']$ is the graph $(S\cup S',E')$ where $E' = \{\{u,v\} \mid \{u,v\} \in E(G), u \in S,v \in S'\}$. If $V(G) = S \cup S'$ and $S \cap S' = \emptyset$ and $E(G) \cap \{\{v,w\} \mid v,w \in S\} = \emptyset = E(G) \cap \{\{v,w\} \mid v,w \in S'\}$, then $G$ is \emph{bipartite (on bipartition $(S,S')$)}. If, additionally, $\big\{\{v,w\} \mid v \in S, w \in S'\big\} = E(G[S,S'])$, the graph $G$ is \emph{complete bipartite}.

A \emph{path} (of length $\ell$) from $u \in V(G)$ to $v \in V(G)$ is a sequence $u = v_0, v_1 \dots, v_\ell = v$ such that $\{v_{i-1},v_{i}\} \in E(G)$ and $v_i \neq v_j$ for all $i \in [\ell]$ and all $j \in [0,\ell]$ with $j \notin \{i,\ell\}$. The \emph{distance} $d(u,v)$ between $u$ and $v$ is the length of a shortest path between $u$ and $v$. 
	A \emph{vertex colouring} of $G$ is a mapping $\lambda$ from $V(G)$ into some set $\curlC$, whose elements we call \emph{colours}. We call the pair $(G,\lambda)$ a \emph{vertex-coloured graph} and denote by $\pi_G(\lambda)$ the partition of $V(G)$ into the \emph{colour classes}, i.e.\ into the maximal sets of vertices that have the same colour with respect to $\lambda$.

	An \emph{isomorphism} between graphs $G$ and $H$ is a bijective mapping $\varphi \colon V(G) \rightarrow V(H)$ such that for all $u,v \in V(G)$, it holds that $\{u,v\} \in E(G)$ if and only if $\{\varphi(u),\varphi(v)\} \in E(H)$. Isomorphisms between vertex-coloured graphs $(G,\chi_G)$ and $(H,\chi_H)$ must additionally satisfy $\chi_G(u) = \chi_H(\varphi(u))$.

\section{Colour Refinement}\label{sec:colref}

In this section, we give a brief introduction to the Colour Refinement algorithm, which is also called the \emph{$1$-dimensional Weisfeiler--Leman algorithm}, and state the main facts about it that we will use in our further analysis. Given a (vertex-coloured or monochromatic) input graph, Colour Refinement uses a simple local criterion to iteratively refine the partition of the vertices induced by the current colouring. As soon as a stable state with respect to the criterion is reached, the algorithm returns the graph together with the current vertex colouring.

\begin{definition}[Colour Refinement]
	Let $G$ be a graph with vertex colouring $\lambda \colon V(G)\rightarrow\curlC$. 
	On input $(G,\lambda)$, the Colour Refinement algorithm computes its output recursively with the following update rules. Let $\chi^0 \coloneqq \lambda$ and for $i \in \N$, the \emph{colouring after $i$ iterations of Colour Refinement} is
		\[\chi^i(v) \coloneqq (\chi^{i-1}(v),\doubleBrackets{\chi^{i-1}(u) | \{u,v\} \in E(G)}).\]
\end{definition}

Thus, two vertices obtain different colours in iteration $i$ precisely if they already had different colours in iteration $i-1$ or if they differ in the $\chi^{i-1}$-colour multiplicities among their neighbours in iteration $i-1$.
Letting $\pi^i \coloneqq \pi_{G,\lambda}(\chi^i)$ be the vertex partition induced by the colouring after $i$ iterations of Colour Refinement, it is immediate that $\pi^{i+1} \preceq \pi^{i}$ holds for every $i \in \N$ and that, hence, there is a minimal $j \leq |G|-1$ such that $\pi^{j+1} = \pi^{j}$. With this choice of $j$, the \emph{output} of Colour Refinement on input $(G,\lambda)$ is $\chi_{G,\lambda} \coloneqq \chi^j$. We call $j$ the \emph{iteration number} of Colour Refinement on input $(G,\lambda)$ and denote it by $\WL_1(G,\lambda)$. If $\lambda$ induces the trivial partition, i.e.\ $(G,\lambda)$ is monochromatic, we set $\WL_1(G) \coloneqq \WL_1(G,\lambda)$. We call $(G,\lambda)$ and the induced partition $\pi(\lambda)$ \emph{stable} if $\WL_1(G,\lambda) = 0$, i.e.\ if $\pi(\lambda) = \pi(\chi_{G\lambda})$. It follows from the definition of the Colour Refinement algorithm that for all $P, Q \in \pi(\chi_{G,\lambda})$ with $P \neq Q$, the graph $G[P]$ is regular and $H\coloneqq G[P,Q]$ is \emph{biregular}, i.e.\ there are $k,\ell \in \N$ such that for every $v \in P$, it holds that $|N_H(v)| = k$, and for every $w \in Q$, it holds that $|N_H(w)| = \ell$.

A \emph{long-refinement graph} is a graph $G$ for which $\WL_1(G) = |G|-1$, i.e.\ on which Colour Refinement takes the maximum possible number of iterations to compute its output. Note that, in our quest for long-refinement graphs, we do not have to consider vertex-coloured graphs, because, as observed in \cite{KMcK20}, all long-refinement graphs $G$ of order at least $2$ are monochromatic. Moreover, they are connected, have $|\deg(G)| = 2$, and Colour Refinement produces the discrete partition on them.

The only connected graphs with degrees $1$ and $2$ are paths, and a monochromatic path on $n$ vertices takes $\lfloor\frac{n-1}{2}\rfloor$ iterations of Colour Refinement to stabilise (see, e.g., \cite[Fact 5]{KMcK20}). Solving $\lfloor\frac{n-1}{2}\rfloor = n - 1$ yields $n \leq 1$, but the graph consisting of a singleton vertex does not satisfy the degree restrictions. We hence obtain the following observation.

\begin{observation}[cf.\ \cite{KMcK20}]\label{obs:maxdeg2}
    There is no long-refinement graph $G$ with $\deg(G)=\{1,2\}$.
\end{observation}

Therefore, the smallest degree pairs with which there may be long-refinement graphs are $\{1,3\}$ and $\{2,3\}$. For both, examples are given in \cite{KMcK20}. More precisely, they present one long-refinement graph $G$ with $\deg(G) = \{1,3\}$, one with $\deg(G) = \{1,5\}$, and infinitely many with $\deg(G) = \{2,3\}$. For specific ones of the infinite families with $\deg(G) = \{2,3\}$, the proof of their Lemma 27 yields a simple construction to show the following. 

\begin{observation}
    There are infinitely many sets $S$ with $3 \in S$ and such that $\deg(G) = S$ for some long-refinement graph $G$.
\end{observation}

\begin{proof}
To see this, one just needs to verify that there are infinitely many numbers $3 \neq s \in \N$ such that there are long-refinement graphs $G$ with $\deg(G)=\{2,3\}$ and $s = |\{v \in V(G) \mid \deg(v) = 2\}|$, for which one can consider, for example, the last infinite family in \cite[Theorem 24]{KMcK20}. Now inserting a fresh vertex $w$ with edges to all $v \in V(G)$ with $\deg(v) = 2$ yields a long-refinement graph $G'$ with $\deg(G') = \{3,s\}$.
\end{proof}


In the remainder of this section, we develop Lemma \ref{rmk:work-backwards}, which is the main structural insight about the partition classes in successive Colour Refinement iterations and which will allow us to pursue the following reverse-engineering approach. To study long-refinement graphs $G$ with $|G|>1$, we intend to construct the partition $\pi^{i-1}$ from the partition $\pi^i$. Each partition $\pi^i$ must have exactly one more colour class than $\pi^{i-1}$, see \cite[Proposition 8]{KMcK20}. So there must be disjoint $A_i$, $B_i$ such that $\{A_i, B_i\}  = \pi^i \setminus \pi^{i-1}$ and whose union is the unique $C_{i-1}$ with $\{C_{i-1}\} = \pi^{i-1}\backslash\pi^{i}$. We say that \emph{$C_{i-1}$ splits into $A_i$ and $B_i$} in iteration $i$. The class $C_i$ splits immediately after and as a direct consequence of $C_{i-1}$ splitting, so all vertices in $C_i$ must have the same number of neighbours in $C_{i-1}$, but differ in their numbers of neighbours in $A_{i}$ or in $B_i$. We formalise this by defining the \emph{degree of vertex $v$ with respect to set} $C \subseteq V(G)$ as $\deg_C(v) \coloneqq |N(v)\cap C|$ and, for vertex sets $C,C'$ and $v \in C$, we define the \emph{degree of $C$ into $C'$} as follows:
	\begin{equation*}\deg_{C'}(C) \coloneqq \left\{
	\begin{array}{ll}
		\deg_{C'}(v),\  & \text{if $\deg_{C'}(v) = \deg_{C'}(u)$ holds for every $u \in C$} \\
		\bot,\  & \mbox{otherwise }
	\end{array}
	\right.
	\end{equation*}
If $C' = V(G)$, we let $\deg(C) \coloneqq \deg_{C'}(C)$. If $\deg_{C'}(C) = \bot$, then $C$ is \emph{unbalanced wrt} $C'$. Otherwise, $C$ is \emph{balanced wrt }$C'$. If we do not want to specify the class $C'$ but just claim its existence, we call $C$ simply \emph{unbalanced}. With these notions, notice that every class in $\pi^i$ with $i>0$ is balanced wrt every class in $\pi^{i-1}$. In particular, $C_i$ is balanced wrt $C_{i-1}$, and also with respect to each class in $\pi^i\cap \pi^{i-1}$. Since $C_i$ is split in iteration $i+1$, $C_{i}$ cannot also be balanced wrt all classes in $\pi^i$, hence there must be a class in $\pi^i \setminus \pi^{i-1} = \{A_i,B_i\}$ wrt which $C_i$ is unbalanced (cf.\ \cite[Lemma 11]{KMcK20}). In fact, since $C_i$ is balanced wrt  $C_{i-1}$, i.e.\ $A_i\cup B_i$, all vertices in $C_i$ have the same degree into $A_i \cup B_i$, and therefore $\deg_{A_i}(C_i) = \bot$ holds if and only if $\deg_{B_i}(C_i) = \bot$. We summarise all these observations in the following lemma.

\begin{lemma}\label{rmk:work-backwards}
	Let $G$ be a long-refinement graph with $|G| > 1$. Then in any non-stable partition $\pi^i$ on $V(G)$, there is exactly one unbalanced class $C_i \in \pi^i$. 
    
    For the following, assume that, $i > 0$, that is, $\pi^i \neq \pi^0$. Then there are exactly two distinct options $A_i$, $B_i$ for $C' \in \pi^i$ with $\deg_{C'}C_i = \bot$. Hence,
\begin{equation*}
    	\deg_{A_i}(C_i) = \bot = \deg_{B_i}(C_i).
\end{equation*}	
Furthermore, $C_i = A_{i+1} \cup B_{i+1}$ and
    \begin{equation}\label{eq:equal-diff}
        	\deg_{A_{i}}(A_{i+1}) + \deg_{B_{i}}(A_{i+1}) = \deg_{C_{i-1}}(C_{i}) =  \deg_{A_{i}}(B_{i+1}) + \deg_{B_{i}}(B_{i+1}) \in [1,|C_{i-1}|-1].
    \end{equation}
    In particular, it holds that $\deg_{C_{i-1}}(C_{i}) \neq \bot$.
    
    Finally, if $i>1$, then $C_i \neq V(G)$ and, since the first iteration of Colour Refinement separates vertices by degree, $\deg(C_i) \neq \bot$ and thus $\deg(A_i) = \deg(B_i)$.
\end{lemma}

As a first application of our formal definition of classes being unbalanced wrt each other, we note that the class of long-refinement graphs is closed under edge complementation.

\begin{lemma}\label{lem:complement}
    A graph $G = (V,E)$ is a long-refinement graph if and only if the graph $G' = (V,E')$ is a long-refinement graph, where $E' = \{ \{u,v\} \mid u,v \in V, \{u,v\} \notin E, u \neq v\}$.
\end{lemma}

\begin{proof}
    It suffices to show that $\pi^i_G = \pi^i_{G'}$ holds for every $i \in \N_0$. First note that for vertex sets $C, C' \subseteq V$ that are either equal or disjoint, we have that $\deg_{C'}(C) = \bot$ holds in $G$ if and only if it holds in $G'$. Indeed, if $C = C'$, then $G[C]$ is a subgraph of the complete graph on $C$. The latter is $(|C|-1)$-regular, hence $G[C]$ is regular if and only if the ``complement subgraph'' is regular. Similarly, if $C \cap C' = \emptyset$, the graph $G[C,C']$ is a subgraph of the complete bipartite graph on bipartition $(C,C')$. The latter is biregular, meaning it satisfies $\deg_{C'}(C) \in \N \ni \deg_{C}(C')$. Hence, $\deg_{C'}(C) \neq \bot$ holds in $G[C,C']$ if and only if it holds in the ``complement subgraph'', and the same is true for the condition $\deg_{C}(C') \neq \bot$.
    
    We can now proceed by induction on $i$.  It is clear that $\pi^0_G = \pi^0_{G'}$. For the inductive step, suppose that $\pi^i_G = \pi^i_{G'}$ holds for some $i \geq 1$. It suffices to see that, up to possibly swapping the roles of $A_i$ and $B_i$, the classes $A_i$, $B_i$, and $C_i$ from Lemma \ref{rmk:work-backwards} are the same in both graphs.
    Indeed, the graph $G$ satisfies Lemma \ref{rmk:work-backwards}, so there is a unique unbalanced class $C_i \in \pi^i$ that splits into $A_{i+1}$ and $B_{i+1}$ in iteration $i+1$. By the observation above, $C_i$ is also unbalanced in $G'$, precisely to $A_i$ and to $B_i$, and there is no other unbalanced class in $G'$ in $\pi^i$ (as otherwise the same would hold in the long-refinement graph $G$, contradicting Lemma \ref{rmk:work-backwards}).
\end{proof}

\section{Reverse-Engineering Long-Refinement Graphs}\label{sec:rev-eng}

We now collect structural insights about long-refinement graphs with \emph{arbitrary} degrees. We first treat the last iterations of Colour Refinement to deduce the connections between the partition classes in those iterations where we are guaranteed that all classes have size at most $2$. We then reverse-engineer the previous, larger classes and their connections by analysing unions of the pairs and singletons. The structural insights allow us to deduce helpful facts about the arrangement of pairs and singletons, yielding necessary conditions for a graph to qualify as a long-refinement graph. We will use those in the Section \ref{sec:max4} to show that long-refinement graphs with small degrees can be represented as strings and that the only strings whose graphs satisfy the necessary conditions are the ones found and presented in \cite{KMcK20}. We also extend the classification to maximum degree $4$ there.

Let $G$ be a long-refinement graph of order $|G|>1$, and let $p$ be the minimal index of an iteration of Colour Refinement in which every element in the partition $\pi^p$ has size at most $2$, i.e.\ $\pi^p$ contains only pairs and singletons. Note that $p < \infty$ since Colour Refinement must produce the discrete partition on $G$. Let $\pi^p_{\mathcal{S}} \coloneqq \{M \in \pi^p \mid |M| = 1\}$ be the singletons in $\pi^p$, and $\pi^p_{\mathcal{P}} \coloneqq \{M \in \pi^p \mid  |M| = 2\}$ be the pairs in $\pi^p$. Note that for all $i\geq p$, the class $C_i$ that splits in iteration $i+1$ must be a pair $P \in \pi^p_{\mathcal{P}}$. This gives us a linear order $\prec$ on $\pi^p_\mathcal{P}$ by the iterations in which pairs split, and we denote the $i$-th element of $\prec$, the pair $C_{p+i-1}$, by $P_i$. Let $\alpha(P_i) \coloneqq P_{i-1}$ and $\beta(P_i) \coloneqq P_{i+1}$, respectively (for $i = \min(\prec)$, we let $\alpha(P_i) \coloneqq \bot$, and for $i = \max(\prec)$, we let $\beta(P_i) \coloneqq \bot$). We also call $P,P'$ \emph{consecutive pairs} if $P' = \beta(P)$. We let $n_\mathcal{P} \coloneqq|\pi^p_\mathcal{P}|$ be the number of pairs.

We now analyse the structure of long-refinement graphs, in particular how the elements in $\pi^p_\mathcal{P}$ are connected. In the remainder of this paper, by a \emph{pair}, we mean an element of $\pi^p_\mathcal{P}$. As the following lemma states, consecutive pairs are always connected via a matching. The proof is identical to the proof of the same statement for graphs $G$ with $\deg(G) = \{2,3\}$, which is \cite[Corollary 15]{KMcK20}.

\begin{lemma}\label{lem:successive-matching}
For all $i \in [n_\mathcal{P}-1]$, $\deg_{P_i}(P_{i+1})=\deg_{P_{i+1}}(P_{i})=1$.
\end{lemma}

Recall that $C_p$ is the class that splits in iteration $p+1$ and therefore is equal to $P_1$, and that by definition, we have $A_p, B_p \in \pi^p$. By the choice of $p$, both $A_p$ and $B_p$ have cardinality at most $2$. 

By Lemma \ref{rmk:work-backwards}, we know $\deg_{C_{p-1}}(P_1) \neq \bot$ and $\deg_{A_p}(P_1) = \bot =\deg_{B_p}(P_1)$. Furthermore, a pair $P$ is always balanced wrt itself, since $\deg_{P}(P) \in \{0,1\}$. We obtain the following proposition.

\begin{proposition}\label{prop:p1neqakbk}
$P_1 \notin \{A_p,B_p\}$.
\end{proposition}

We now show that $A_p$ and $B_p$ are in fact pairs.

\begin{lemma}\label{lem:C_p-1}
	$A_{p},B_{p} \in \pi^p_\mathcal{P}$.
\end{lemma}

\begin{proof}
If $|A_p| = |B_p| = 1$, then $|C_{p-1}| = |A_p \cup B_p| = 2$, contradicting the minimality of $p$.

Let us next assume that $|A_p \cup B_p| = 3$, i.e.\ they are a singleton and a pair. Without loss of generality, let $A_p = \{s\}$ be the singleton. By Lemma \ref{rmk:work-backwards}, we know $\deg_{C_{p-1}}(P_1) \neq \bot$ and $\deg_{A_p}(P_1) = \bot =\deg_{B_p}(P_1)$. The latter yields that $G[A_p,P_1]$ contains a single edge. With Proposition \ref{prop:p1neqakbk} and since $P_1$ is the only class that splits in iteration ${k+1}$, Lemma \ref{rmk:work-backwards} gives us that $B_p$ is balanced wrt $P_1$, i.e.\ $\deg_{P_1}(B_p) \neq \bot$. Together with $\deg_{B_p}(P_1) = \bot$, this yields that the two vertices in $B_p$ have the same (unique) neighbour in $P_1$. But then $\deg_{C_{p-1}}(P_1) = \bot$, a contradiction.

Therefore, the only remaining possibility is that $|A_p| = |B_p| = 2$, which concludes the proof.
\end{proof}

 Let $a$ and $b$ be the positions of $A_p$ and $B_p$, respectively, in $\prec$. That is, $A_p = P_a$ and $B_p = P_b$. By Proposition \ref{prop:p1neqakbk}, we have $a > 1$ and $b > 1$. We can assume $b>a$ by symmetry. By Lemma \ref{lem:successive-matching}, $\deg_{P_2}(P_1) = 1$ and $\deg_{A_p}(P_1) = \bot = \deg_{B_p}(P_1)$, and therefore $P_2 \notin \{P_a,P_b\}$. Thus, with Proposition \ref{prop:p1neqakbk}, we obtain that $n_\mathcal{P} \geq 4$.  

 The following lemma describes how the first pair in $\prec$ is connected to the last classes that cause its splitting in iteration $p$. It generalises the same statement from degrees $2$ and $3$ to arbitrary degrees. 
  
\begin{lemma}[cf.\ \cite{KMcK20}, Lemma 16]\label{lem:P_1-P_aP_b}
    There exists $u \in P_1$ such that $G[\{u\},P_a]$ and $G[P_1\backslash\{u\},P_b]$ are complete bipartite, and $E(G[\{u\},P_b])=E(G[P_1\backslash\{u\},P_a])=\emptyset$.
\end{lemma}

\begin{proof}
	Recall again, by Lemma \ref{rmk:work-backwards}, that $\deg_{A_p}(P_1) = \bot =\deg_{B_p}(P_1)$. In particular, there is a vertex $u \in P_1$ such that $N(u) \cap P_a \neq \emptyset$. 
    Furthermore, by Proposition \ref{prop:p1neqakbk} and again with Lemma \ref{rmk:work-backwards}, we must have $\deg_{P_1}(P_a) \neq \bot \neq \deg_{P_1}(P_b)$. The first inequality lets us conclude that $G[\{u\},P_a]$ is complete bipartite and that $E(G[P_1\backslash\{u\},P_a])=\emptyset$. With $\deg_{C_{p-1}}(P_1) \neq \bot$, we obtain the claimed statements about $G[P_1,P_b]$.
\end{proof}

We already know that consecutive pairs are connected via a matching. 
The following lemma says that non-consecutive pairs (except for the case of $P_a$ and $P_b$) as well as a pair and a singleton form either an empty or a complete bipartite graph, i.e.\ their connections are trivial.

\begin{lemma}\label{lem:non-successive-pairs}
Let $P_i$ and $P_j$ be pairs with $i+1<j$, and, if $i=1$,  then $j \notin \{a,b\}$. Let $S \in \pi^p_\mathcal{S}$.
Then $\deg_{S}(P_i) \in \{0,1\} \text { and } deg_{P_i}(P_j)=\deg_{P_j}(P_i)= \deg_{P_i}(S) \in \{0,2\}$.
\end{lemma}
\begin{proof}
Note that Lemma \ref{lem:C_p-1} says that $C_{p-1}$ splits into two pairs and hence does not introduce new singletons in iteration $p$, i.e.\ $\pi^p_\mathcal{S} \subseteq \pi^{p-1}$. Also $\pi^p_\mathcal{P} \subseteq \pi^p$, so $S \in \pi^{p-1}$ and $P_i \in \pi^p$. Thus, $\bot \neq \deg_{S}(P_i) \in \{0,1\}$ or, equivalently, $\deg_{P_i}(S)=\{0,2\}$.

For the second part, choose $u$ and $v$ such that $\{u,v\} = P_i$ and $A_{p+i}=\{u\}$, and $B_{p+i} = \{v\}$. Naturally, $\{A_{p+i}, B_{p+i}\} \subseteq \pi^{p+i}$ and, since $j > i+1$, $P_j  \in \pi^{p+i+1}$ . Therefore, $\deg_{A_{p+i}}(P_j) \neq \bot \neq \deg_{B_{p+i}}(P_j)$. By the fact that $|A_{p+i}| = |B_{p+i}|=1$, we have that each of $\deg_{A_{p+i}}(P_j)$ and $ \deg_{B_{p+i}}(P_j)$ is either $0$ or $1$. Now, we show that $\deg_{A_{p+i}}(P_j) = \deg_{B_{p+i}}(P_j)$. If $i=1$, then by assumption $P_j \notin \{P_a,P_b\}$ and thus $P_j \in \pi^{p-1}$, while $P_i \in \pi^p$. If $i\neq 1$, then $P_j \in \pi^p$ and $P_i \in \pi^{p+1}$. In either case, since $P_j$ splits later than $P_i$, this gives us $\deg_{P_{i}}(P_j) \neq \bot$ and therefore $\deg_{A_{p+i}}(P_j) = \deg_{B_{p+i}}(P_j) \in \{0,1\}$. This tells us that either $G[P_i,P_j]$ is complete bipartite or $E(G[P_i,P_j]) = \emptyset$. 
\end{proof}

The following lemma states that the distance in $\prec$ between $P_a$ and $P_b$, i.e.\ between the classes that induce the splitting of $P_1$, is at most $2$.

 \begin{lemma}\label{lem:bdista}
     $b \in \{a+1,a+2\}$.
 \end{lemma}
 \begin{proof}
Assume that $b > a + 2$. Then consider the set of distinct pairs $\Pi \coloneqq \{P_{a-1}, P_{a+1}, P_{b-1}\}$.
 For each $P \in \Pi$, it holds that exactly one set of $\{P_a,P\}$ and $\{P_b,P\}$ consists of consecutive pairs. Therefore, by Lemmas \ref{lem:successive-matching} and \ref{lem:non-successive-pairs}, $C_{p-1} = P_{a } \cup P_b$ is unbalanced wrt $P$. 
Thus, the first iteration $i$ for which any pair $P \in \Pi$ is in $\pi^i$ must directly precede the iteration $p$ in which $C_{p-1}$ splits. We have that $|\pi^i \setminus \pi^{i-1}| = 2$ and $(\pi^p \setminus \pi^{p-1}) \cap \Pi = \emptyset$, so there must be a $P \in \Pi$ that is not in $\pi^p$, a contradiction. Thus, $a < b \leq a + 2$.
 \end{proof}

In the previous lemmas, we have analysed the splits that happen from the first moment on where the largest class has size $2$. Our structural insights have revealed the relations between the pairs in $\pi^p$. Now we will work back in time from $\pi^p$, i.e.\ starting from $\pi^p$, we reverse-construct the partitions $\pi^{p-i}$ for increasing $i$. Recall that $a$ and $b$ are the positions of $A_p$ and $B_p$ in the splitting order $\prec$ of the pairs, whose collection is $\pi^p_\mathcal{P}$ with size $n_\mathcal{P} = |\pi^p_{\mathcal{P}}|$, and that $\pi^p_\mathcal{S}$ is the collection of singletons in $\pi^p$. Define $\ell_{4 \rightarrow 2} \coloneqq \min\{n_\mathcal{P}-b,a-1\}$. Then the technical statement in the lemma below expresses that each of the $\ell_{4 \rightarrow 2}$ iterations preceding $\pi^p$ splits one class of four vertices into two pairs, which will be elements of $\pi^p$. 

\begin{lemma}\label{lem:first-symmetry}
    For every $h \in [0,\ell_{4 \rightarrow 2}]$, the following holds:
	\begin{multline*}
		\pi^{p-h-1} =
		\pi^p_{\mathcal{S}} \, 
		\cup \, 
		\{P_{a-i}\, \cup \, P_{b+i} \mid i \in [0,h]\} \, \cup \,
		\left\{P_{i} \mid i \in [1,a-h-1] \, \cup \, [a+1,b-1]\cup [b+h+1, n_\mathcal{P}] \right\}.
	\end{multline*}
    \end{lemma}

\begin{proof}
    We prove the statement by induction on $h$. The case for $h=0$ follows from Lemma \ref{lem:C_p-1} and the fact that $\pi^{p-1}=\pi^p\backslash\{A_p,B_p\}\cup \{A_p \cup B_p\}$. This is illustrated in Figure \ref{h=0}. 
    
    \begin{figure}[htpb]
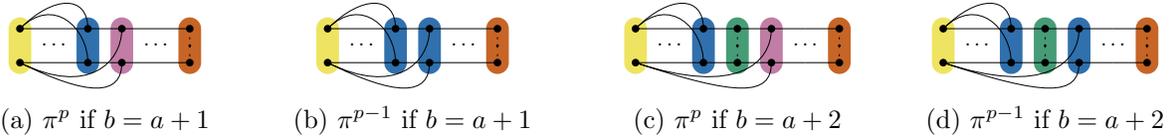

    \centering
    \begin{subfigure}[htpb]{0.24\textwidth}
    	\centering
    	\includegraphics[scale=0.8]{sym1-1.pdf}
    	\caption{$\pi^p$ if $b=a+1$}
    \end{subfigure}
    \begin{subfigure}[htpb]{0.24\textwidth}
    	\centering
    	\includegraphics[scale=0.8]{sym1-2.pdf}
    	\caption{$\pi^{p-1}$ if $b=a+1$}
    \end{subfigure}
    \hspace{0.1cm}
    \begin{subfigure}[htpb]{0.24\textwidth}
    	\centering
    	\includegraphics[scale=0.8]{sym1-5.pdf}
    	\caption{$\pi^p$ if $b=a+2$}
    \end{subfigure}
    \begin{subfigure}[htpb]{0.24\textwidth}
    	\centering
    	\includegraphics[scale=0.8]{sym1-6.pdf}
    	\caption{$\pi^{p-1}$ if $b=a+2$}
    \end{subfigure}
    \caption[$\pi^{p-1}$ when $b < n_\mathcal{P}$.]{$\pi^{p-1}$ when $b < n_\mathcal{P}$. $P_a$ is blue, $P_b$ is pink and blue, respectively. For clarity, we omit all edges between non-consecutive pairs.}\label{h=0}
\end{figure}

    For the following arguments, see Figure \ref{fig:inductionstep} for a visualisation of the situation. Assume $h \in [\ell_{4 \rightarrow 2}]$ and suppose the claim holds for all values in $[0,h-1]$. We show that it also holds for $h$ then.
    
    The definition of $\ell_{4 \rightarrow 2}$ ensures that $P_{a-h}$ and $P_{b+h}$ are well-defined, and the induction hypothesis for $h-1$ ensures that they, as well as $P_{a-h+1}\cup P_{b+h-1}$, are elements of $\pi^{p-h}$.
	
    By Lemma \ref{lem:successive-matching}, $\deg_{P_{a-h}}(P_{a-h+1}) = 1 = \deg_{P_{b+h}}(P_{b+h-1})$. By Lemma \ref{lem:non-successive-pairs}, $\deg_{P_{b+h}}(P_{a-h+1})$ and $\deg_{P_{a-h}}(P_{b+h-1})$ are each either $0$ or $2$. But this implies that $P_{a-h+1} \cup P_{b+h-1}$ is unbalanced wrt each of $P_{a-h}$ and $P_{b+h}$. So, by Lemma \ref{rmk:work-backwards}, it follows that $P_{a-h}$ and $P_{b+h}$ must be in $\pi^{p-h} \setminus \pi^{p-h-1}$. That is, $\{P_{a-h},P_{b+h}\} = \{A_{p-h},B_{p-h}\}$, and hence $A_{p-h} \cup B_{p-h} = C_{p-h-1} \in \pi^{p-h-1}$. Moreover, Equation \ref{eq:equal-diff} yields that $\deg_{P_{b+h}}(P_{a-h+1})=\deg_{P_{a-h}}(P_{b+h-1})$. Hence, $\pi^{p-h-1} =  \{P_{a-h} \cup P_{b+h}\}\cup (\pi^{p-h}\setminus\{P_{a-h},P_{b+h}\})$, and thus the desired result holds. 
    \end{proof}

\begin{figure}[htpb]
    \centering
    \begin{subfigure}[htpb]{0.45\textwidth}
    	\centering
    	\includegraphics{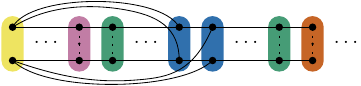}
    	\caption{$\pi^{p-h+1}$}
    \end{subfigure}
    \begin{subfigure}[htpb]{0.45\textwidth}
    	\centering
    	\includegraphics{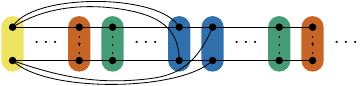}
    	\caption{$\pi^{p-h}$}
    \end{subfigure}
    \caption{Induction step for the proof of Lemma \ref{lem:first-symmetry} in the case that $b = a + 1$. The arrangement of the pairs follows the splitting order $\prec$, i.e.\ a pair further left has a lower index in $\prec$. The unique new pairs in $\pi^{p-h+1}$ (in pink and brown) are in symmetric position to $P_a$ and $P_b$ (in blue) and united to a class of size $4$ (in brown) in $\pi^{p-h}$.} \label{fig:inductionstep}
\end{figure}

For the purpose of better readability, let $\ell \coloneqq \ell_{4 \rightarrow 2}$ in the remainder of this section. The above allows us to describe the partition $\pi^{p-\ell-1}$. With the following lemma, we describe how the classes in $\pi^{p-\ell-1}$ are connected to each other. By pattern matching with the classes in Lemma \ref{lem:first-symmetry}, the reader will be able to verify that it states the following. The classes of size $4$ which split into pairs that are in $\pi^p$ are connected non-trivially if and only if they split in consecutive iterations.

\begin{lemma}\label{lem:first-symmetry-neighbours}
    For every $h \in [0,\ell-1]$, the following holds:    \begin{equation}\label{eq:consecutive-fours}
        \deg_{P_{a-h}\cup P_{b+h}}(P_{a-h-1}\cup P_{b+h+1}) = \deg_{P_{a-h-1}\cup P_{b+h+1}}(P_{a-h}\cup P_{b+h}) \in \{1,3\}.
    \end{equation}
    Moreover, for every $h \in [0,\ell-1]$ and every $C \in \pi^{p-h-2}\backslash \{C_{p-h}, C_{p-h-1}, C_{p-h-2}\}$, the following hold:    \begin{equation}\label{eq:deg_fours}
        \deg_{C}(P_{a-h}\cup P_{b+h}) \in \{0,|C|\} \ \text{ and } \deg_{P_{a-h}\cup P_{b+h}}(C) \in \{0,4\}.
    \end{equation}
\end{lemma}

\begin{proof}
Equation \ref{eq:consecutive-fours} follows directly from the proof of Lemma \ref{lem:first-symmetry}. 
Statement \ref{eq:deg_fours} follows from the fact that $C_{p-h-1} = P_{a-h} \cup P_{b+h}$ is in $\pi^{p-h-1}$ and thus balanced wrt all classes in $\pi^{p-h-2}$. If $|C|=1$, this is sufficient to prove $\deg_{C}(P_{a-h}\cup P_{b+h}) \in \{0,1\}$ and the second, equivalent statement. If $C$ is a pair, then because $P_{a-h-1},P_{a-h+1},P_{b+h-1},P_{b+h+1}$ are not in $\pi^{p-h-2}$, Lemma \ref{lem:non-successive-pairs} implies that $\deg_{C}(P_{a-h})=\deg_{C}(P_{b+h})\in \{0,2\}$, giving the required result. Finally, if $|C| = 4$, then $C=P_{a-h'}\cup P_{b+h'}$ for some $h' \in [0,h-2]$, and the above argument implies $\deg_{C}(P_{a-h})=\deg_{C}(P_{b+h})\in \{0,4\}$, giving the desired result. 
\end{proof}

With the insights so far, we have collected information about the classes and their relations from the point on when all remaining classes of size $4$ split consecutively into pairs until only pairs and singletons are left, with the pairs then consecutively splitting into singletons until a discrete partition is reached. We intend to carry this one step further and to analyse the split(s) before this moment. We are going to prove that $C_{p-\ell-2}$ is either a union of three pairs or a union of a certain pair and a singleton $S \in \pi^p_\mathcal{S}$. 
Here, we need to perform a case distinction, and we introduce an auxiliary variable $t$ to deal with the cases in parallel. We let $t \coloneqq 0$ in the first case, and $t \coloneqq 1$ in the second. 
Recall that $p$ is the index of the first iteration in which all classes have size at most $2$, and $p-\ell$ is the index of the first iteration in a cascade of splits of classes of size $4$ into pairs (until iteration $p$). Depending on the case whether $t = 0$ or $t = 1$, we let $q(t) \coloneqq p - \ell - 1 - t$, which will be the first iteration in which there are classes of size at most $4$ (\textbf{q}uadruples). Define $\pi^{q(t)}_\mathcal{S} \coloneqq \{M \in \pi^{q(t)} \mid |M| = 1\}$. 
We also define  $c\coloneqq a-\ell$,  and $d\coloneqq c-t-1$ to simplify notation, which will play a similar role in future proofs that $a$ and $b$ did in proofs up to this point.

\begin{lemma}\label{lem:first-symmetry-end}
$C_{p-\ell-2}$ is either $P_{c-1}\cup P_{c}\cup P_{n_\mathcal{P}}$ or $P_{c-1}\cup S$ for a singleton $S \in \pi^p_\mathcal{S}$. In the first case, set $t \coloneqq 0$, in the second $t \coloneqq 1$. Then,  
\begin{multline*}
    \pi^{q(t)-1} =
    \pi^{q(t)}_{\mathcal{S}} \cup
    \{P_{a-i}\cup P_{b+i} \mid i \in [0,\ell-1]\} \cup \{P_{d} \cup P_{c} \cup P_{n_\mathcal{P}}\}\\\cup 
    \left\{P_{i} \mid i \in [c-t-2]\cup [a+1,b-1] \right\} \cup M_t,
\end{multline*}
and also $M_0 = \emptyset$ and $M_1 = \{P_{c-1}\cup S\}$.
\end{lemma}

\begin{proof}
     We first prove by contradiction that $P_{c-1}$ exists. Assume that  $\ell = a-1$, i.e.\ $P_{c} = P_1$. Since $P_2 \prec P_a$ , $\ell$ must be greater than $1$, allowing us to apply Lemma \ref{lem:first-symmetry-neighbours}, along with Lemma \ref{lem:P_1-P_aP_b}, to get $\deg_{P_a\cup P_b}(P_{c}) - \deg_{P_a\cup P_b}(P_{b+\ell}) \in \{-2,2\}$. If $P_{b+\ell} \neq P_{n_\mathcal{P}}$, then we have $\deg_{P_{b+\ell+1}}(P_{c}) - \deg_{P_{b+\ell+1}}(P_{b+\ell}) \in \{-1,1\}$, by Lemmas \ref{lem:successive-matching} and \ref{lem:non-successive-pairs}. This contradicts Lemma \ref{rmk:work-backwards}. Assume then that $P_{b+\ell} = P_{n_\mathcal{P}}$, implying there are no pairs in $\pi^{p-\ell-1}$. By Lemma \ref{rmk:work-backwards}, there must be a class $C$ which satisfies $\deg_{C}(P_{a-\ell}) - \deg_{C}(P_{b+\ell}) \in \{-2,2\}$. By Lemma \ref{lem:first-symmetry-neighbours}, there are no classes of size $4$ that can satisfy this. Thus, we must have $P_1 \prec P_{a-\ell}$ and $\ell = n_\mathcal{P}-b$. 
    
 Then, $C_{p-\ell-1}$ is unbalanced by $P_{c-1}$, by Lemmas \ref{lem:non-successive-pairs} and \ref{lem:successive-matching}. Note that we have $\deg_{P_{c-1}}(P_{c})-\deg_{P_{c-1}}(P_{n_\mathcal{P}}) \in \{-1,1\}$. Next, we look for a class $B_{p-\ell-1}$ with $\deg_{B_{p-\ell-1}}(P_{c}) - \deg_{B_{p-\ell-1}}(P_{b+\ell}) \in \{-1,1\}$. By Lemma \ref{lem:first-symmetry-neighbours}, no class $P_{a-h} \cup P_{b+h} \in \pi^{p-\ell-1}$ with $h<\ell$ can satisfy this requirement. By Lemma \ref{lem:non-successive-pairs}, no pair other than $P_{c-1}$ can satisfy this equation. This leaves only $\pi^p_\mathcal{S}\cup \{P_{c}\cup P_{n_\mathcal{P}}\}$, proving the first statement of Lemma \ref{lem:first-symmetry-end}. This immediately gives us the desired result for $\pi^{q(t)-1}$ for $t=0$. 
 
 For $t=1$, the above yields $C_{q(t)}= P_{a-\ell-1} \cup S$, from which we can find $\pi^{q(t)}$, but we must still prove that $C_{q(t)-1}$ is $P_{d} \cup P_{c} \cup P_{n_\mathcal{P}}$, in order to get the desired expression for $\pi^{q(t)-1}$. First, notice that $\deg_{P_{c}\cup P_{n_\mathcal{P}}}(S) = 2$ and $\deg_{P_{c} \cup P_{n_\mathcal{P}}}(P_{c-1}) \in \{1,3\}$. Thus, without loss of generality, $A_{q(t)} = P_{c}\cup P_{n_\mathcal{P}}$ and $B_{q(t)}$ satisfies $\deg_{B_{q(t)}}(S)-\deg_{B_{q(t)}}(P_{c-1}) \in \{-1,1\}$. If $P_{c-1}=P_1$, then $P_{c-1}\cup S$ is unbalanced wrt $P_a \cup P_b$, but $P_a \cup P_b$ does not satisfy the above equation, so $P_1 \prec P_{c-1}$. Thus, $P_{c-2}$ is well-defined and, by Lemmas \ref{lem:successive-matching} and \ref{lem:first-symmetry-neighbours}, $\deg_{P_{c-2}}(P_{c-1})=1$ and $\deg_{P_{c-2}}(S) \in \{0,2\}$. This satisfies the requirement for $B_{q(t)-1}$, so $\pi^{q(t)} = \{P_{d} \cup P_{c} \cup P_{n_\mathcal{P}}\} \cup (\pi^{p-\ell-1}\backslash \{P_d, P_c \cup P_{n_\mathcal{P}}\})$. 
 \end{proof}

In the case that $t=1$, we will let $S_{q(t)}$ denote the singleton $S$ described in Lemma \ref{lem:first-symmetry-end}.  Notice that $\pi^{q(t)}_\mathcal{S} = \pi^{p}_\mathcal{S}$ if $t=0$, and $\pi^{q(t)}_\mathcal{S} = \pi^{p}_\mathcal{S}\backslash\{S_{q(t)}\}$ if $t=1$. Also, for either value of $t$, let $M_t$ be defined as in the lemma.

Having defined $C_{q(t)}$, where $t=1$, we now make some observations about its neighbourhood.

\begin{lemma}\label{lem:end-neighbours}
    Suppose that $t=1$, i.e.\ that $C_{p-\ell-2} = S_{q(t)} \cup P_{c-1}$. Then 
    \[\deg_{S_{q(t)} \cup P_{c-1}}(P_{d}\cup P_c \cup P_{n_\mathcal{P}}) \in \{1,2\}\] and 
    \[\deg_{P_{d}\cup P_c \cup P_{n_\mathcal{P}}}({S_{q(t)} \cup P_{c-1}}) = 2\cdot \deg_{S_{q(t)} \cup P_{c-1}}(P_{d}\cup P_c \cup P_{n_\mathcal{P}}) \in \{2,4\}.\]
    Furthermore, 
    \[\deg_{C}(S_{q(t)} \cup {P_{c-1}}) \in \{0,|C|\} \quad \text{and, equivalently,} \quad 
    \deg_{S_{q(t)} \cup P_{c-1}}(C)\in\{0,3\}\] 
    for every $C \in \pi^{q(t)-\ell-2}\backslash \{S_{q(t)} \cup P_{c-1}, P_d \cup P_c \cup P_{n_\mathcal{P}}\}$.
\end{lemma}

\begin{proof}
     We get the values for $\deg_{P_{d}\cup P_c \cup P_{n_\mathcal{P}}}({S_{q(t)} \cup \mspace{-2mu} P_{c-1}})$ directly from the ones for 
     $\deg_{P_{c}\cup P_{n_\mathcal{P}}}(S_{q(t)})$, $ \deg_{P_{c} \cup P_{n_\mathcal{P}}}(P_{c-1}),$ $\deg_{P_{d}}(S_{q(t)}),$ and $\deg_{P_{d}}(P_{c-1})$.  Since ${P_{d}\cup P_c \cup P_{n_\mathcal{P}}}$ has double the vertices as ${S_{q(t)} \cup P_{c-1}}$, is incident to the same number of edges, and both are balanced with respect to the other, we obtain $\deg_{S_{q(t)} \cup P_{a-\ell-1}}(P_{d}\cup P_c \cup P_{n_\mathcal{P}}) \in \{1,2\}$. 
     
     The second statement holds by the fact that $S_{q(t)} \cup P_{c-1}$ is balanced wrt all classes $C \in \pi^{p-\ell-1}\setminus\{P_d,P_c \cup P_{n_\mathcal{P}}\}$. Separately, $S_{q(t)}$ and $P_{c-1}$ are trivially attached to the classes in $\pi^{p-\ell-1} \setminus \{P_d,P_c \cup P_{n_\mathcal{P}}\}$. This holds by Lemma \ref{lem:non-successive-pairs} for singletons and pairs,  and by Lemma \ref{lem:first-symmetry-neighbours} for classes of four vertices. Since $S_{q(t)} \cup P_{c-1}$ is balanced wrt $C$, it holds that if  $S_{q(t)}$ is adjacent to $C$, then $P_{c-1}$ is as well, and vice versa. Thus $S_{q(t)} \cup P_{c-1}$ is trivially attached to $C$.
 \end{proof}

By a similar argument to that used in Lemma \ref{lem:first-symmetry}, we can work backwards from partition $\pi^{q(t)}$. Recall that $d = a - \ell - t - 1$. We define $\ell' \coloneqq \min\{\ell,d-1\}$. This definition guarantees that $P_{d-h}$ and $P_{c+h} \cup P_{n_\mathcal{P}-h}$ are elements of every $\pi^{q(t)-h}$ for $h \in [\ell']$.

\begin{figure}[htpb]
\centering
\begin{subfigure}{0.32\linewidth}
    \centering
    \includegraphics{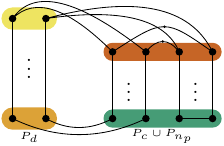}
    \caption{The iteration $p-\ell =q(t)$ when $t=0$}
    \label{fig:sym2-1}
\end{subfigure}
\begin{subfigure}{0.32\linewidth}
    \centering
    \includegraphics{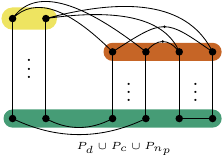}
    \caption{The iteration $p-\ell-1=q(t)-1$ when $t=0$}
    \label{fig:sym2-2}
\end{subfigure}
\begin{subfigure}{0.32\linewidth}
    \centering
    \includegraphics{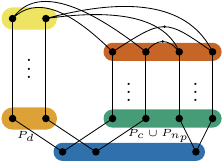}
    \caption{The iteration $p-\ell-1=q(t)$ when $t=1$}
    \label{fig:sym2-3}
\end{subfigure}
\begin{subfigure}{0.32\linewidth}
    \centering
    \includegraphics{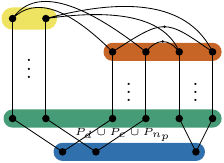}
    \caption{The iteration $p-\ell-2=q(t)$ when $t=1$}
    \label{fig:sym2-4}
\end{subfigure}
\begin{subfigure}{0.32\linewidth}
    \centering
    \includegraphics{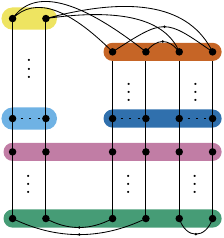}
    \caption{Induction Hypothesis}
    \label{fig:sym2-5}
\end{subfigure}
\begin{subfigure}{0.32\linewidth}
    \centering
    \includegraphics{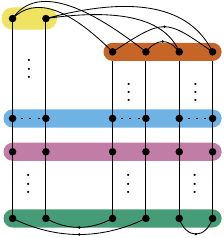}
    \caption{Induction Step}
    \label{fig:sym2-6}
\end{subfigure}
\caption{Diagrams depicting the base case and induction arguments for Lemma \ref{lem:second-symmetry}.}
\label{fig:sym2}
\end{figure}

\begin{lemma}\label{lem:second-symmetry}
For all $h \in [0,\ell']$,
\begin{multline*}
\pi^{q(t)-h-1} = \pi^{q(t)}_{\mathcal{S}} \cup \{P_{c+i} \cup P_{n_\mathcal{P}-i} \mid i \in [h+1,\ell]\} \cup \{P_{d-i}\cup P_{c+i} \cup P_{n_\mathcal{P}-i} \mid i \in [0,h]\} \\\cup 
	\{P_{i} \mid i \in [d-h-1]\cup [a+1,b-1]\} \cup M_t.
\end{multline*}
\end{lemma}

\begin{proof}
By using the relations between the variables $a$, $b$, $c$, $d$, the base case of $h=0$ is given by Lemma \ref{lem:first-symmetry-end}. Consider $h \in [\ell']$, with the above statement holding up to $h-1$. Then, by Lemmas \ref{lem:successive-matching}, \ref{lem:first-symmetry-neighbours}, and \ref{rmk:work-backwards}, and by the same argument as used in Lemma \ref{lem:first-symmetry}, we obtain $C_{q(t)-h-1} = P_{d-h} \cup P_{c+h} \cup P_{n_\mathcal{P}-h}$.
\end{proof}

\begin{corollary}\label{cor:second-symmetry-neighbours}
    For $h \in [0,\ell'-1]$,
    \begin{equation} \label{eq:consecutive-sixes}
    \deg_{C_{q(t)-h-2}}(C_{q(t)-h-1}) = \deg_{C_{q(t)-h-1}}(C_{q(t)-h-2}) \in \{1,5\}.
    \end{equation}
    Also,	
\begin{equation}\label{eq:selfdeg_sixes}
     \deg_{C_{q(t)-h-1}}(C_{q(t)-h-1}) \in \{0,1,4,5\}.\end{equation}
    For every $h \in [0,\ell'-1]$ and $C \in \pi^{q(t)-h-2}\backslash \{C_{q(t)-h-2}, C_{q(t)-h-1}, C_{q(t)-h}\}$, it holds that
	\begin{equation}\label{eq:deg_sixes}
    \deg_{C_{q(t)-h-1}}(C) \in \{0,6\} \text{ and } \deg_C(C_{q(t)-h-1}) \in \{0,|C|\}.\end{equation}   
\end{corollary}
\begin{proof}
     Using Equation \ref{eq:equal-diff}, we obtain \[\bot \neq \deg_{C_{q(t)-h-1}}(C_{q(t)-h-2}) \in \{1,5\}\] and hence Equation \ref{eq:consecutive-sixes}. 

    To see Equation \ref{eq:selfdeg_sixes}, recall that $C_{q(t)-h-1} \in \pi^{q(t)-h-1}$ is balanced wrt $C_{q(t)-h-1}\in \pi^{q(t)-h-2}$. Also $P_{d-h}$ and $P_{c+h} \cup P_{n_\mathcal{P}-h}$ are trivially attached, from Lemma \ref{lem:first-symmetry-neighbours}. There is also $\deg_{P_{d-h}}(P_{d-h}) \in \{0,1\}$. These ingredients together give $\deg_{C_{q(t)-h-1}}(P_{d-h}) \in \{0,1,4,5\}$ and therefore, by balancedness $\deg_{C_{q(t)-h-1}}(C_{q(t)-h-1}) \in \{0,1,4,5\}$.

Equation \ref{eq:deg_sixes} follows from the fact that $P_{a-h} \cup P_{b+h} = C_{q(t)-h-1} \in \pi^{q(t)-h-1}$, and it is thus balanced wrt all classes in $\pi^{q(t)-h-2}$. If $C$ is a singleton, this is sufficient to show $\deg_{C}(P_{a-h}\cup P_{b+h}) \in \{0,1\}$ and the second, equivalent statement. If $C$ is a pair, then by assumption $C \notin \{P_{a-h-1},P_{a-h+1},P_{b+h-1},P_{b+h+1}\}$. Hence, Lemma \ref{lem:non-successive-pairs} implies that $\deg_{C}(P_{a-h})=\deg_{C}(P_{b+h})\in \{0,2\}$, giving the required result. Finally, if $|C| = 4$, then $C=P_{a-h'}\cup P_{b+h'}$ for some $h' \in [0,h-2]$ and the above argument implies $\deg_{C}(P_{a-h})=\deg_{C}(P_{b+h})\in \{0,4\}$, giving the desired result. 
\end{proof}

Recall that by Lemma \ref{lem:bdista}, we know that $b \in \{a+1,a+2\}$. Depending on which of the two situations occurs, the following lemma establishes the possible ranges for $d$. We will use this insight, and a case distinction on the values of $d$, in Section \ref{sec:max4} to establish the existence and non-existence of long-refinement graphs and to classify them.

\begin{lemma}\label{lem:value-of-d}
If $b=a+1$, then $d \in [\ell,\ell+4]$, and if $b=a+2$, then $d \in [\ell,\ell+2]$.
\end{lemma}

\begin{proof}

    \begin{figure}[htpb]
	\centering
	\begin{subfigure}{0.3\textwidth}
		\centering
		\includegraphics{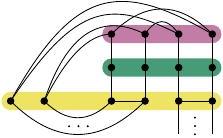}
		\caption{$\pi^{q(t)-d}$ if $d=\ell-1$}\label{val-d=l-1}
	\end{subfigure}
    \hspace{1mm}
    \begin{subfigure}{0.3\textwidth}
		\centering
		\includegraphics{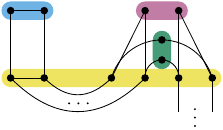}
		\caption{$\pi^{q(t)-\ell-1}$ if $d>\ell+2$ and $b = a+2$}\label{val-d>l+2}
	\end{subfigure}
    \hspace{3mm}
    \begin{subfigure}{0.3\textwidth}
		\centering
		\includegraphics{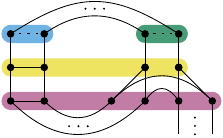}
		\caption{$\pi^{q(t)-\ell-2}$ if $d>\ell+4$}\label{val-d>l+4}
	\end{subfigure}
	\caption[Possible values of $d$]{Examples of contradictions occurring in the proof of Lemma \ref{lem:value-of-d} for $d<\ell$, $d > \ell+2$, and $d>\ell+4$ when the degrees are $2$ and $3$. Notice that yellow classes are unbalanced wrt blue, green, and pink classes.}
	\label{fig:value-of-d:contradictions}
    \end{figure}
    
    We prove this by showing that the assumption $d < \ell$ leads to a contradiction, as does $d>\ell+2$ when $b=a+2$ and $d>\ell+4$ when $b=a+1$. For the following, consider Figure \ref{fig:value-of-d:contradictions}.

   Assume that $1\leq d \leq \ell-1$. Then $P_{1} = P_{d-h}$, for some $h \leq \ell-2$. Recall that $C_{q(t)-h-1} = P_{1} \cup P_{c+h} \cup P_{n_\mathcal{P}-h}$. See the lowest class in Figure \ref{val-d=l-1}, which is unbalanced wrt all three classes. We have $\deg_{P_a \cup P_b}(P_1) = 2$. Since $P_{c+h}$ and $P_{n_\mathcal{P}-h}$ are not consecutive to $P_a$ or $P_b$, we have $\deg_{P_a \cup P_b}(P_{c+h}\cup P_{n_\mathcal{P}-h}) \in \{0,4\}$. Here, we have a difference of $2$. We also know that \[\deg_{P_{c+h+1}\cup P_{n_\mathcal{P}-h-1}}(P_{c+h}\cup P_{n_\mathcal{P}-h}) \in \{1,3\},\] whereas \[\deg_{P_{c+h+1}\cup P_{n_\mathcal{P}-h-1}}(P_1) \in \{0,4\}.\] We cannot get a difference of $2$ from these values, hence we have a contradiction to Lemma \ref{rmk:work-backwards}, proving $d \geq \ell$.

   Next, assume that $d \geq \ell+3$ and $b=a+2$. Then, $C_{q(t)-\ell-1} = P_{d-\ell}\cup P_a \cup P_b$, with $d-\ell \geq 3$. See Figure \ref{val-d>l+2}, where the yellow class $C_{q(t)-\ell-1}$ is unbalanced wrt the green class $P_{a+1}$, the blue class $P_{d-\ell-1}$ and the pink class $P_1$. Formally, since $b = a+2$, $\deg_{P_{a+1}}(P_a \cup P_b) = 1$ and, by Lemma \ref{lem:non-successive-pairs}, $\deg_{P_{a+1}}(P_{d-\ell}) \in \{0,2\}$. Next, we have $\deg_{P_{1}}(P_a \cup P_b) = 1$ and, by Lemma \ref{lem:non-successive-pairs}, $\deg_{P_{1}}(P_{d-\ell}) \in \{0,2\}$.  Finally, we have $\deg_{P_{d-\ell-1}}(P_a \cup P_b) \in \{0,2\}$ and $\deg_{P_{d-\ell-1}}(P_{d-\ell}) = 1$. Thus, as we have  $C_{q(t)-\ell-1}$ unbalanced wrt at least three classes, we obtain a contradiction to Lemma \ref{rmk:work-backwards} and have proved $d \leq \ell+2$ for $b = a+2$.

   Finally, assume that $b=a+1$ and, towards a contradiction, that $d \geq \ell+5$. See Figure \ref{val-d>l+4}. As before, we have $\deg_{P_{1}}(P_a \cup P_b) = 1$ and with Lemma \ref{lem:non-successive-pairs}, $\deg_{P_{1}}(P_{d-\ell}) \in \{0,2\}$. Also, $\deg_{P_{d-\ell-1}}(P_a \cup P_b) \in \{0,2\}$ and $\deg_{P_{d-\ell-1}}(P_{d-\ell}) = 1$. We must therefore have $C_{q(t)-\ell-2} = P_1 \cup P_{d-\ell-1}$. Then, $C_{q(t)-\ell-2}$ is unbalanced wrt $P_2$, $P_{d-\ell-1}$, and $C_{q(t)-\ell-1}$. This is a contradiction and therefore $d \leq \ell+4$. 
\end{proof}

With the reverse-engineering methods in this section, we have managed to reconstruct partitions during the execution of Colour Refinement in which the classes have size up to $6$. This could be carried further, but the analysis would become more and more involved. However, when restricting the degrees to a maximum of $4$, these last iterations comprise enough of the entire execution of Colour Refinement to yield characterisations and compact representations for all long-refinement graphs with small degrees. We will see this in Section \ref{sec:max4}, where we apply the ingredients collected in this section.

\section{Long-Refinement Graphs with Maximum Degree at most 4}\label{sec:max4}

In this section, we establish a comprehensive list of all the long-refinement graphs with maximum degree at most $4$. We use the observation from Lemma \ref{lem:value-of-d} that $d \in [\ell,\ell+4]$ and consider the cases separately. Unifying previous results \cite{KMcK20}, we include a string notation to represent the long-refinement graphs $G$ with $\deg(G) = \{2,3\}$, where each letter represents a pair.

\begin{notation}\label{notation:23}
	We define $\sigma \colon [n_\mathcal{P}] \rightarrow \Sigma$ as the function that maps $i$ to the element in $\Sigma$ that matches the type of $P_i$. More precisely, $\sigma(i)$ is defined as follows:
	\begin{itemize}
		\setlength\itemsep{0em}
		\item $\mathrm{S}$, \ if $P_i = \min(\prec) = P_1$.
		\item $\mathrm{X}$, \ if $P_i \in \{P_a,P_b\}$ and $P_i$ is not adjacent to any singleton $S \in \pi^p_{\mathcal{S}}$. 
		\item $\mathrm{X}_2$, \ if $P_i \in \{P_a,P_b\}$ and $P_i$ is adjacent to a singleton $S_i \in \pi^p_{\mathcal{S}}$.
		\item $0$, \ if $P_i \notin \{P_1,P_a,P_b\}$, $\deg(P_i)= 2$ and $P_i$  is not adjacent to any singleton $S \in \pi^p_{\mathcal{S}}$.
		\item $0_2$, \ if $P_i \notin \{P_1,P_a,P_b\}$, $\deg(P_i)= 2$, and $P_i$ is adjacent to a singleton $S_i  \in \pi^p_{\mathcal{S}}$.
		\item $1$, \ if $P_i \notin \{P_1,P_a,P_b\}$ with $\deg(P_i)= 3$, such that $P_i$ is not adjacent to any  $S \in \pi^p_{\mathcal{S}}$.
		\item $1_2$, \ if $P_i \notin \{P_1,P_a,P_b\}$, $\deg(P_i)= 3$, and $P_i$ is adjacent to a singleton $S_i  \in \pi^p_{\mathcal{S}}$.
	\end{itemize}
	Since $\prec$ is a linear order, we can define a string representation based on a long-refinement graph $G$, where the $i$-th character in the string is $\sigma(i)$ and corresponds to the $i$-th element of $\prec$, $P_i$. We call a string defined in this way a \emph{long-refinement string}.  
\end{notation}

The strings describe the exact neighbourhood of each pair and each singleton, as long as there is at most one singleton. We will prove that this is the case. Thus, every string $\Xi$ unambiguously represents exactly one graph, up to isomorphism, which we denote by $G(\Xi)$.

The limitation on degree allows for some immediate insights into the structure of long-refinement graphs. For example, Corollary \ref{cor:second-symmetry-neighbours}  yields the following when the maximum degree is $4$.

\begin{observation}\label{degree-middle6}
    For $h \in [1,\ell'-1]$, $C_{q(t)-h-1}$ (i.e.\ the class $P_{d-h} \cup P_{c+h} \cup P_{n_\mathcal{P}-h}$) has degree $2$ or $3$, with $N(C_{q(t)-h-1}) \subseteq C_{q(t)-h-2} \cup  C_{q(t)-h-1} \cup  C_{q(t)-h}$.
\end{observation}

With $h = 0$, Corollary \ref{cor:second-symmetry-neighbours} also describes the neighbourhood of $C_{q(t)-1}$ (i.e. $P_d \cup P_c \cup P_{n_\mathcal{P}}$), which depends on the value of $t$. Figure \ref{fig:graph-bottoms} illustrates the possible edge layouts between $C_{q(t)}$ and $C_{q(t)-1}$ when $t=0$ or $t=1$, which are a result of the proof of Lemma \ref{lem:first-symmetry-end}. This yields the following observation.

\begin{observation}\label{degree-bottom6}
    If $t=0$, then $\deg(C_{q(t)-1})= 2$ with \[\deg_{C_{q(t)-2}}(C_{q(t)-1}) = 1 = \deg_{C_{q(t)-1}}(C_{q(t)-1}),\]
    and $\deg_{C}(C_{q(t)-1}) = 0$ for every other class $C \subseteq V(G)$. 
    
    \vspace{1ex}
    If $t=1$, then $\deg(C_{q(t)-1})\in \{2, 3, 4\}$ with
    \[\deg_{C_{q(t)-2}}(C_{q(t)-1}) = 1, \quad \deg_{C_{q(t)-1}}(C_{q(t)-1}) \in \{0,1\}, \quad \deg_{C_{q(t)}}(C_{q(t)-1}) \in \{1,2\},\] 
    and $\deg_{C}(C_{q(t)-1}) = 0$ for every other class $C \subseteq V(G)$.
\end{observation}

\begin{figure}[htpb]
         \centering 
    \begin{subfigure}{0.3\linewidth}
        \centering 
        \includegraphics{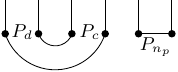}
         \caption{}\label{fig:graph-bottom-t=0}
    \end{subfigure}
    \begin{subfigure}{0.3\linewidth}
        \centering 
        \includegraphics{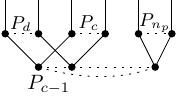}
         \caption{}
         \label{fig:graph-bottom-t=1A}
    \end{subfigure}
    \begin{subfigure}{0.3\linewidth}
        \centering 
        \includegraphics{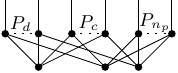}
         \caption{}
         \label{fig:graph-bottom-t=1B}
    \end{subfigure}
         \caption[]{All possible layouts between $C_{q(t)-1}$ ($P_d \cup P_c \cup P_{n_\mathcal{P}}$) and $C_{q(t)}$ ($P_c \cup P_{n_\mathcal{P}}$ or ${S_{q(t)} \cup P_{c-1}}$) when the maximum degree is at most $4$. Note}    \label{fig:graph-bottoms}
\end{figure}

We begin by proving the following two technical lemmas, which will assist us in cases where $\ell$ can be arbitrarily large. Recall that $C_{q(t)-i-1} = P_{d-i} \cup P_{c+i} \cup P_{n_\mathcal{P}-i}$ for $i \in [0,\ell']$ and $C_{q(t)-i-1} = P_{c-1} \cup S_{q(t)}$ for $i = -1$ when $t=1$.

\begin{lemma}\label{lem:1classgettingbigger}
    Let $j$ be an iteration $j>0$ and let the following statements hold.
    \begin{enumerate}
            \item $\{C_{q(t)-i-1} \mid i \in [-t,\ell']\} \subseteq\pi^j$, with $l'+t > 0$.
            \item $B_j$ is $C_{q(t)-\ell'-1}$, and $A_j$ is not contained in $\{C_{q(t)-i-1} \mid i \in [-t,\ell']\}$ 
    	\item $\deg_{C_{j-1}}(A_j) = \deg_{C_{j-1}}(B_j) + \deg_{C_{q(t)-\ell'}}(B_{j})$, where $\deg_{C_{q(t)-\ell'}}(B_{j}) = 1$ (or $2$ if $\ell'=0$ and $t=1$ as depicted in Figure \ref{fig:graph-bottom-t=1B}) and $\{\deg_{C_{j-1}}(A_j)\} = \deg(A_j)$.
    \end{enumerate}
    Then, for $j' \in [0,\ell'+t]$,  $\pi^{j-1-j'}\backslash \pi^j = \{C_{j-1-j'}\}$ and $C_{j-1-j'} = A_{j} \cup \left(\bigcup_{i \in [0,j']}C_{q(t)-\ell'+i-1}\right)$.
\end{lemma}
\begin{proof}
For $j'= 0$, we simply observe that $C_{j-1} = A_j \cup B_j$ is $A_j \cup C_{q(t)-\ell'-1}$, as required.

For $j'>1$, we assume the statement holds up until $j'-1$. We wish to show that \[C_{j-j'} = C_{j-j'+1} \cup  C_{q(t)-\ell'+j'-1}.\] To do this, we prove that  $C_{j-j'+1}$ is unbalanced wrt itself and $C_{q(t)-\ell'+j'-1}$.  
Since $A_j \subseteq C_{j-j'+1}$ and $N(A_j) \subseteq C_{j-1} \subseteq C_{j-j'+1}$, there are vertices $v \in A_j$ in $C_{j-j'+1}$ such that $\deg(v) = \deg_{C_{j-j'+1}}(v)$. The upper limit on $j'$ gives  ${{q(t)-\ell'+j'-1}\leq q(t)-1+t}$ so we can apply Corollary \ref{cor:second-symmetry-neighbours} and Lemma \ref{lem:end-neighbours} to show that $\deg_{C_{q(t)-\ell'+j'-1}}(C_{q(t)-\ell'+j'-2}) = 1$ (or $2$ if $\ell'=0$, and $t=1$ as depicted in Figure \ref{fig:graph-bottom-t=1B}). Since $C_{q(t)-\ell'+j'-2} \subseteq C_{j-j'+1}$, there are vertices $u$ with $\deg_{C_{q(t)-\ell'+j'-1}}(u)=1$ (or $2$) and thus also $\deg_{C_{j-j'+1}}(u) < \deg(u)$. We thus get that $C_{j-j'+1}$ is unbalanced wrt itself and $C_{q(t)-\ell'+j'-1}$. By Lemma \ref{rmk:work-backwards}, this proves the statement for $j'$.
\end{proof}

\begin{lemma}\label{lem:2classgettingbigger}
    Let $j>0$ be an iteration such that the following statements hold.
    \begin{enumerate}
            \item $\{C_{q(t)-i-1} \mid i \in [-t,\ell']\} \subseteq \pi^j$, with $\ell' > 0$.
    	\item $B_j$ is $C_{q(t)-\ell'-1}$, and $A_j$ is not contained in $\{C_{q(t)-i-1} \mid i \in [-t,\ell']\}$.
            \item There exists a class $M$ with $\deg_{A_j}(M) = 2$, $\deg_M(A_j) = 1$, and $\deg_{B_j}(M) = \deg_{M}(B_j)=0$\label{it:classM}.
            \item $\{\deg_{C_{j-1}}(B_j)+\deg_{C_{q(t)-\ell'}}(B_j)\} = \deg(B_j)$ with  $\deg_{C_{q(t)-\ell'}}(B_{j}) = 1$ (or $2$ if $\ell'=0$ and $t=1$ as depicted in Figure \ref{fig:graph-bottom-t=1B}).
    \end{enumerate}
Then $\deg_{C_{q(t)-\ell'}}(B_{j}) = 1$ and, for $j' \in [0,\ell'-1+t]$, there are two classes $E_{j'},D_{j'}$ in $\pi^{j-2-j'}$ which satisfy the following statements.
\begin{itemize}
    \item $\pi^{j-2-j'}\backslash\pi^{j} = \{D_{j'},E_{j'}\}$
    \item $D_{j'} = M \cup \bigcup_{i\in [0, \lfloor\frac{j'-1}{3}\rfloor]}C_{q(t)-\ell'+3i} $
    \item $E_{j'} = A_j \cup B_j \cup \left(\bigcup_{i \in [0, \lfloor\frac{j'-2}{3}\rfloor]} C_{q(t)-\ell'+1+3i}\right) \cup \left(\bigcup_{i \in [0, \lfloor\frac{j'-3}{3}\rfloor]}C_{q(t)-\ell'+2+3i}\right)$
\end{itemize}
\end{lemma}

\begin{proof}
    We begin with the base case, $j'=0$, by finding an expression for $\pi^{j-2}$. Given the assumed values of $A_j$ and $B_j$, we can see that $\pi^{j-1} = \pi^{j}\backslash\{A_j,B_j\} \cup \{A_j\cup B_j\}$. Then, by Assumption \ref{it:classM}, the class $C_{j-1}=A_j\cup B_j$ is clearly unbalanced wrt $M$. Since $\ell'+t>0$, the class $C_{q(t)-\ell'}$ is in $\pi^{j-1}$, which implies that $C_{j-1}$ is unbalanced wrt $C_{q(t)-\ell'}$. Thus, by Lemma \ref{rmk:work-backwards}, $C_{j-2} = M \cup C_{q(t)-\ell'}$. It is therefore apparent that $\pi^{j-2}\backslash \pi^j = \{C_{j-1},C_{j-2}\}$ and that letting $D_{0} \coloneqq C_{j-2}$ and $E_0 \coloneqq C_{j-1}$ satisfies all desired statements. Note that this also contradicts $\deg_{C_{q(t)-\ell'}}(B_j)=2$, as it would not satisfy Lemma \ref{rmk:work-backwards}.

    We next consider the induction step, where $j \in [1,\ell'-1+t]$. We assume that the statements hold up to $j'-1$ (i.e. $\pi^{j-1-j'}$ and $D_{j'-1},E_{j'-1}$ are as described) and prove the same for $j'$ (i.e $\pi^{j-2-j'}$ and $D_{j'},E_{j'}$). We must consider three options separately: a) $C_{j-1-j'} = D_{j'-1}$, i.e.\ when $j'-1 \equiv 1 \mod 3$; b) $C_{j-1-j'} = E_{j'-1}$ and $C_{j-j'} = D_{j'-1}$, i.e.\  $j'-1 \equiv 2 \mod 3$; and c) $C_{j-1-j'} = E_{j'-1}$ with $C_{j-j'} \subseteq E_{j'-1}$ when $j'-1 \equiv 0 \mod 3$. 
    
    In all three cases, since $j'$ is no greater than $\ell'-1+t$, we may apply Corollary \ref{cor:second-symmetry-neighbours} and Lemma \ref{lem:end-neighbours} to observe that $C_{q(t)-\ell'+j'-2}$ has degree $1$ (or $2$ in the case where $t=1$ as depicted in Figure \ref{fig:graph-bottom-t=1B}), into $C_{q(t)-\ell'+j'-1}$. Therefore, $C_{j-1-j'}$ is unbalanced wrt $C_{q(t)-\ell'+j'-1}$, , which we may denote $A_{j-1-j'}$. To construct $\pi^{j-2-j'}$, we need $B_{j-1-j'}$. The analysis differs by case, with arguments presented below for each.
    
    In the first case, where $C_{j-1-j'} = D_{j'-1}$, we have by assumption $\deg_{E_{j'-1}}(M)=2$ whereas $\deg_{E_{j'-1}}({C_{q(t)-\ell'+j'-2}}) = 1$. This gives us that $B_{j-1-j'}=E_{j'-1}$ and thus 
    \begin{align*}
        D_{j'} &= D_{j'-1}\, ,\\
        E_{j'} &= E_{j'-1}\cup C_{q(t)-\ell'+j'-1}\, ,
    \end{align*}
    which is what we expect for $j' \equiv 2 \mod 3$, and also $\pi^{j-2-j'}\backslash\pi^{j} = \{D_{j'},E_{j'}\}$. 
    
    Now, consider the second case where $C_{j-1-j'} = E_{j'-1}$ and $C_{j-j'} = D_{j'-1}$. Consider an arbitrary $u \in B_j \subseteq E_{j'-1}$ which, by assumption, has degree $1$ into $D_{j'-1}$ and degree $\deg(u)-1$ into $E_{j'-1}$. Then $v \in C_{q(t)-\ell'+j'-2}$ has degree $1$ into $D_{j'-1}$, degree $1$ (or 2) into $A_{j-1-j'}$, and degree $\deg(v)-2$ (or $\deg(v) - 3) $ into $E_{j'-1}$. Hence, in this case, we obtain $B_{j-1-j'}= E_{j'-1}$ and 
    \begin{align*}
        D_{j'} &= D_{j'-1}\, ,\\
        E_{j'} &= E_{j'-1}\cup C_{q(t)-\ell'+j'-1}\, ,
    \end{align*}
    which is what we expect for $j' \equiv 0 \mod 3$. 
    
    Finally, we treat the case where $C_{j-1-j'} = E_{j'-1}$ with $C_{j-j'} \subseteq E_{j'-1}$. In this case, $C_{q(t)-\ell'+j'-2}$ has degree $0$ into $D_{j'-1}$,  whereas $B_j\subseteq E_{j'-1}$ has degree $1$ into $D_{j'-1}$, making $B_{j-1-j'} = D_{j'-1}$. Then we have $B_{j-1-j'}= D_{j'-1}$ and
    \begin{align*}
        D_{j'} &= D_{j'-1} \cup C_{q(t)-\ell'+j'-1}\, ,\\
        E_{j'} &= E_{j'-1}\, ,
    \end{align*}
    which is what we expect for $j' \equiv 1 \mod 3$. 
\end{proof}

Having established these results, we now begin to look for long-refinement graphs with maximum degree at most $4$. We make the following assumption throughout the rest of this section.

\begin{assumption}
$G$ is a long-refinement graph with $\max \deg(G) \leq 4$.
\end{assumption}

Recall that, by Lemma \ref{lem:value-of-d}, it holds that $d \geq \ell$. We first treat the case that $d = \ell$. In this case, $\ell'$, which we recall is defined as $\min\{\ell,d-1\}$, is thus $\ell-1$. Thus, Lemma \ref{lem:second-symmetry} with some substitutions yields the following equation.
 \begin{multline}\label{eq:d=l-pi^q(t)-l}
	\pi^{q(t)-\ell} = \pi^{q(t)}_{\mathcal{S}} \cup \{P_{i}\cup P_{a-i} \cup P_{b+i} \mid i \in [\ell-1]\} \\\cup \{P_{i} \mid i \in [a+1,b-1]\}\cup 
	\{P_{a} \cup P_{b}\} \cup \{S\cup P_{c-1} \mid t = 1\}
\end{multline} 

The next lemma will allow us to determine the possible expressions for $\pi^{q(t)-\ell-1}$.

\begin{lemma}\label{lem:d=l-A,B}
    Let $d=\ell$. Then (after possibly swapping the roles of $A_{q(t)-\ell}$ and $B_{q(t)-\ell}$) we have that $A_{q(t)-\ell} = P_a \cup P_b$, and either $B_{q(t)-\ell} \in \pi^{q(t)}_\mathcal{S}$ or $B_{q(t)-\ell} = P_1 \cup P_{a-1} \cup P_{b+1}$%
    , each with $\deg_{B_{q(t)-\ell}}(P_{a-1}\cup P_{b+1}) - \deg_{B_{q(t)-\ell}}(P_{1}) = 1$.
\end{lemma}
 \begin{proof}
 Note the inclusion of $P_1 \cup P_{a-1} \cup P_{b+1}$ and $P_a \cup P_b$ in the expression for $\pi^{q(t)-\ell}$ (Equation \ref{eq:d=l-pi^q(t)-l}). Lemma \ref{lem:P_1-P_aP_b} yields $\deg_{P_a\cup P_b}(P_1) = 2$, and Lemma \ref{lem:first-symmetry-neighbours} with the degree restrictions on $P_a\cup P_b$ yields $\deg_{P_a\cup P_b}(P_{a-1}\cup P_{b+1}) = 1$. Then $P_1 \cup P_{a-1} \cup P_{b+1}$ is unbalanced wrt $P_a \cup P_b$ and therefore, by Lemma \ref{rmk:work-backwards}, either $A_{q(t)-\ell}$ or $B_{q(t)-\ell}$ is $P_a \cup P_b$. Let us fix $A_{q(t)-\ell} = P_a \cup P_b$, without loss of generality. By Lemma \ref{rmk:work-backwards}, the class $B_{q(t)-\ell} \in \pi^{q(t)-\ell}\backslash\{A_{q(t)-\ell}\}$ must satisfy $\deg_{B_{q(t)-\ell}}(P_{a-1} \cup P_{b+1}) - \deg_{B_{q(t)-\ell}}(P_{1}) = 1$. 
     By Corollary \ref{cor:second-symmetry-neighbours}, there is no class $C \in \{P_{i}\cup P_{a-i} \cup P_{b+i} \mid i \in [2,\ell-1]\}$ that can satisfy this: by Observations \ref{degree-middle6} and \ref{degree-bottom6}, either $i=2$ and $G[P_1 \cup P_{a-1} \cup P_{b+1},C]$ is a matching, or $E(G[P_1 \cup P_{a-1} \cup P_{b+1},C]) = \emptyset$.  Lemma \ref{lem:end-neighbours} discounts  $B_{q(t)-\ell} = S \cup P_{c-1}$. So, observing the expression for $\pi^{q(t)-\ell}$, we see that $B_{q(t)-\ell}\in \pi^{q(t)}_\mathcal{S}$ or $B_{q(t)-\ell+1} = P_1 \cup P_{a-1} \cup P_{b+1}$. 
 \end{proof}

 Lemma \ref{lem:d=l-A,B} yields two possible expressions for $\pi^{q(t)-\ell-1}$. Taking $B_{q(t)-\ell} \in \pi^{q(t)}_\mathcal{S}$, we fix one expression for $\pi^{q(t)-\ell-1}$. By analysing this partition, we can determine possible expressions of previous partitions and the graphs that realise them.

\begin{lemma}\label{lem:dldegnot23}
Assume $d=\ell$, $A_{q(t)-\ell} = P_a \cup P_b$, and $B_{q(t)-\ell} \in \pi^{q(t)}_\mathcal{S}$. Then $G$ is one of the graphs depicted in Figure \ref{fig:LR-with-d=l-singleton}. 
     \begin{figure}[htpb]
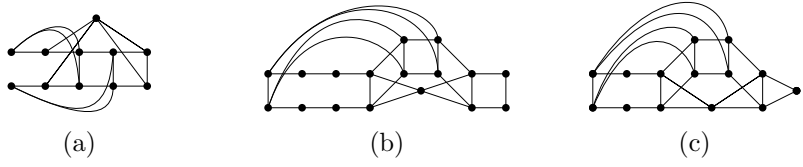

         \centering 
         \begin{subfigure}{0.24\linewidth}
            \centering 
            \includegraphics[scale=0.8]{d=l-deg34-11.pdf}
         \caption{}\label{fig:d=l-deg34-11}
         \end{subfigure}
         \begin{subfigure}{0.24\linewidth}
            \centering 
            \includegraphics[scale=0.8]{d=l-deg24-17.pdf}
         \caption{}\label{fig:d=l-deg24-17}
         \end{subfigure}
         \begin{subfigure}{0.24\linewidth}
            \centering 
            \includegraphics[scale=0.8]{d=l-deg24-14.pdf}
         \caption{}\label{fig:d=l-deg24-14}
         \end{subfigure}
         \caption{The long-refinement graphs with $d=\ell$ and $B_{q(t)-\ell} \in \pi^{q(t)}_\mathcal{S}$.}
         \label{fig:LR-with-d=l-singleton}
     \end{figure}
\end{lemma}

\begin{proof}
	The expression for $\pi^{q(t)-\ell}$ from Equation \ref{eq:d=l-pi^q(t)-l} and the values of $A_{q(t)-\ell}$ and $B_{q(t)-\ell}$ give the following expression for $\pi^{q(t)-\ell-1}$.
	\begin{multline*}
		\pi^{q(t)-\ell-1} = \pi^{q(t)}_{\mathcal{S}}\backslash\{B_{q(t)-\ell}\} \cup \{P_{i}\cup P_{a-i} \cup P_{b+i} \mid i \in [\ell-1]\} \\\cup \{P_i \mid i \in [a+1,b-1]\ \} \cup
		\{P_{a} \cup P_{b} \cup B_{q(t)-\ell}\} \cup M_t
	\end{multline*}
	
	Recall that  $\deg_{B_{q(t)-\ell}}(P_{a-1}\cup P_{b+1}) - \deg_{B_{q(t)-\ell}}(P_{1}) = 1$. Since $|B_{q(t)-\ell}|=1$, this implies that $\deg_{B_{q(t)-\ell}}(P_1) = 0$ and $\deg_{B_{q(t)-\ell}}(P_{a-1}\cup P_{b+1}) = 1$. Then, $\deg_{P_1 \cup P_{a-1} \cup P_{b+1}}(B_{q(t)-\ell}) = 4$ and $\deg_{P_1 \cup P_{a-1} \cup P_{b+1}}(P_a \cup P_b) = 2$ by Lemmas \ref{lem:P_1-P_aP_b} and \ref{lem:non-successive-pairs}. Thus $C_{q(t)-\ell-1}$ is unbalanced wrt $P_1 \cup P_{a-1} \cup P_{b+1}$ so we can assume, without loss of generality, that $A_{q(t)-\ell-1} = P_1 \cup P_{a-1} \cup P_{b+1}$. By degree restrictions on $B_{q(t)-\ell}$, $\deg_{B_{q(t)-\ell}}(B_{q(t)-\ell-1}) = 0$, so to satisfy Lemma \ref{rmk:work-backwards}, $P_a \cup P_b$ will have degree $2$ into $A_{q(t)-\ell-1}$ and $B_{q(t)-\ell-1}$, and $0$ into all other classes in $\pi^{q(t)-\ell-1}$. If $b=a+2$, the degree of $P_a \cup P_b$  into class $P_{a+1}$ is $1$,  
	so $b=a+1$ and $P_a,P_b$ must be consecutive pairs. As $P_a \cup P_b$ has non-zero degree into $B_{q(t)-\ell} \cup P_a \cup P_b$, the latter class must be $B_{q(t)-\ell-1}$. This implies that $P_a \cup P_b$ has degree $2$ into $B_{q(t)-\ell} \cup P_a \cup P_b$ and thus that $P_a,P_b \in E(G)$. With this and the expression for $A_{q(t)-\ell-1}$, we get 
    $C_{q(t)-\ell-2} = B_{q(t)-\ell}\cup P_1 \cup P_{a-1} \cup P_a \cup P_b \cup P_{b+1}$. 
    Either $C_{q(t)-\ell-2}$ is $V(G)$ and $\deg(B_{q(t)-\ell} \cup P_a \cup P_b) \neq \deg(P_1 \cup P_{a-1} \cup P_{b+1})$, or 
	all vertices in ${C_{q(t)-\ell-2}}$ have degree $4$. In the first case, since vertices in $P_1$, and therefore also $P_{a-1} \cup P_{b+1}$, have degree at least 3, $G$ must be the graph illustrated in Figure \ref{fig:d=l-deg34-11} with $\deg(G) = \{3,4\}$. 
	
	Now, consider the second case. Notice that  $B_{q(t)-\ell} \cup P_a \cup P_b$ has degree $4$ into ${C_{q(t)-\ell-2}}$, so must have degree $0$ into all other classes. Thus, by Lemma \ref{rmk:work-backwards}, $P_1 \cup P_{a-1} \cup P_{b+1}$ must be adjacent to exactly one other class, which we may call $B_{q(t)-\ell-2}$ without loss of generality, and $\deg_{B_{q(t)-\ell-2}}(P_1 \cup P_{a-1} \cup P_{b+1})$ is at most $2$.  
	Degree restrictions on $B_{q(t)-\ell-2}$ imply that it is not a singleton. Therefore $B_{q(t)-\ell-2}$ is either $P_2 \cup P_{a-2} \cup P_{b+2}$ ($\ell>1$) or $P_2 \cup S_{q(t)}$  ($\ell=1$ and $t=1$). 
	By Observations \ref{degree-middle6} and \ref{degree-bottom6}, notice that $\deg(B_{q(t)-\ell-2}) = 4$ only in the case illustrated in Figure \ref{fig:graph-bottom-t=1B}, where $t=1$ and $\deg_{P_d \cup P_c \cup P_{n_\mathcal{P}}}(S\cup P_{c-1})=4$. Note, however, that this would result in a graph that is either $4$-regular or disconnected. Therefore $C_{q(t)-\ell-3}$ must be $V(G)$, with either $B_{q(t)-\ell-2} = P_d \cup P_c \cup P_{n_\mathcal{P}}$, as depicted in Figure \ref{fig:graph-bottom-t=0}, or $S\cup P_{c-1}$, as depicted in Figure \ref{fig:graph-bottom-t=1A}. The resulting graphs are depicted in Figure \ref{fig:d=l-deg24-17} and \ref{fig:d=l-deg24-14}. 
\end{proof}

Now, consider the second case from Lemma \ref{lem:d=l-A,B}, in which $B_{q(t)-\ell} = P_1 \cup P_{a-1} \cup P_{b+1}$, and thus consider the second possible expression for $\pi^{q(t)-\ell-1}$. There are infinitely many graphs which can realise this partition at iteration $q(t)-\ell-1$. We describe them below. 

\begin{lemma}\label{lem:d=l--P1Pa-1PaPbPb+1}
   Assume $d =\ell$, $A_{q(t)-\ell} = P_a \cup P_b$  and $B_{q(t)-\ell} = P_1 \cup P_{a-1} \cup P_{b+1}$. If $\deg(G) = \{2,3\}$, then $G$ is represented by a string in one of the following sets:
   \newpage
    \begin{itemize}
        \item $\{\mathrm{S1^{\mathit{k}}01^{\mathit{k}}1XX1^{\mathit{k}}1_2} \mid k \in \N_0\}$
        \item $\{\mathrm{S1^{\mathit{k}}0011^{\mathit{k}}XX1^{\mathit{k}}10} \mid k \in \N_0\}$ 
        \item $\{\mathrm{S(011)^{\mathit{k}}00(110)^{\mathit{k}}1X0X1(011)^{\mathit{k}}0} \mid k \in \N_0\}$
    \end{itemize}
    If $\deg(G) \neq \{2,3\}$, then $G$ is one of the graphs in Figure \ref{fig:LR-with-d=l} or in Table \ref{tab:adj-list-d=l-34}.
    \begin{figure}[htpb]
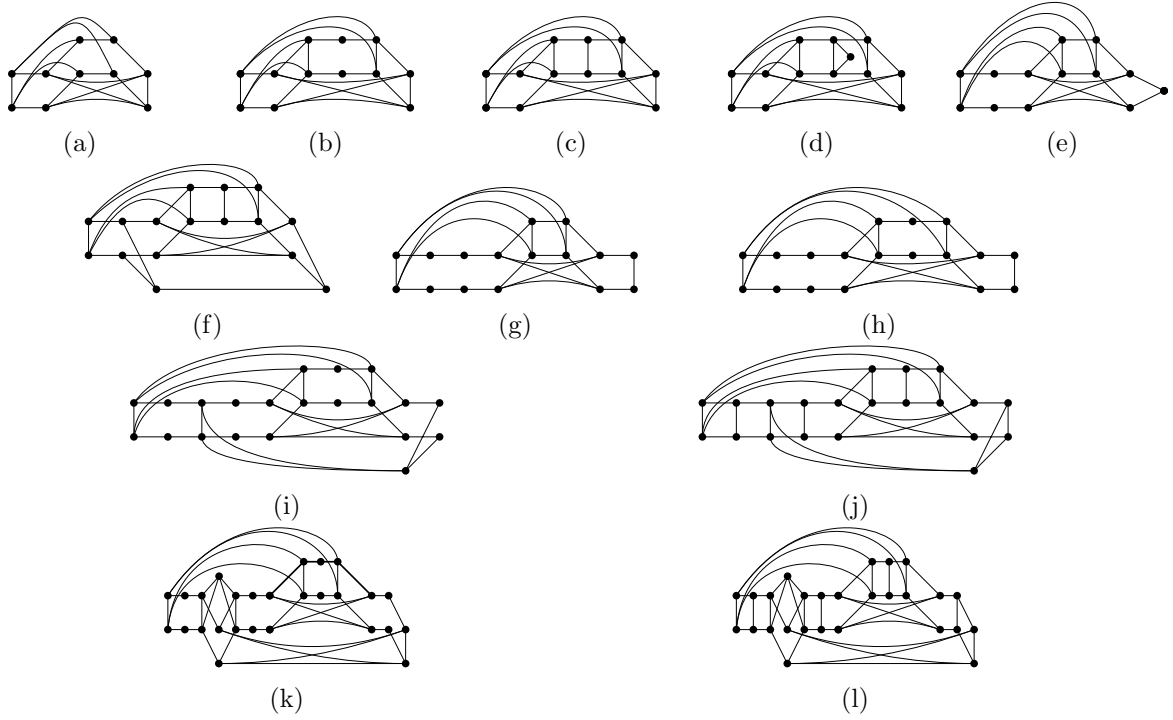

    \centering
        \begin{subfigure}{0.19\linewidth}
            \centering 
            \includegraphics[scale=0.8]{d=l-deg34-10.pdf}
         \caption{}\label{fig:d=l-deg34-10}
         \end{subfigure}
        \begin{subfigure}{0.19\linewidth}
            \centering 
            \includegraphics[scale=0.8]{d=l-deg24-12.pdf}
         \caption{}\label{fig:d=l-deg24-12}
        \end{subfigure}
        \begin{subfigure}{0.19\linewidth}
            \centering 
            \includegraphics[scale=0.8]{d=l-deg34-12.pdf}
         \caption{}\label{fig:d=l-deg34-12}
        \end{subfigure}
        \begin{subfigure}{0.19\linewidth}
            \centering 
            \includegraphics[scale=0.8]{d=l-deg24-13A.pdf}
         \caption{}\label{fig:d=l-deg34-13A}
        \end{subfigure}
        \begin{subfigure}{0.19\linewidth}
            \centering 
            \includegraphics[scale=0.8]{d=l-deg24-13.pdf}
         \caption{}\label{fig:d=l-deg24-13}
        \end{subfigure}
        \begin{subfigure}{0.24\linewidth}
            \centering 
            \includegraphics[scale=0.8]{deg34-16.pdf}
         \caption{}\label{fig:d=l-deg34-16}
        \end{subfigure}
        \begin{subfigure}{0.24\linewidth}
            \centering 
            \includegraphics[scale=0.8]{d=l-deg24-16.pdf}
         \caption{}\label{fig:d=l-deg24-16}
        \end{subfigure}
        \begin{subfigure}{0.33\linewidth}
            \centering 
            \includegraphics[scale=0.8]{d=l-deg24-18.pdf}
         \caption{}\label{fig:d=l-deg24-18}
         \end{subfigure}
        \begin{subfigure}{0.45\linewidth}
            \centering 
            \includegraphics[scale=0.8]{deg24-21.pdf}
         \caption{}\label{fig:d=l-deg24-21}
        \end{subfigure}
        \begin{subfigure}{0.45\linewidth}
            \centering 
            \includegraphics[scale=0.8]{deg34-21.pdf}
         \caption{}\label{fig:d=l-deg34-21}
        \end{subfigure}
        \begin{subfigure}{0.45\linewidth}
            \centering 
            \includegraphics[scale=0.8]{d=l-deg24-25.pdf}
         \caption{}\label{fig:d=l-deg24-27}
        \end{subfigure}
        \begin{subfigure}{0.45\linewidth}
            \centering 
            \includegraphics[scale=0.8]{d=l-deg34-25.pdf}
            \caption{}\label{fig:d=l-deg34-27}
        \end{subfigure}
        \caption{The long-refinement graphs $G$ with $d=\ell$, $B_{q(t)-\ell} = P_1 \cup P_{a-1} \cup P_{b+1}$, and $\deg(G) \neq \{2,3\}$.}
    \label{fig:LR-with-d=l}
\end{figure}

\begin{table}[htpb]
    \centering
\begin{tabular}{c|c|c}
    vertex & adjacency 1 with $n \in \N$ & adjacency 2 with $n \in \N_0$\\ \hline
    0 & $2n+1, 2n+2, 2n+5, 2n+6 $ & $2n+1, 2n+2, 2n+5, 2n+6 $\\
    1 & $3,4n+7, 4n+8$ & $3,4n+7, 4n+8$\\
    2 & $4, 4n+9, 4n+10$&  $4, 4n+9, 4n+10$\\
    odd  $i \in  I$& $i-2,i+1,i+2$& $i-2,i+1,i+2$\\
    even $i \in I$& $i-2,i-1, i+2$& $i-2,i-1, i+2$\\
    $2n+1^*$& $0, 2n-1, 2n+3$& $2n+1, 2n+2, 2n+5$\\
    $2n+2^*$& $0, 2n, 2n+4$& $2n, 2n+1 2n+4$\\
    $2n+3$& $2n+1, 2n+5, 6n+11, 6n+12$&$0, 2n+1, 2n+4, 2n+5$\\
    $2n+4$& $2n+4, 2n+8, 6n+17, 6n+18$&$0, 2n+2, 2n+3, 2n+6$\\
    $2n+5$& $0, 2n+3, 2n+7$&$2n+3,2n+6,2n+7$\\
    $2n+6$& $0, 2n+4, 2n+8$&$2n+4,2n+5,2n+8$\\
    $4n+7$& $1, 4n+5, 4n+9$&$1, 4n+5, 4n+9$\\
    $4n+8$&  $1, 4n+6, 4n+10$& $1, 4n+6, 4n+10$\\
    $4n+9$& $2,4n+7, 4n+11$&$2,4n+7, 4n+11$\\
    $4n+10$& $2, 4n+8, 4n+12$&$2, 4n+8, 4n+12$\\
    $6n+11$& $2n+3, 2n+4, 6n+9$&$0,6n+12,6n+10$\\
    $6n+12$&$2n+3, 2n+4, 6n+10$&$0,6n+11,6n+10$\\
\end{tabular}
    \caption{The two infinite families of long-refinement graphs $G$ with order $6n+13$ for $\deg(G) = \{3,4\}$ where $I= [3,2n]\cup[2n+7,4n+6] \cup [4n+11, 6n+10]$. Note also that $2n+1, 2n+2$ are omitted when $n=0$.} \label{tab:adj-list-d=l-34}
\end{table}

\end{lemma}

\begin{proof}
	The expression for $\pi^{q(t)-\ell}$ from Equation \ref{eq:d=l-pi^q(t)-l} and the assumed values of $A_{q(t)-\ell}$ and $B_{q(t)-\ell}$ yield the following expression for the previous partition.
	\begin{multline*}
		\pi^{q(t)-\ell-1} = \pi^{q(t)}_{\mathcal{S}} \cup \{P_{i}\cup P_{a-i} \cup P_{b+i} \mid i \in [2,\ell-1]\} \\
		\cup \{P_{i} \mid i \in [a+1,b-1]\}\cup 
		\{P_1 \cup P_{a-1} \cup P_{a} \cup P_{b} \cup P_{b+1}\} \cup M_t
	\end{multline*}

    Recall that $\deg_{P_1\cup P_{a-1} \cup P_{b+1}}(P_{a-1}\cup P_{b+1}) -\deg_{P_1\cup P_{a-1} \cup P_{b+1}}(P_1) = 1$. 
    The degrees of $P_1$ and $P_{a-1} \cup P_{b+1}$ into $P_1 \cup P_{a-1} \cup P_{b+1}$ depend on whether $P_{a-1}$ is $P_2$ (which occurs when $\ell=1$ and $t=0$), whether $G[P_a,P_b]$ is empty or complete bipartite, and which among $P_1,P_{a-1},$ and $P_{b+1}$ are in $E(G)$. Careful consideration of all cases yields four layouts for the edges incident to $P_1 \cup P_{a-1}\cup P_{b+1}$ which satisfy the degree difference of $1$. These layouts are depicted in Figure \ref{fig:d=lP1Pa-1Pb+1-edges} with $P_1 \cup P_{a-1} \cup P_{b+1}$ depicted in blue. 
	
	\begin{figure}[H]
		\centering
		\begin{subfigure}{0.24\linewidth}
			\includegraphics[]{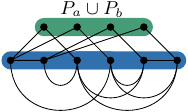}
			\caption{$\ell=1$ and $t=0$}\label{fig:d=lP1Pa-1Pb+1-edges-1}
		\end{subfigure}
		\begin{subfigure}{0.24\linewidth}
			\includegraphics[]{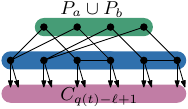}
			\caption{$l=1$ and $t=1$}\label{fig:d=lP1Pa-1Pb+1-edges-4}
		\end{subfigure}
		\begin{subfigure}{0.24\linewidth}
			\includegraphics[]{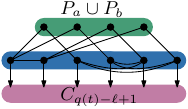}
			\caption{$\ell>1$ or $t=1$}\label{fig:d=lP1Pa-1Pb+1-edges-3}
		\end{subfigure}
		\begin{subfigure}{0.24\linewidth}
			\includegraphics[]{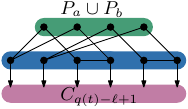}
			\caption{$l>1$ or $t=1$}\label{fig:d=lP1Pa-1Pb+1-edges-2}
		\end{subfigure}

		\caption{Edge layouts satisfying $\deg_{P_1\cup P_{a-1} \cup P_{b+1}}(P_{a-1}\cup P_{b+1}) -\deg_{P_1\cup P_{a-1} \cup P_{b+1}}(P_1) = 1$. Note these depict colour classes as in $\pi^{q(t)-\ell}$ and only the edges incident to $P_1 \cup P_{a-1}\cup P_{b+1}$. The union of the green and blue classes is $C_{q(t)-\ell-1}$. }
		\label{fig:d=lP1Pa-1Pb+1-edges}
	\end{figure} 
	\textbf{Consider Figure~\ref{fig:d=lP1Pa-1Pb+1-edges-1}}, where $\ell=1$ and $t=0$.
	First, let $C_{q(t)-\ell-1} = V(G)$.  Then $\deg(P_a \cup P_b) \neq \deg(P_1 \cup P_{a-1} \cup P_{b+1})$, so $\deg(P_a \cup P_b) \neq 4$. Since $P_a$ and $P_b$ have minimum degree $3$, $G$ must be as depicted in Figure \ref{fig:d=l-deg34-10}. 	
	Next, assume $C_{q(t)-\ell-1} \subsetneq V(G)$. Then $C_{q(t)-\ell-1}$ must be unbalanced wrt exactly two classes $A_{q(t)-\ell-1}$ and $B_{q(t)-\ell-1}$. One of those classes must be $C_{q(t)-\ell-1}$, as it contains all the neighbours of $P_1 \cup P_{a-1}\cup P_{b+1}$ and we say without loss of generality that $A_{q(t)-\ell-1}=C_{q(t)-\ell-1}$. 
	We can observe from Equation \ref{eq:d>l-pi^q(t)-l-1}, taking $\ell=1$ and $t=0$, that the only other classes possible in $\pi^{q(t)-\ell-1}$ are $P_{a+1}$ (if $b=a+2$) and singletons $\pi^{q(t)}_\mathcal{S}$. 
	We first assume that $B_{q(t)-\ell-1}$ is a singleton and prove a contradiction. Since $P_a \cup P_b$ is adjacent to but balanced wrt $B_{q(t)-\ell-1}$, $B_{q(t)-\ell-1}$ must be adjacent to all four vertices of $P_a \cup P_b$. Then all vertices in set $C_{q(t)-\ell-1} \cup B_{q(t)-\ell-1}$ have $4$ neighbours in that same set, making $G$ either disconnected or $4$-regular. Thus $B_{q(t)-\ell-1}$ is $P_{a+1}$ and $C_{q(t)-\ell-2} = P_1 \cup P_{a-1} \cup P_a \cup P_{a+1} \cup P_b \cup P_{b+1}$. Since $\ell=1$ and $t=0$, this contains all the pairs in $\pi^p_\mathcal{P}$.
	If  $\deg (P_{a+1}) \neq \{4\}$, then $V(G) = C_{q(t)-\ell-2}$ and $G$ is as depicted in Figures \ref{fig:d=l-deg24-12}, where $\deg(P_{a+1}) = 2$, and \ref{fig:d=l-deg34-12}, where $\deg(P_{a+1}) = 3$. 
	If $\deg (P_{a+1}) = 4$, then $P_{a+1}$ can be adjacent to exactly one class in $\pi^{q(t)-\ell-2} \backslash \{C_{q(t)-\ell-2}\}$, i.e.\ a singleton  $S\in\pi^{q(t)}_\mathcal{S}$. We then get $C_{q(t)-\ell-3} = S \cup C_{q(t)-\ell-2}$. Since $\deg_{S}(P_{a+1})$ can only be $1$, we also need $P_{a+1} \in E(G)$ to satisfy $\deg (P_{a+1}) = \{4\}$.  Singleton $S$ then has degree $2$ into $P_{a+1}$ and shares no edges with the other pairs. Therefore, $S$ can only have degree $4$ if it is itself adjacent to two other singletons, at which point $C_{q(t)-\ell-3}$ would be unbalanced wrt three classes contradicting Lemma \ref{rmk:work-backwards}. Therefore $C_{q(t)-\ell-2}= V(G)$ and $\deg(P_{a+1})=2$ and $G$ is as depicted in Figure \ref{fig:d=l-deg34-13A}.
	
	\textbf{Next, consider Figure \ref{fig:d=lP1Pa-1Pb+1-edges-4}}. Here, $C_{q(t)-\ell-1} \neq V(G)$, so $\deg(P_a \cup P_b) = \deg(P_1 \cup P_{a-1} \cup P_{b+1})=4$. The vertices of $P_1 \cup P_{a-1}\cup P_{b+1}$ have two neighbours in $C_{q(t)-\ell+1}$, whereas $P_a \cup P_b$ have none by Observations \ref{degree-middle6} and \ref{degree-bottom6}. Thus, $C_{q(t)-\ell+1}$ is in $\pi^{q(t)-\ell-1}\backslash\pi^{q(t)-\ell-2}$ and without loss of generality we may denote it $A_{q(t)-\ell-1}$. Thus, to satisfy Lemma \ref{rmk:work-backwards}, we require the class $B_{q(t)-\ell-1}$ to satisfy $\deg_{B_{q(t)-\ell-1}}(P_a\cup P_b) - \deg_{B_{q(t)-\ell-1}}(P_1 \cup P_{a-1}\cup P_{b+1}) = 2$. Thus $B_{q(t)-\ell-1}$ can be neither a singleton (trivial), nor $P_{a+1}$ (as $\deg_{P_{a+1}}(P_a \cup P_b)=1$), nor any of the classes of size $6$ (Observations \ref{degree-middle6},\ref{degree-bottom6}), and must therefore be $C_{q(t)-\ell-1}$. However, we then have the vertices of set $C_{q(t)-\ell-1} \cup C_{q(t)-\ell+1}$ each with $4$ neighbours in $C_{q(t)-\ell-1} \cup C_{q(t)-\ell+1}$, implying that either $G$ is $4$-regular, or disconnected.
	
	\textbf{Now, consider  Figure \ref{fig:d=lP1Pa-1Pb+1-edges-3}}.  Here, $C_{q(t)-\ell-1} \neq V(G)$, so $\deg(P_a \cup P_b) = \deg(P_1 \cup P_{a-1} \cup P_{b+1})= 4$. 
First, consider the case where $b = a+2$. We then have $\pi^{q(t)-\ell}$ satisfying the assumptions for $\pi^j$ in Lemma \ref{lem:2classgettingbigger}, with $M = P_{a+1}$. Recall that $\ell' = \ell-1$. Then, the results of the aforementioned lemma yield that \[\pi^{q(t)-2\ell-t}\backslash \pi^{q(t)-\ell} = \{D_{\ell-2+t},E_{\ell-2+t}\} \text{ with }P_1 \cup P_{a-1} \cup P_{b+1} = C_{q(t)-\ell} \subseteq E_{\ell-2+t}\]
    We therefore require $\deg(E_{\ell-2+t}) = \deg(P_1 \cup P_{a-1} \cup P_{b+1}) = 4$, which implies that it cannot include any class $P_{d-h} \cup P_{c+h} \cup P_{n_{\mathcal{P}-h}} = C_{q(t)-h-1}$ where $h\in [1, \ell-2]$, by Observations \ref{degree-middle6}. 
    This gives an upper limit to $\ell$, yielding $\ell \in [1,3]$.
    Where $\ell=1$, $t$ must be $1$ as $C_{q(t)-\ell+1}$ exists . By Lemma \ref{lem:2classgettingbigger},  $\deg_{C_{q(t)-\ell+1}}(B_j) = 1$, so we must have edges depicted as in Figure \ref{fig:graph-bottom-t=1A}, and we have $C_{q(t)-\ell-2} = P_{a+1} \cup P_{c-1} \cup S_{q(t)}$. The graph with $\deg(C_{q(t)-\ell-2})=2$ is not a long-refinement graph as $C_{q(t)-\ell-2}$ is then balanced at iteration ${q(t)-\ell-2}$. Degree $\deg(C_{q(t)-\ell-2})=4$ can only be achieved if $G[P_{c-1}\cup S_{q(t)}]$ is complete, $P_{a+1}$ is in $E(G)$, and $P_{a+1}$ is adjacent to a singleton. That singleton, however, would also need to have degree $4$, which is not achievable, given the classes in $\pi^{q(t)-\ell-2}$ without violating Lemma \ref{rmk:work-backwards}. Finally, if $\deg(C_{q(t)-\ell-2})=3$, then $G[P_{c-1}\cup S_{q(t)}]$ is an empty graph but $P_{c-1}\cup S_{q(t)}$ is adjacent to a singleton which we denote $S_{q(t)-\ell-2}$. Degree restrictions on $S_{q(t)-\ell-2}$ imply it is not adjacent to $P_{a+1}$, so either $P_{a+1} \in E(G)$ or $P_{a+1}$ is adjacent to another singleton $S_{q(t)-\ell-2}' \neq S_{q(t)-\ell-2}$. The former option yields the graph depicted in Figure \ref{fig:d=l-deg34-16}, while the latter leads to a contradiction as $C_{q(t)-\ell-3}$ would then be $S_{q(t)-\ell-2}' \cup S_{q(t)-\ell-2}$ implying $\deg(S'_{q(t)-\ell-2})=3$ and thus that it is adjacent to another singleton $S_{q(t)-\ell-3}$. As $\deg(G)=\{3,4\}$, $S_{q(t)-\ell-3}$ cannot have a satisfying degree without violating Lemma \ref{rmk:work-backwards}, given the classes in $\pi^{q(t)-\ell-3}$.
    Where $\ell=2$ and $t=0$, by Lemma \ref{degree-bottom6}, we get the graph depicted in Figure \ref{fig:d=l-deg24-18}. 
    Where $\ell=2$ and $t=1$, we have $C_{q(t)-2\ell-1} = E_{\ell-2+t} \cup P_{c-1} \cup S_{q(t)}$ . Then $\deg(P_{c-1} \cup S_{q(t)}) = \deg(E_{\ell-2+t}) = 4$. This implies $P_{c-1} \cup S_{q(t)}$ has edge layout as depicted in Figure \ref{fig:graph-bottom-t=1B}, or as in Figure \ref{fig:graph-bottom-t=1A} with $G[P_{c-1} \cup S_{q(t)}]$ a complete graph. The former is not possible as both $C_{q(t)-2\ell+1}$ and $D_{\ell-2+t}$ are then unbalanced. The latter leads to the long-refinement graph depicted in Figure \ref{fig:d=l-deg24-21} when $\deg(D_{\ell-2+t})=2$ and Figure \ref{fig:d=l-deg34-21} when $\deg(D_{\ell-2+t})=3$. 
    Finally, consider when $\ell=3$ and therefore $\deg(C_{q(t)-1}) = 4$, where $C_{q(t)-1}= P_d \cup P_c \cup P_{n_\mathcal{P}}$. This can only be satisfied by $t=1$, with edges as depicted in Figure \ref{fig:graph-bottom-t=1B} and $P_d,P_c,P_{n_\mathcal{P}}$. This case yields the graphs depicted in Figure \ref{fig:d=l-deg24-27} and \ref{fig:d=l-deg34-27}, for $\deg(D_{\ell-2+t})$ equal to $2$ or $3$, respectively. 

    Next, let us consider the case where $b=a+1$. The class $P_1 \cup P_{a-1}\cup P_{b+1}$ has degree $1$ into $C_{q(t)-\ell+1}$, whereas $P_a \cup P_b$ has degree $0$ into $C_{q(t)-\ell+1}$. Thus, $C_{q(t)-\ell+1}$ is in $\pi^{q(t)-\ell-1}\backslash\pi^{q(t)-\ell-2}$ and without loss of generality we may denote it $A_{q(t)-\ell-1}$. 
    Then, by Lemma \ref{rmk:work-backwards}, we require the class $B_{q(t)-\ell-1}$ to satisfy $\deg_{B_{q(t)-\ell-1}}(P_a\cup P_b) - \deg_{B_{q(t)-\ell-1}}(P_1 \cup P_{a-1}\cup P_{b+1}) = 1$. This cannot hold if $B_{q(t)-\ell-1}$ is a class of size $6$, by Corollary \ref{cor:second-symmetry-neighbours}, or if it is the class $P_{c-1}\cup S_{q(t)}$, by Lemma \ref{lem:end-neighbours}. 
    It also cannot be a singleton; we would then have $\deg(B_{q(t)-\ell-1}) = 4$, and since $A_{q(t)-\ell-1} \cup B_{q(t)-\ell-1} \subsetneq V(G)$, also $\deg(A_{q(t)-\ell-1}) = \{4\}$, which implies by Remark \ref{degree-middle6} and Observation \ref{degree-bottom6} that $C_{q(t)-\ell+1}$ must be $P_d \cup P_c \cup P_{n_\mathcal{P}}$, with edges as depicted in Figure \ref{fig:graph-bottom-t=1B}. However, we then get a graph that is either $4$-regular or disconnected. 
    Thus, $B_{q(t)-\ell-1}$ is $C_{q(t)-\ell-1}$. To satisfy $\deg(P_a \cup P_b)=4$, we must have $P_a,P_b \in E(G)$ and $b=a+1$. Now, either $\deg({C_{q(t)-\ell+1}}) = 4$ or $V(G) = C_{q(t)-\ell-1}\cup C_{q(t)-\ell+1}$. 
    The former can only be satisfied by ${C_{q(t)-\ell+1}}$ being $P_d \cup P_c \cup P_{n_\mathcal{P}}$ with edges as depicted in Figure \ref{fig:graph-bottom-t=1B}. 
    However, this results in $G$ being $4$-regular or disconnected. Thus, $V(G) = C_{q(t)-\ell-1}\cup C_{q(t)-\ell+1}$, with $C_{q(t)-\ell-1}$ either $P_d \cup P_c \cup P_{n_\mathcal{P}}$ as depicted in Figure \ref{fig:graph-bottom-t=0}, or $P_{c-1} \cup S_{q(t)}$, as in Figure \ref{fig:graph-bottom-t=1A}. 
    These graphs are depicted in Figures \ref{fig:d=l-deg24-13} and \ref{fig:d=l-deg24-16}. 
	
	\textbf{Finally, consider Figure \ref{fig:d=lP1Pa-1Pb+1-edges-2}}.  Here, we have $\deg(P_1 \cup P_{a-1}\cup P_{b+1})= 3$, and therefore also $\deg(P_a \cup P_b) = 3$. The class $P_a \cup P_b$ has degree $2$ into $P_1 \cup P_{a-1}\cup P_{b+1}$, and degree $1$ into either $P_a \cup P_b$ if $b=a+1$, or into $P_{a+1}$ if $b=a+2$.   
    In the former case, iteration $q(t)-\ell$ satisfies the assumptions for iteration $j$ in Lemma \ref{lem:1classgettingbigger}. Thus, taking $j' = \ell'+t = \ell-1+t$, we get $C_{q(t)-2\ell+t} = \bigcup_{i\in[1,d]\cup[c,n_\mathcal{P}]} P_i$. 
    If $t=0$, then $P_d \cup P_c \cup P_{n_\mathcal{P}} \subsetneq C_{q(t)-2\ell-1}$ has degree $2$ by Lemma \ref{degree-bottom6}, while $\deg(\bigcup_{i\in[1,d-1] \cup [c+1,n_\mathcal{P}-1]}P_i) = 3$ and therefore $V(G) = \bigcup_{i=1}^{n_\mathcal{P}}P_i$ and that $G$ is a graph whose representative string is contained in the set $\mathrm{\{S1^k001^k1XX11^k0 \mid k \in \N_0\}}$.  
    If $t=1$, then $C_{q(t)-2\ell-t}\subsetneq V(G)$ so $\deg(P_i)=3$ for $i \in [1,d]\cup [c,n_\mathcal{P}]$. The vertices $P_d \cup P_c \cup P_{n_\mathcal{P}}$ have degree $1$ into $P_{c-1} \cup S_{q(t)}$, where other vertices in $C_{q(t)-2\ell-1}$ have degree $0$, and have all their neighbours within $C_{q(t)-2\ell-1}$. 
    This implies that $\pi^{q(t)-2\ell-2}\backslash\pi^{q(t)-2\ell-1} = \{C_{q(t)-2\ell-1}, P_{c-1}\cup S\}$. If $C_{q(t)-2\ell-2} \neq V(G)$, then $\deg(P_{c-1} \cup S) = 3$, which can only be satisfied by having edges as depicted in Figure \ref{fig:graph-bottom-t=1A} and $P_{c-1} \cup S_{q(t)}$ also adjacent to a singleton $S$. 
    Here, either the singleton has degree $3$ and the graph is $3$-regular, or the singleton has degree greater than $3$ so there are vertices not in $C_{q(t)-2\ell-3} = C_{q(t)-2\ell-2} \cup S$, yet $\deg(C_{q(t)-2\ell-3}) = \bot$.
    We therefore let $C_{q(t)-2\ell-2} = V(G)$, in which case $\deg(P_{c-1} \cup S)$ is either $2$ or $4$. These can be achieved by edges as in \ref{fig:graph-bottom-t=1A}, with $\deg_{P_{c-1}\cup S}(P_{c-1}\cup S)\in \{0,2\}$. 
    We thus get three possible infinite families of graphs, either represented by the set $\mathrm{\{S1^k01^k1XX1^k1_2 \mid k \in \N_0\}}$, or by Table \ref{tab:adj-list-d=l-34}.

	Next, we consider the case where $b = a+2$. Here, iteration $q(t)-\ell$ satisfies the assumptions for iteration $j$ in Lemma~\ref{lem:2classgettingbigger}, with $M = P_{a+1}$, thus yielding $\pi^{q(t)-2\ell-t}\backslash \pi^{q(t)-\ell} = \{D_{\ell-2+t},E_{\ell-2+t}\}$. 
    If $t=0$, then $\deg(P_d \cup P_c \cup P_{n_\mathcal{P}}) = 2$ by Observation \ref{degree-bottom6}, so cannot be a subset of $E_{\ell-2+t}$ with degree $3$. It follows then that $\deg(D_{\ell-2+t}) = 2$, so $P \in E(G)$ for all pairs $P \in \pi^p_\mathcal{P}$ such that $P\subseteq  E_{\ell-2+t}\backslash\{P_{a-1}\cup P_{b+1}\}$, and $P \notin E(G)$ if $P \subseteq D_{\ell-3}\backslash\{P_{n_\mathcal{P}}\}$. We therefore get the graph represented by the set $\mathrm{\{S(011)^k00(110)^k1X0X1(011)^k0 \mid k\in \N_0\}}$. 
    Consider $t=1$ with the layout of edges as depicted in Figure \ref{fig:graph-bottom-t=1A}.  When $C_{q(t)-2\ell-t} = P_{c-1} \cup S_{q(t)} \cup D_{l-2+t}$, the argument that $G$  cannot be a long-refinement graph is the same as that used for Figure \ref{fig:d=lP1Pa-1Pb+1-edges-3}, when $\ell=1$ and $t=1$, except that since $E_{\ell-2+t}$ has degree $3$, the case that yields the long-refinement graph there instead yields a $3$-regular graph here. When $C_{q(t)-2\ell-t} = P_{c-1} \cup S_{q(t)} \cup E_{l-2+t}$ , $\deg(C_{q(t)-2\ell-t})=3$ and hence $P_{c-1}\cup S_{q(t)}$ must be adjacent to a singleton, which we will denote $S_{q(t)-2\ell-t}$.  If $C_{q(t)-2\ell+1-t} = D_{\ell-2+t}$, then $P_{c-1}\cup S_{q(t)}$ contradicts Lemma \ref{rmk:work-backwards} having higher degree into $D_{l-2+t}$ and $S_{q(t)-2\ell-t}$. If $C_{q(t)-2\ell+1-t} = E_{\ell-2+t}$,  then by Lemma \ref{rmk:work-backwards}, $C_{q(t)-2\ell-t-1} =  S_{q(t)-2\ell-t} \cup D_{l-2+t}$, and $\deg(D_{\ell-2+t}) \geq 3$. If $\deg(D_{\ell-2+t}) = 3$, we have a $3$-regular graph, whereas if $\deg(D_{\ell-2+t}) = 4$, then $\deg(P_{a+1})=4$ and must, by a previous argument, be adjacent to a singleton with degree $4$ - which is not possible. 
    Next, consider $t=1$ with the layout of edges depicted in Figure \ref{fig:graph-bottom-t=1B}. Here, $C_{q(t)-2\ell-t}= P_{c-1} \cup S_{q(t)} \cup D_{\ell-2+t}$ since $\deg(E_{\ell-2+t})=3$. Thus $C_{q(t)-l+1} \subseteq D_{\ell-2+t}$ can only be $P_d \cup P_c \cup P_{n_\mathcal{P}}$ or $P_{c-1}\cup S_{q(t)}$. The latter contradicts Lemma \ref{lem:2classgettingbigger}, the former contradicts Lemma~\ref{rmk:work-backwards}.
\end{proof}

Having identified the graphs that correspond to both cases from Lemma \ref{lem:d=l-A,B}, this concludes the case that $d = \ell$.  Now, since the remaining cases have $d >\ell$, we henceforth have $\ell' = \ell$ and thus Lemma \ref{lem:second-symmetry} yields 
\begin{multline}\label{eq:d>l-pi^q(t)-l-1}
	 \pi^{q(t)-\ell-1} = 
	 \pi^{q(t)}_{\mathcal{S}}
	 \cup \{P_{d-i}\cup P_{c+i} \cup P_{n_\mathcal{P}-i} \mid i \in [0,\ell]\}  
	 \\\cup\{P_i \mid i \in [d-\ell-1]\cup[a+1,b-1]\} \cup M_t
 \end{multline}

We now consider the cases where $d = \ell+1$ and $d =\ell+2$ simultaneously. In Lemma \ref{lem:l+1l+2AB}, similarly to Lemma \ref{lem:d=l-A,B}, we identify all the classes that may be $A_{q(t)-\ell-1}$ or $B_{q(t)-\ell-1}$, so as to determine all possible expression of $\pi^{q(t)-\ell-2}$ and treat them separately.  

\begin{lemma}\label{lem:l+1l+2AB}
     If $d\in \{\ell+1,\ell+2\}$, then $\{A_{q(t)-\ell-1},B_{q(t)-\ell-1}\} \subseteq \{P_{d-\ell} \cup P_a \cup P_b, P_{a+1}\}\cup \pi^{q(t)}_\mathcal{S}$.
\end{lemma}

\begin{proof}
	Observe the classes in Equation \ref{eq:d>l-pi^q(t)-l-1}.
	By Remark \ref{degree-middle6}, $C_{q(t)-\ell-1}$ is balanced wrt all classes in $\{P_{d-i}\cup P_{c+i} \cup P_{n_\mathcal{P}-i} \mid i \in [0,\ell-1]\}$. Degree restrictions on $C_{q(t)-\ell-1}$ and $S\cup P_{c-1}$ give the same result for $S \cup P_{c-1}$ if $t=1$. When $d = \ell+2$, Lemmas \ref{lem:successive-matching} and \ref{lem:P_1-P_aP_b} imply $C_{q(t)-\ell-1} = P_2 \cup P_a \cup P_b$ is balanced wrt $P_1$. Thus, $C_{q(t)-\ell-1}$  may only be unbalanced wrt the classes $P_{d-\ell} \cup P_a \cup P_b$ and $P_{a+1}$, and the singletons in $\pi^{q(t)}_\mathcal{S}$. 
\end{proof}

This result allows us to divide the search for long-refinement graphs where $d \in \{\ell+1,\ell+2\}$ into the following cases: $C_{q(t)-\ell-1}$ is either unbalanced wrt by a) two singletons,\, b) $P_{a+1}$ and a singleton,\, c) $P_{d-\ell} \cup P_a \cup P_b$ and a singleton adjacent to $P_{d-\ell}$,\, d) $P_{d-\ell} \cup P_a \cup P_b$ and a singleton adjacent to $P_{a}\cup P_b$, and\, e) $P_{a+1}$ and $P_{d-\ell} \cup P_a \cup P_b$. We treat these cases in the next five lemmas. We determine that the first two cases cannot be satisfied by any long-refinement graphs.

\begin{lemma}\label{lem:l+1l+2ABsingletons}
    Let $d\in \{\ell+1,\ell+2\}$. There is no long-refinement graph with $A_{q(t)-\ell-1},B_{q(t)-\ell-1} \in \pi^{q(t)}_\mathcal{S}$. 
\end{lemma}

\begin{proof}
From Equation \ref{eq:d>l-pi^q(t)-l-1} and the values of $A_{q(t)-\ell-1}, B_{q(t)-\ell-1}$, we get the following.
\begin{multline*}
\pi^{q(t)-\ell-2} = 
\pi^{q(t)}_{\mathcal{S}}\setminus\{A_{q(t)-\ell-1},B_{q(t)-\ell-1}\}
\cup \{A_{q(t)-\ell-1} \cup B_{q(t)-\ell-1}\}
\\\cup \{P_{d-i}\cup P_{c+i} \cup P_{n_\mathcal{P}-i} \mid i \in [0,\ell]\}  
\cup\{P_i \mid i \in [d-\ell-1] \cup [a+1,b-1]\} \cup M_t
\end{multline*}
Since $|A_{q(t)-\ell-1}|=|B_{q(t)-\ell-1}|=1$ and $C_{q(t)-\ell-1}$ is unbalanced wrt both classes, we must have $P_{d-\ell}$ adjacent to one singleton, and $P_{a}\cup P_b$ adjacent to the other. 
Without loss of generality, let $\deg_{A_{q(t)-\ell-1}}(P_{d-\ell}) = 1$ and $\deg_{B_{q(t)-\ell-1}}(P_a \cup P_b)=1$. 
Then $\deg_{P_{d-\ell} \cup P_a \cup P_b}(A_{q(t)-\ell-1}) = 2$ and $\deg_{P_{d-\ell} \cup P_a \cup P_b}(B_{q(t)-\ell-1}) = 4$, and without loss of generality $A_{q(t)-\ell-2} = P_{d-l}\cup P_a \cup P_b$. 
By degree restrictions, $B_{q(t)-\ell-1}$ has degree $0$ into $B_{q(t)-\ell-2}$, so $\deg_{B_{q(t)-\ell-2}}(A_{q(t)-\ell-1})=2$. 
All classes in $\pi^{q(t)-\ell}$ are balanced wrt -- and since $A_{q(t)-\ell-1}$ is a singleton, trivially connected to -- $A_{q(t)-\ell-1} \in \pi^{q(t)-\ell-1}$. 
The class $B_{q(t)-\ell-2}$ must therefore be a pair, so $d = \ell+2$ and $B_{q(t)-\ell-2} = P_1$. 
From this and by Lemma \ref{rmk:work-backwards}, we get $C_{q(t)-\ell-3} = P_1 \cup P_2 \cup P_a \cup P_b$. With $P_1$ adjacent to singleton $A_{q(t)-\ell-1}$, it has degree $4$, so $P_2 \cup P_a \cup P_b$ must also have degree $4$. 
If $\ell = 0$ and $t=0$, the graph would be $4$-regular or disconnected, so $\ell > 0$ or $t=1$ must hold. Thus $C_{q(t)-\ell}$ is either $P_3 \cup P_{a-1}\cup P_{b+1}$, or $P_3 \cup S_{q(t)}$.
In either case, class $C_{q(t)-\ell-3}$ is unbalanced wrt itself and $C_{q(t)-\ell}$, so by Lemma \ref{rmk:work-backwards}, $C_{q(t)-\ell-4} = C_{q(t)-\ell-3} \cup C_{q(t)-\ell}$ and thus  $C_{q(t)-\ell}$ must also have degree $4$. 
This can only occur if $C_{q(t)-\ell}$ is $\{P_d, P_c, P_{n_\mathcal{P}}\}$ or $M_t$, where the layout of edges is as depicted in Figure \ref{fig:graph-bottom-t=1B}. 
Again, in this case, the graph is $4$-regular or disconnected.
\end{proof}

\begin{lemma}\label{lem:Pa+1sing}
        Let $d\in \{\ell+1,\ell+2\}$ and $A_{q(t)-\ell-1} = P_{a+1}$ and $B_{q(t)-\ell-1} \in \pi^{q(t)}_\mathcal{S}$. Then $G$ is depicted in Figure \ref{fig:d=l+2-deg34-15}.
\end{lemma}

\begin{figure}
    \centering    \includegraphics{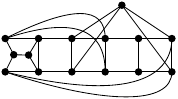}
    \caption{The long-refinement graph satisfying the prerequisites of Lemma \ref{lem:Pa+1sing}.}
    \label{fig:d=l+2-deg34-15}
\end{figure}

\begin{proof}
	The expression for $\pi^{q(t)-\ell-1}$ from Equation \ref{eq:d>l-pi^q(t)-l-1} and the assumed values of $A_{q(t)-\ell-1}$ and $B_{q(t)-\ell-1}$ yield the following expression for $\pi^{q(t)-\ell-2}$.
	\begin{multline}\label{eq:d=l+1l+2ABsingletons-pi^q(t)-l-2}
		\pi^{q(t)-\ell-2} = 
		\pi^{q(t)}_\mathcal{S}\backslash\{B_{q(t)-\ell-1}\}
		\cup \{P_{d-i}\cup P_{c+i} \cup P_{n_\mathcal{P}-i} \mid i \in [0,\ell]\}  \\
		\cup \{P_{a+1} \cup B_{q(t)-\ell-1}\}
		\cup\{P_i \mid i \in [d-\ell-1]\} \cup M_t
	\end{multline}
	Clearly $C_{q(t)-\ell-2}= P_{a+1} \cup B_{q(t)-\ell-1}$ is a strict subset of $V(G)$, so we require $C_{q(t)-\ell-2}$ to be unbalanced wrt two classes $A_{q(t)-\ell-2}$ and $B_{q(t)-\ell-2}$. We will show that there is no pair of classes in $\pi^{q(t)-\ell-2}$ that satisfy the requirements of $A_{q(t)-\ell-2}$ and $B_{q(t)-\ell-2}$.
	
	First, consider the classes in $\{P_{d-i}\cup P_{c+i} \cup P_{n_\mathcal{P}-i} \mid i \in [0,\ell]\}$. Observe that $P_{d-\ell}\cup P_a\cup P_b$ ($i=\ell$) cannot be $A_{q(t)-\ell-2}$ or $B_{q(t)-\ell-2}$ since $\deg_{P_{d-\ell}\cup P_a\cup P_b}(C_{q(t)-\ell-2})=2$. 
    Next, by Remark \ref{degree-middle6} and Observation \ref{degree-bottom6}, there are no edges between $C_{q(t)-\ell-2}$ and any class $P_{d-i} \cup P_{c+i} \cup P_{n_\mathcal{P}-i}$ where $i \in [0,\ell-1]$, so such classes are not $A_{q(t)-\ell-2}$ or $B_{q(t)-\ell-2}$. 
	By Lemma \ref{lem:end-neighbours} and degree restrictions on $A_{q(t)-\ell-1}$ and $B_{q(t)-\ell-1}$, $S_{q(t)} \cup P_{c-1} \notin \{A_{q(t)-\ell-2},B_{q(t)-\ell-2}\}$. 
	
	Observing Equation \ref{eq:d=l+1l+2ABsingletons-pi^q(t)-l-2}, this leaves us with \[\{A_{q(t)-\ell-2},B_{q(t)-\ell-2}\} \subseteq \pi^{q(t)}_\mathcal{S}\backslash\{B_{q(t)-\ell-1}\}\cup \{P_{a+1} \cup B_{q(t)-\ell-1}\}\}\cup\{P_i \mid i \in [d-\ell-1]\}.\]
    Let us assume that $P_1 =A_{q(t)-\ell-2}$ ($d=\ell+2$) and prove a contradiction. In this case, $\deg_{P_1}(B_{q(t)-\ell-1})- \deg_{P_1}(P_{a+1})$ must be $2$, by Lemmas \ref{rmk:work-backwards}, \ref{lem:non-successive-pairs}, and degree restrictions on $P_1$. However, no other class in $\pi^{q(t)}_\mathcal{S}\backslash\{B_{q(t)-\ell-1}\}\cup \{P_{a+1} \cup B_{q(t)-\ell-1}\}$ can then satisfy Lemma \ref{rmk:work-backwards}: trivially, singleton classes are not large enough, while $\deg_{P_{a+1} \cup B_{q(t)-\ell-1}}(P_{a+1}) \in \{0,1\}$ by degree restrictions on $B_{q(t)-\ell-1}$.
    By this and the symmetric argument for $B_{q(t)-\ell-2}$, $P_1 \notin \{A_{q(t)-\ell-2},B_{q(t)-\ell-2}\}$. Thus, either $A_{q(t)-\ell-2} \cup B_{q(t)-\ell-2}$ is the union of two singletons $S_{q(t)-\ell-2}, S'_{q(t)-\ell-2}$, or of a singleton $S_{q(t)-\ell-2}$ and $P_{a+1} \cup B_{q(t)-\ell-1}$, when $P_{a+1} \in E(G)$. 

    Both cases, within a few simple applications of Lemma \ref{rmk:work-backwards}, lead to the class $C_{q(t)-\ell-4}$ in the first case, and $C_{q(t)-\ell-3}$ in the second case, being $P_1 \cup P_2 \cup P_a \cup P_b$ with $\deg(P_1 \cup P_2 \cup P_a \cup P_b)= 4$ and therefore $d=\ell+2$. 
    This requires the previous iteration to contain, by Lemma \ref{rmk:work-backwards}, either $P_1 \cup P_2 \cup P_a \cup P_b \cup M_t$ with $M_t$ as depicted in Figure \ref{fig:graph-bottom-t=1A} with $G[M_t]$ a complete graph, or $P_1 \cup P_2 \cup P_3 \cup P_4 \cup P_a \cup P_b \cup P_{b+1}$, with all vertices of degree $4$, therefore $M_t$ as depicted in Figure \ref{fig:graph-bottom-t=1B}. The latter leads to contradictions, either leading to a $4$-regular graph, or leading to a class with two degrees after iteration $1$. This statement holds for the former case when both $A_{q(t)-\ell-2}$ and $B_{q(t)-\ell-2}$ are singletons. However, when $A_{q(t)-\ell-2}$ or $B_{q(t)-\ell-2}$ is $P_{a+1}$, we can have the graph depicted in Figure \ref{fig:d=l+2-deg34-15}.
\end{proof}

Now, we consider the cases where $A_{q(t)-\ell-1}= P_{d-\ell} \cup P_a \cup P_b$ and $B_{q(t)-\ell-1} \in \pi^p_\mathcal{S}$. Here we distinguish cases according to which of $A_{q(t)-\ell}$ or $B_{q(t)-\ell}$, i.e.\ which of $P_{d-\ell}$ and $P_a \cup P_b$, the singleton $B_{q(t)-\ell-1}$ is adjacent to. Although both cases have the same expression for $\pi^{q(t)-\ell-2}$, they require dissimilar proofs. 

\begin{lemma}\label{lem:d=l+1-singleton-adjXX}
    Let $d\in \{\ell+1,\ell+2\}$, $A_{q(t)-\ell-1} = P_{d-\ell} \cup P_a \cup P_b,$ and $B_{q(t)-\ell-1} \in \pi^{q(t)}_\mathcal{S}$ with $B_{q(t)-\ell-1}$ adjacent to $P_a \cup P_b$. Then $G$ is depicted in Figure \ref{fig:d=l+1l+2-singleton-adj-ab}.
\end{lemma}

\begin{figure}[htpb]
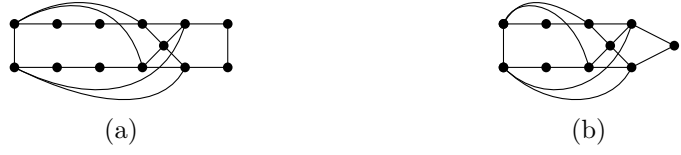

    \centering
    \begin{subfigure}{0.24\linewidth}
            \centering 
            \includegraphics{d=l+1-singleton-adjXX-1.pdf}
         \caption{}\label{fig:d=l+1-singleton-adjXX-1}
    \end{subfigure}
    \hspace{2cm}
    \begin{subfigure}{0.24\linewidth}
            \centering 
            \includegraphics{d=l+1-singleton-adjXX-2.pdf}
         \caption{}\label{fig:d=l+1-singleton-adjXX-2}
         \end{subfigure}
    \caption{The two graphs satisfying the prerequisites from Lemma \ref{lem:d=l+1-singleton-adjXX}.}
    \label{fig:d=l+1l+2-singleton-adj-ab}
\end{figure}

\begin{proof}
	We know by the fact that $P_{a+1}\notin \{A_{q(t)-\ell-1},B_{q(t)-\ell-1}\}$, that $b=a+1$. Let us first assume that $d = \ell+2$. 
    Substituting $2$ for $d-\ell$, we then get that $\deg_{P_{2} \cup P_a \cup P_b}(P_2) > \deg_{P_{2} \cup P_a \cup P_b}(P_a \cup P_b)$. However, no possible layout of edges satisfies this requirement without violating degree restrictions or Lemma \ref{lem:first-symmetry-neighbours}. 
    We must therefore have $d=\ell+1$. There is a layout of edges that satisfies this, in which $P_1 \in E(G)$ and $P_a, P_b \notin E(G)$.
	
The expression for $\pi^{q(t)-\ell-1}$ from Equation \ref{eq:d>l-pi^q(t)-l-1} and the values of $A_{q(t)-\ell-1}$ and $B_{q(t)-\ell-1}$ yield 
\begin{multline*}
 \pi^{q(t)-\ell-2} = 
\pi^{q(t)}_{\mathcal{S}}
\cup \{P_{d-i}\cup P_{c+i} \cup P_{n_\mathcal{P}-i} \mid i \in [0,\ell]\}  
\\\cup\{P_i \mid i \in [d-\ell-1]\cup[a+1,b-1]\} \cup M_t
\end{multline*}

By assumption, $\deg_{B_{q(t)-\ell-1}}(P_a \cup P_b) = 1$ and thus also $\deg_{P_a \cup P_b}(B_{q(t)-\ell-1}) = 4$. Therefore, since $C_{q(t)-\ell-2} \subsetneq V(G)$ (it must omit at least $P_2$) then $\deg(P_{d-\ell} \cup P_a \cup P_b)= 4$. Note that, now, the class that includes pair $P_2$ must either be the only other class, or must also have degree $4$. From Lemma \ref{degree-bottom6} and Remark \ref{degree-middle6}, it is clear that, if the latter is true, $G$ is either $4$-regular or disconnected. Thus, the class containing $P_2$ is the only other class at iteration $q(t)-\ell-2$ and can be either $P_2 \cup S_{q(t)}$ or $P_2 \cup P_{3} \cup P_5$. This yields the two graphs in Figure \ref{fig:d=l+1l+2-singleton-adj-ab}.
\end{proof}

\begin{lemma}
    Let $d\in \{\ell+1,\ell+2\}$. There is no long-refinement graph satisfying $A_{q(t)-\ell-1} = P_{d-\ell} \cup P_a \cup P_b$ and $B_{q(t)-\ell-1} \in \pi^{q(t)}_\mathcal{S}$ with $B_{q(t)-\ell-1}$ adjacent to $P_{d-\ell}$. 
\end{lemma}
\begin{proof}
From the expression for $\pi^{q(t)-\ell-1}$ in Equation \ref{eq:d>l-pi^q(t)-l-1} and the given values of $A_{q(t)-\ell-1}$ and $B_{q(t)-\ell-1}$, we get the expression
\begin{multline}
    \pi^{q(t)-\ell-2} = 
	 \pi^{q(t)}_{\mathcal{S}}\backslash\{B_{q(t)-\ell-1}\}
	 \cup \{P_{d-i}\cup P_{c+i} \cup P_{n_\mathcal{P}-i} \mid i \in [0,\ell-1]\}  
	 \\\cup\{ P_{d-\ell} \cup P_a \cup P_b \cup B_{q(t)-\ell-2}\} \cup\{P_i \mid i \in [d-\ell-1]\cup[a+1,b-1]\} \cup M_t.
\end{multline}

First, assume that $d = \ell+1$. Then, recalling that $a > 2$, we must have either $\ell >0$ or $t=1$, implying that $C_{q(t)-\ell} \in \{P_2 \cup P_{a-1} \cup P_{b+1},\ P_2 \cup S\}$ is a class in $\pi^{q(t)-\ell-2}$. As the singleton $B_{q(t)-\ell-1}$ is only adjacent to $P_{d-\ell}$, it has degree $2$ into $C_{q(t)-\ell-1}$. Meanwhile, $A_{q(t)-\ell-1}$ has degree $3$ into $C_{q(t)-\ell-2}$ and $1$ into $C_{q(t)-\ell}$. However, by Remark \ref{degree-middle6}, $B_{q(t)-\ell-1}$ shares no edges with $C_{q(t)-\ell}$. This implies that $C_{q(t)-\ell-2}$ is unbalanced by both $C_{q(t)-\ell-2}$ and $C_{q(t)-\ell}$, with $A_{q(t)-\ell-1}$ having higher degree into both, contradicting Lemma \ref{rmk:work-backwards}. 

Letting $d =\ell+2$, $A_{q(t)-\ell-1}$ is $P_2 \cup P_a \cup P_b$. Consequently, $C_{q(t)-\ell-2}$ is unbalanced by $P_1$ as $\deg_{P_1}(P_2 \cup P_a \cup P_b) = 1$ by Lemmas \ref{lem:successive-matching} and \ref{lem:P_1-P_aP_b}, whereas $B_{q(t)-\ell-1}$ has degree $0$ into $P_1$ by Lemma \ref{lem:non-successive-pairs} and degree restrictions. If either $\ell > 0$ or $t=1$, then $C_{q(t)-\ell}$ is another class in $\pi^{q(t)-\ell-2}$ which by Corollary \ref{cor:second-symmetry-neighbours} and Remark \ref{degree-middle6}, $P_2 \cup P_a \cup P_b$ is adjacent to and $B_{q(t)-\ell-1}$ is not. This contradicts Lemma \ref{rmk:work-backwards}. Thus, we can only have $\ell=0$ and $t=0$, in which case $\pi^{q(t)-\ell-2} = \pi^{q(t)}_{\mathcal{S}}\backslash\{B_{q(t)-\ell-2}\} \cup \{P_1, P_2 \cup P_a \cup P_b \cup B_{q(t)-\ell-1}\}$. We can deduce that $\deg (P_2 \cup P_a \cup P_b \cup B_{q(t)-\ell-1})=3$, since $P_2\cup P_a\cup P_b$ could not be adjacent to any singleton without contradicting the degree restrictions. Thus, we must have $B_{q(t)-\ell-1}$ adjacent to a singleton to satisfy $\deg(B_{q(t)-\ell-1})=3$ and thus $C_{q(t)-\ell-2}$ must be unbalanced wrt some singleton $S'$. Thus $C_{q(t)-\ell-2} = P_1 \cup S'$ with $C_{q(t)-\ell-2}$ unbalanced wrt $C_{q(t)-\ell-1}$ with $P_1$ having degree $3$ into $C_{q(t)-\ell-1}$ and $S'$ having degree $0$. There are no other classes that can satisfy a degree difference of $2$, so we have a final contradiction. 
\end{proof}

We now treat the final case for $d \in \{\ell+1,\ell+2\}$.

\begin{lemma}
    Let $d\in \{\ell+1,\ell+2\}$, $A_{q(t)-\ell-1} = P_{d-\ell} \cup P_a \cup P_b$, and $B_{q(t)-\ell-1} = P_{a+1}$. If $\deg(G) = \{2,3\}$, then $G$ is represented by a long-refinement string in one of the following sets
    \begin{itemize}
        \item $\mathrm{\{S1X1_2X\}}$
        \item $\{\mathrm{S(011)^{\mathit{k}}00(110)^{\mathit{k}}X1_2X(011)^{\mathit{k}}0} \mid k \in \N_0\}$
        \item $\mathrm{\{S1^{\mathit{k}}001^{\mathit{k}}X1X1^{\mathit{k}}0} \mid k \in \N_0\}$
        \item $\mathrm{\{S0X1X_2\}} \cup \{\mathrm{S1^{\mathit{k}}1011^{\mathit{k}}X1X1^{\mathit{k}}1_2}  \mid k \in \N_0\}$
    \end{itemize}
If $\deg(G) \neq \{2,3\}$, then $G$ is as represented in Figure \ref{fig:d=l+1l+2-graphs} or is isomorphic to a graph in Table \ref{tab:adj-list2-d=l-34}.
\end{lemma}
\begin{figure}[htpb]
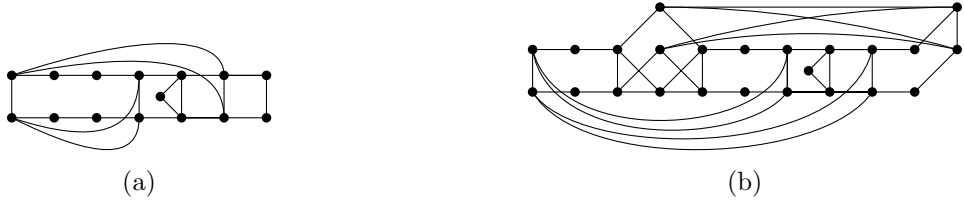

    \centering
\begin{subfigure}{0.48\linewidth}
    \centering
    \includegraphics{d=l+1-graphs1.pdf}
    \caption{}\label{fig:S+00X+1_2+X+0}\label{fig:deg24-14a}
\end{subfigure}
\begin{subfigure}{0.48\linewidth}
    \centering
    \includegraphics{d=l+1-graphs2.pdf}
    \caption{}\label{fig:deg24-24a}
\end{subfigure}
    \caption{The long-refinement graphs $G$ with $d \in \{\ell+1,\ell+2\}$, where $A_{q(t)-\ell-1}=P_{d-\ell}\cup P_a \cup P_b$, $B_{q(t)-\ell-1}=P_{a+1}$, and $\deg(G) \neq \{2,3\}$.}
    \label{fig:d=l+1l+2-graphs}
\end{figure}

\begin{table}[tpb]
    \centering
\begin{tabular}{c|c|c}
    vertex & adjacency 1 & adjacency 2 \\ \hline
    0 & $2k+3, 2k+4, 2k+7, 2k+8 $&$2k+5,2k+6,6k+13,6k+14$\\
    1 & $3,4k+9, 4k+10$&$3,4k+9,4k+10$ \\
    2 & $4, 4k+11, 4k+12$&$4, 4k+11, 4k+12$\\
    odd  $i \in I$ & $i-2,i+1,i+2$& $i-2,i+1,i+2$\\
    even $i \in I$ & $i-2,i-1, i+2$& $i-2,i-1, i+2$\\
    $2k+3$ & $0, 2k+1, 2k+5$&$2k+1,2k+4,2k+5$\\
    $2k+4$ & $0, 2k+2, 2k+6$&$2k+2,2k+3,2k+6$\\
    $2k+5$& $2k+3, 2k+7, 6k+13, 6k+14$&$0,2k+3,2k+6,2k+7$\\
    $2k+6$& $2k+4, 2k+8, 6k+13, 6k+14$&$0,2k+4,2k+5,2k+8$\\
    $2k+7$& $0, 2k+5, 2k+9$&$2k+5,2k+8,2k+9$\\
    $2k+8$& $0, 2k+6, 2k+10$&$2k+6,2k+7,2k+10$\\
    $4k+9$& $1, 4k+7, 4k+11$&$1, 4k+7, 4k+11$\\
    $4k+10$&  $1, 4k+8, 4k+12$&$1, 4k+8, 4k+12$\\
    $4k+11$& $2,4k+9, 4k+13$&$2,4k+9, 4k+13$\\
    $4k+12$& $2, 4k+10, 4k+14$&$2, 4k+10, 4k+14$\\
    $6k+13$& $2k+5, 2k+6, 6k+11$&$0, 6k+14, 6k+11$\\
    $6k+14$&$2k+5, 2k+6, 6k+12$&$0, 6k+13, 6k+12$\\
\end{tabular}
    \caption{All infinite families of long-refinement graphs $G$ with order $6k+15$ and $\deg(G) = \{3,4\}$, where $I= [3,2k+2]\cup[2k+9,4k+8] \cup [4k+13, 6k+12]$, defined for $k\in \N_0$.}\label{tab:adj-list2-d=l-34}
\end{table}

\begin{proof}
Given the expression for $\pi^{q(t)-\ell-1}$ in Equation \ref{eq:d>l-pi^q(t)-l-1} and the assumed values of $A_{q(t)-\ell-1}$ and $B_{q(t)-\ell-1}$, we get the following expression for $\pi^{q(t)-\ell-2}$.
\begin{multline}\label{eq:d=l+34-pi^q(t)-l-2}
	 \pi^{q(t)-\ell-2} = 
	 \pi^{q(t)}_{\mathcal{S}} \cup \{P_{d-\ell}\cup P_a\cup P_{a+1}\cup P_{b}\}
	 \cup \{P_{d-i}\cup P_{c+i} \cup P_{n_\mathcal{P}-i} \mid i \in [0,\ell-1]\}  
	 \\\cup\{P_i \mid i \in [d-\ell-1]\} \cup M_t.
 \end{multline}
Let us start by considering $d = \ell+1$. Then $C_{q(t)-\ell-2}$ is $P_1 \cup P_a \cup P_{a+1} \cup P_b$ with $\deg(C_{q(t)-\ell-2})\geq 3$. Since $P_2 \prec P_a$, we must have either $\ell > 0$ or $t = 1$. 
If $\ell=0$ and $t=1$ with edges as depicted in Figure \ref{fig:graph-bottom-t=1B}, then $\deg_{C_{q(t)-\ell}}(P_1 \cup P_a \cup P_b)- \deg_{C_{q(t)-\ell}}(P_{a+1}) = 2$. However, the other classes in $\pi^{q(t)-\ell-2}$ (the singletons in $\pi^{q(t)}_\mathcal{S}$, $P_1 \cup P_a \cup P_{a+1} \cup P_b$) cannot satisfy Lemma \ref{rmk:work-backwards} to complete $\{A_{q(t)-\ell-2},B_{q(t)-\ell-2}\}$.
In all cases where $l>0$, $\deg_{C_{q(t)-\ell}}(P_1 \cup P_a \cup P_b)- \deg_{C_{q(t)-\ell}}(P_{a+1}) = 1$, by Lemmas \ref{degree-middle6} and \ref{degree-bottom6}. Without loss of generality, let $A_{q(t)-\ell-2} = C_{q(t)-\ell}$. To satisfy Lemma \ref{rmk:work-backwards}
$B_{q(t)-\ell-2}$ must satisfy $\deg_{B_{q(t)-\ell-2}}(P_1 \cup P_a \cup P_b)- \deg_{B_{q(t)-\ell-2}}(P_{a+1}) = -1$. Thus, by \ref{degree-middle6} and \ref{degree-bottom6}, $B_{q(t)-\ell-2}$ cannot be a class of size $6$, nor $M_t$ as there is no layout of edges that satisfies requirements. 
Thus $B_{q(t)-\ell-2}$ may only be $P_1 \cup P_a \cup P_{a+1} \cup P_b$, or in $\pi^{q(t)}_\mathcal{S}$. 

Let us first assume $B_{q(t)-\ell-2}$ is $P_1 \cup P_a \cup P_{a+1} \cup P_b$. Then $\deg(P_{a+1}) = 3$, as in this case $P_{a+1}$ is not adjacent to a singleton. Furthermore, $P_{a+1}$ has all its neighbours in its own class, which means that the partition $\pi^{q(t)-\ell-2}$ satisfies the requirements of partition $\pi^j$, which means one can apply Lemma \ref{lem:1classgettingbigger} to yield $C_{j-1-j'} = C_{q(t)-2\ell-3-t}= A_j \cup \left(\bigcup_{i\in [0,\ell+t]}C_{q(t)-\ell-1+i}\right)$. 
If $t=0$, then by the result from Remark \ref{degree-bottom6} that $\deg(P_d\cup P_c\cup P_{n_\mathcal{P}}) = 2$, $C_{q(t)-2\ell-3-t}=V(G)$ so we have exactly the graphs represented by the strings in the set $\mathrm{\{S1^{\mathit{k}}001^{\mathit{k}}X1X1^{\mathit{k}}0} \mid k \in \N_0\}$.
If $t=1$, then either $C_{q(t)-2\ell-3-t}=V(G)$ and $\deg(P_{c-1}\cup S_{q(t)}) \in \{2,4\}$, or $C_{q(t)-2\ell-3-t}\subsetneq V(G)$ and $\deg(P_{c-1}\cup S_{q(t)}) = 3$. By an argument identical to that used in the proof Lemma \ref{lem:d=l--P1Pa-1PaPbPb+1}, the former case yields exactly three infinite families of graphs -- one represented by the strings in $\mathrm{\{S1^{\mathit{k}}1011^{\mathit{k}}X1X1^{\mathit{k}}1_2} \mid k \in \N_0\}$, and two represented in Table \ref{tab:adj-list2-d=l-34} -- while the latter case yields none.

We return to the case where $B_{q(t)-\ell-2} \in \pi^{q(t)}_\mathcal{S}$. Note that, in this case, $\pi^{q(t)-\ell-1}$ matches the requirements for partition $j$ in Lemma \ref{lem:2classgettingbigger} with $M$ as the singleton $B_{q(t)-\ell-2}$. Note that $\deg B_{q(t)-\ell-2} = 2$, as it cannot be adjacent to the larger classes, by degree restrictions, and any singleton adjacent to $B_{q(t)-\ell-2}$ could not then have degree greater than $2$, leading to a contradiction to Lemma \ref{rmk:work-backwards} at the previous iteration. Now, $\deg P_{a+1}$ can be $3$ or $4$ depending on whether $P_{a+1}$ is in $E(G)$. Where $\deg P_{a+1} = 4$, we can have at most $\ell = 2$ by Remark \ref{degree-middle6}. We then get the graph in Figure \ref{fig:deg24-14a} when $\ell=1$ and the graph in Figure \ref{fig:deg24-24a} when $\ell=2$. When $\deg P_{a+1} = 3$, we obtain the graphs represented by the strings $\{\mathrm{S(011)^{\mathit{k}}00(110)^{\mathit{k}}X1_2X(011)^{\mathit{k}}0} \mid k \in \N_0\}$.

Now, consider the case with $d = \ell+2$ and recall the partition $\pi^{q(t)-\ell-2}$ (Equation \ref{eq:d=l+34-pi^q(t)-l-2}). The class $P_1$ is adjacent to $P_2 \cup P_a \cup P_b$, but not to $P_{a+1}$, by Lemma \ref{lem:non-successive-pairs} and degree restrictions on $P_1$. By Lemma \ref{rmk:work-backwards}, the same cannot be true of either $P_3 \cup P_{a-1} \cup P_{b+1}$ or $P_3 \cup S$, yet neither class can be adjacent to $P_{a+1}$ by Lemmas \ref{degree-middle6} and \ref{degree-bottom6}, respectively, so neither class exists and thus $\ell=0$ and $t=0$. Taking, without loss of generality, $A_{q(t)-\ell-2} = P_1$, we then need $\deg_{B_{q(t)-\ell-2}}(P_{a+1})- \deg_{B_{q(t)-\ell-2}}(P_2 \cup P_a \cup P_b) = 1$. This may be satisfied by any of the other classes in $\pi^{q(t)-\ell-2}$, $B_{q(t)-\ell-2} = P_2 \cup P_a \cup P_{a+1} \cup P_b$ or $B_{q(t)-\ell-2} \in \pi^{q(t)}_\mathcal{S}$, but the former results in a $3$-regular or disconnected graph. The latter yields the long-refinement graph $G(\mathrm{S1X1_2X})$. 
\end{proof}

Having addressed every case from Lemma \ref{lem:l+1l+2AB}, this concludes the case for both $d = \ell+1$ and $d = \ell+2$. We now consider the remaining cases for the value of $d$. As per Lemma \ref{lem:value-of-d}, those are $d \in \{\ell+3,\ell+4\}$, in which case $b=a+1$. 

\begin{figure}[tpb]
    \centering
        \includegraphics{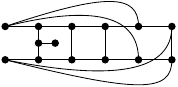}
    \caption{The unique long-refinement graph with $\deg(G)=\{1,3\}$.}
    \label{fig:d=l+3l+4-graphs}
\end{figure}

\begin{lemma}
    Let $d \in \{\ell+3, \ell+4\}$. Then $G$ is either the graph depicted in Figure \ref{fig:d=l+3l+4-graphs} or isomorphic to a graph from Table \ref{tab:adj-list3-d=l-34} or 
    represented by a string in one of the following sets.
    \begin{itemize}
        \item $\{\mathrm{S}111^k001^k\mathrm{X}\mathrm{X}1^k0 \mid k \in \N_0\}$
        \item $\mathrm{\{S110XX_2\}}\cup\mathrm{\{S111^{\mathit{k}}1011^{\mathit{k}}XX1^{\mathit{k}}1_2} \mid k \in \N_0\}$
        \item $\mathrm{\{S011(011)^{\mathit{k}}00(110)^{\mathit{k}}XX(011)^{\mathit{k}}0} \mid k \in \N_0\}$
        \item $\mathrm{\{S11XX, S011XX, S1_211XX\}}$
    \end{itemize}
\end{lemma}

\begin{table}[tpb]
    \centering
\begin{tabular}{c|c|c}
    vertex  &  adjacency 1  &  adjacency 2 \\ \hline
    0  &  $2k+7, 2k+8, 2k+11,2k+12$  & $2k+5,2k+6,6k+17,6k+18$\\
    1  &  $3,4k+13, 4k+14$  &  $3,4k+13, 4k+14$\\
    2  &  $4, 4k+15, 4k+16$  &  $4, 4k+15, 4k+16$\\
    odd  $i \in I$  &  $i-2,i+1,i+2$ &  $i-2,i+1,i+2$\\
    even $i \in I$  &  $i-2,i-1, i+2$ &  $i-2,i-1, i+2$\\
    $2k+7$  &  $0, 2k+5, 2k+9$ &  $2k+5,2k+8,2k+9$\\
    $2k+8$  &  $0, 2k+6, 2k+10$ & $2k+6,2k+7,2k+10$\\
    $2k+9$ &  $2k+7, 2k+11, 6k+17, 6k+18$ & $0, 2k+7, 2k+10, 2k+11$\\
    $2k+10$ &  $2k+8, 2k+12, 6k+17,6k+18$ & $0, 2k+8 ,2k+9 ,2k+12$\\
    $2k+11$ &  $0, 2k+9, 2k+13$ & $2k+9, 2k+12, 2k+13$\\
    $2k+12$ &  $0, 2k+10, 2k+14$ & $2k+10,2k+11,2k+14$\\
    $4k+13$ &  $1, 4k+11, 4k+15$ & $1, 4k+11, 4k+15$\\
    $4k+14$ &   $1, 4k+12, 4k+16$ & $1, 4k+12, 4k+16$\\
    $4k+15$ &  $2,4k+13, 4k+17$ & $2, 4k+13, 4k+17$\\
    $4k+16$ &  $2, 4k+14, 4k+18$ & $2, 4k+14, 4k+18$\\
    $6k+17$ &  $2k+9, 2k+10, 6k+15$ & $0, 6k+18, 6k+15$\\
    $6k+18$ & $2k+9, 2k+10, 6k+16$ & $0, 6k+17, 6k+16$\\
\end{tabular}
    \caption{All infinite families of long-refinement graphs $G$ with order $6k+19$ and $\deg(G) = \{3,4\}$, where $I= [3,2k+6]\cup[2k+13,4k+12] \cup [4k+17, 6k+16]$.}\label{tab:adj-list3-d=l-34}
\end{table}

\begin{proof}
We are given that $d \in \{\ell+3,\ell+4\}$ and therefore have, from Equation \ref{eq:d>l-pi^q(t)-l-1}, an expression for $\pi^{q(t)-\ell-1}$. Since $d>\ell+2$, the pairs $P_{d-\ell-1}$ and $P_1$ are distinct. By Lemmas \ref{lem:successive-matching} and \ref{lem:P_1-P_aP_b}, $\deg_{P_{d-\ell-1}}(P_{d-\ell}) = 1 = \deg_{P_1}(P_a \cup P_b)$, whereas $\deg_{P_1}(P_{d-\ell}) \in \{0,2\}$ and $\deg_{P_{d-\ell-1}}(P_{d-\ell})\in \{0,2\}$ by Lemma \ref{lem:non-successive-pairs}. Thus, by Lemma \ref{rmk:work-backwards}, we must have $C_{q(t)-\ell-2} = P_1 \cup P_{d-\ell-1}$ and
\begin{multline}\label{eq:d=l+3,4-pi^q(t)-l-2}
	 \pi^{q(t)-\ell-2} = 
	 \pi^{q(t)}_{\mathcal{S}}
	 \cup \{P_{d-i}\cup P_{c+i} \cup P_{n_\mathcal{P}-i} \mid i \in [0,\ell]\} \cup \{P_1 \cup P_{d-\ell-1}\}  
	 \cup\{P_i \mid i \in [2,d-\ell-2]\} \cup M_t.
 \end{multline}
 By Lemma \ref{lem:non-successive-pairs} and degree restrictions, we have $\deg_{P_1}(P_{d-\ell}) = \deg_{P_{d-\ell-1}}(P_{d-\ell})=0$. 
 
Now, we examine the neighbourhoods of $P_{d-\ell-1}$ and $P_1$ and their classes at iteration $q(t)-\ell-2$. 
Notice that $\deg_{P_{d-\ell}\cup P_a \cup P_b}(P_1) = 2$ and $\deg_{P_{d-\ell}\cup P_a \cup P_b}(P_{d-\ell-2}) = 1$ and so $P_{d-\ell}\cup P_a \cup P_b \in \pi^{q(t)-\ell-2}\backslash\pi^{q(t)-\ell-3}$ and without loss of generality we can denote it $A_{q(t)-\ell-2}$. Then we need another class $B_{q(t)-\ell-2}$ such that $\deg_{B_{q(t)-\ell-2}}(P_{d-\ell-1})- \deg_{B_{q(t)-\ell-2}}(P_{1}) = 1$. 
By Lemma \ref{lem:successive-matching}, Remarks \ref{degree-middle6} and \ref{degree-bottom6}, and Lemma \ref{lem:end-neighbours} respectively, we can discard $P_2$ (when $d = \ell+4$), the classes of six, and where $t=1$ the class $C_{q(t)}$. We are left with $B_{q(t)-\ell-2} \in \pi^{q(t)}_\mathcal{S}$ or $B_{q(t)-\ell-2} = P_1 \cup P_{d-\ell-1}$. 
We show that the former leads to a contradiction, since it would lead to $C_{q(t)-\ell-3} = P_{d-\ell}\cup P_a \cup P_b \cup S$ for some singleton $S$. We get a contradiction when $\ell+t>0$, as $P_{d-\ell}\cup P_a \cup P_b$ has more neighbours in both $P_{d-\ell}\cup P_a \cup P_b \cup S$ and $C_{q(t)-\ell}$ than $S$, contradicting Lemma \ref{rmk:work-backwards}. 
We get a contradiction when $\ell+t = 0$ because then $\deg_{P_1 \cup P_{d-\ell}}(S) - \deg_{P_1 \cup P_{d-\ell}}(P_{d-\ell}\cup P_a \cup P_b)=1$ whereas $\deg_{P_1 \cup P_{d-\ell}}(S) - \deg_{P_1 \cup P_{d-\ell}}(P_{d-\ell}\cup P_a \cup P_b)=-2$, contradicting Lemma \ref{rmk:work-backwards}. 
As such, we must have $C_{q(t)-\ell-3} = C_{q(t)-\ell-2} \cup C_{q(t)-\ell-1}$, with possible edge layouts within $C_{q(t)-\ell-3}$ depicted in Figure \ref{fig:d=l+3l+4-C_q(t)-l-3}.
\begin{figure}[htpb]
    \centering
    \begin{subfigure}{0.4\linewidth}
    \centering
        \includegraphics{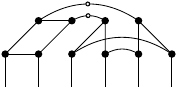}
        \label{fig:d=l+3l+3=l+t>0}
        \caption{Edge layout when $\ell>0$ or $t=1$}
    \end{subfigure}
    \begin{subfigure}{0.4\linewidth}
    \centering
        \includegraphics{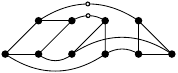}
        \label{fig:d=l+3l+3=l+t=0}
        \caption{Edge layout when $\ell=0$ and $t=0$}
    \end{subfigure}
    \caption{Diagram showing vertices $C_{q(t)-\ell-3}$ and all edges incident to it. Note the empty disks representing $P_2$ only when $d=\ell+4$.}
    \label{fig:d=l+3l+4-C_q(t)-l-3}
\end{figure}
We work backward from each of the two possible expression for partition $\pi^{q(t)-\ell-3}$ to determine all long-refinement graphs in which $d \in \{\ell+3, \ell+4\}$.
    We know that $C_{q(t)-\ell-3} = P_1 \cup P_{d-\ell-1} \cup P_{d-\ell} \cup P_a \cup P_b$ with (up to swapping $A$ and $B$) $A_{q(t)-\ell-2}=  P_1 \cup P_{d-\ell-1}$ and $B_{q(t)-\ell-2} = P_{d-\ell} \cup P_a \cup P_b$. Let us consider the cases according to the value of $t+\ell$, which determines whether $P_{d-\ell}$ and $P_a$ are subsequent pairs. 

    If  $t+\ell = 0$ and $d = \ell+3$, then $C_{q(t)-\ell-3}$ includes all the pairs in $\mathcal{P}$. $B_{q(t)-\ell-2}$ cannot be adjacent to any singletons, as we would then get a vertex with degree at least $6$.  Thus, either $C_{q(t)-\ell-3} = V(G)$ and the graph is $3-$regular, or $C_{q(t)-\ell-3} \subsetneq V(G)$ and $\deg(C_{q(t)-\ell-3}=2$. Either is a contradiction.

    If $t+\ell = 0$ and $d=\ell+4$, then $C_{q(t)-\ell-3}$ includes all the pairs in $\mathcal{P}$, except $P_2$. As $P_1 \cup P_{d-\ell}$ has degree $1$ to $P_2$ whereas the remaining vertices in $C_{q(t)-\ell-3}$ have all their neighbours in $C_{q(t)-\ell-3}$, then $C_{q(t)-\ell-4} = C_{q(t)-\ell-3} \cup P_2 = \bigcup_{P\in\pi^p_\mathcal{P}} P$, by Lemma \ref{rmk:work-backwards}. 
    If $C_{q(t)-\ell-4} = V(G)$, then $\deg (P_2) \neq 3$ and thus we get $G(\mathrm{S011XX})$. If $C_{q(t)-\ell-4} \subsetneq V(G)$, then $\deg(P_2) = 3$, so $P_2$ must have degree $1$ into some singleton $S$ and by the same argument as above, we get $C_{q(t)-\ell-5} = \bigcup_{P\in\pi^p_\mathcal{P}} P \cup S$. If $C_{q(t)-\ell-5} = V(G)$, then $\deg(S) \neq 3$, which is satisfied by $G(\mathrm{S1_211XX})$. If $C_{q(t)-\ell-5} \subsetneq V(G)$, then $\deg(S)=3$, so it must itself be adjacent to a singleton, which we denote $S_{q(t)-\ell-5}$. We cannot have $\deg(S_{q(t)-\ell-5})=3$ unless there are two singletons in $\pi^{q(t)-\ell-6}$ adjacent to $S_{q(t)-\ell-5}$ - neither of which could be adjacent to $\bigcup_{P\in \pi^p_\mathcal{P}}P$, contradicting Lemma \ref{rmk:work-backwards}. Thus, $C_{q(t)-\ell-5} = V(G)$ and the last long-refinement graph in this case is that depicted in Figure \ref{fig:d=l+3l+4-graphs}.

    If $d = \ell+3$ and $t+\ell>0$, then the partition $\pi^{q(t)-\ell-2}$ satisfies the requirements for $\pi^j$ in Lemma \ref{lem:1classgettingbigger}, since $A_{q(t)-\ell-2}$ has all its vertices in $C_{q(t)-\ell-3}$. Applying that lemma gives us that $C_{q(t)-2\ell-3-t}$ is a superset of all the classes of $C_{q(t)-i}$ for $i \in [0,\ell+1]$ with $A_{q(t)-2\ell-2-t}= C_{q(t)-1+t}$. If $t=0$ then by Remark \ref{degree-bottom6}, $A_{q(t)-2\ell-2-t}= C_{q(t)-1+t}$ must have degree $2$ yielding only the graphs represented by strings in the set $\mathrm{\{S111^{\mathit{k}}001^{\mathit{k}}XX1^{\mathit{k}}0} \mid k \in \N_0\}$. If $t=1$, the edges between the vertices of $C_{q(t)}$ and $C_{q(t)-1}$ can be either as in Figure \ref{fig:graph-bottom-t=1A} or as in \ref{fig:graph-bottom-t=1B}. In the former case, we see that if $V(G) = C_{q(t)-2\ell-3-t}$, then $\deg (P_{c-1}\cup S_{q(t)})$ can be $2$ -- which results in the graphs represented by strings $\mathrm{\{S1111^{\mathit{k}}011^{\mathit{k}}XX1^{\mathit{k}}1_2} \mid k \in \N_0\}$ -- or $3$ -- only if it is adjacent to a singleton $S$ with $\deg (S) = 3$ which, without contradicting Lemma \ref{rmk:work-backwards}, yields only $3$-regular or disconnected graphs -- or $4$ -- which yields a graph in the right column of Table \ref{tab:adj-list3-d=l-34}. In the latter case, the vertices of $P_{c-1}\cup S_{q(t)}$ have degree $4$, so $C_{q(t)-2\ell-3-t}=V(G)$ and we obtain the graphs described in the left column of Table \ref{tab:adj-list3-d=l-34}.  
    
    If $d = \ell+4$ and $t+\ell>0$, then we can apply Lemma \ref{lem:2classgettingbigger}, with $j=q(t)-\ell-2$, to show that $\pi^{q(t)-2\ell-3-t} \backslash \pi^{q(t)-\ell-2} = \{D_{\ell-1+t},E_{\ell-1+t}\}$. 

    If $D_{\ell-1+t}\cup E_{\ell-1+t} = V(G)$, then $\deg (D_{\ell-1+t}) \neq \deg (E_{\ell-1+t}) = 3$. However, we cannot achieve $\deg (D_{\ell-1+t}) = 4$ without $P_2$ adjacent to some singleton, which would contradict  $D_{\ell-1+t}\cup E_{\ell-1+t} = V(G)$. Thus $\deg (D_{\ell-1+t}) = 2$. If $t=0$, then $\deg (P_d \cup P_c \cup P_{n_p}) = 2$. Thus $C_{q(t)-1} \subsetneq D_{\ell-1+t}$, and therefore we narrow down to the graphs represented by strings in the set $\mathrm{\{S011(011)^{\mathit{k}}00(110)^{\mathit{k}}XX(011)^{\mathit{k}}0} \mid k \in \N_0\}$. If $t=1$, we must again have $\deg (D_{\ell-1+t})\in \{2,4\}$. The latter, once again, cannot be achieved without a singleton not in $D_{\ell-1+t}\cup E_{\ell-1+t}$. The former leads to a stable graph.

    If $D_{\ell-1+t}\cup E_{\ell-1+t} \subsetneq V(G)$, then $t$ must be $1$ as otherwise, as previously argued $C_{q(t)-2\ell-3-t} = D_{\ell-1+t}$ and $N(D_{\ell-1+t})\subseteq D_{\ell-1+t}\cup E_{\ell-1+t}$, so $G$ is disconnected. With $t=1$, we could have a singleton $S$ adjacent to $P_{c-1}\cup S_{q(t)}$, or where $P_2 \in \{A_{q(t)-2\ell-2-t}, B_{q(t)-2\ell-3-t}\}$, adjacent to $P_2$. The former case can only occur when $t=1$, as in Figure \ref{fig:graph-bottom-t=1A}, with $P_{c-1}\cup S_{q(t)}$ adjacent to a singleton. Without contradicting Lemma \ref{rmk:work-backwards}, this can only give us $3-$regular or disconnected graphs. The latter case is also quickly dismissed: either $\deg (P_2)$ is $3$ -- so $C_{q(t)-1+t}$ must be $P_{c-1}\cup S_{q(t)}$ adjacent to a singleton $S'$ yielding $C_{q(t)-2\ell-3-t}=S\cup S'$ which causes contradictions within two applications of Lemma  \ref{rmk:work-backwards} -- or it is $4$ -- and $\deg S = 4$, which is impossible considering the classes in partition $\pi^{q(t)-2\ell-3-t}$, without contradicting Lemma \ref{rmk:work-backwards}.
\end{proof}

This concludes our case analysis of long-refinement graphs with maximum degree at most $4$. 

\section{The Classification}\label{sec:classification}

We now bring together the results that we have found in the Section \ref{sec:max4} for our final classification. 

Recall Notation \ref{notation:23}. The work \cite{KMcK20} identified families of long-refinement graphs $G$ with $\deg(G) = \{2,3\}$, but left as an open problem whether they are comprehensive. Here, we answer their question affirmatively. 

\begin{theorem}\label{thm:classification23}
    The even-order long-refinement graphs $G$ with $\deg(G) = \{2,3\}$ are exactly the ones represented by strings contained in the following sets:
	\begin{itemize}
		\setlength\itemsep{0em}
		\item $\{\mathrm{S011XX}\}$
		\item $\{\mathrm{S1^{\mathit{k}}001^{\mathit{k}}X1X1^{\mathit{k}}0} \mid k \in \N_0\}$
		\item $\{\mathrm{S1^{\mathit{k}}11001^{\mathit{k}}XX1^{\mathit{k}}0} \mid k \in \N_0\}$
		\item $\{\mathrm{S1^{\mathit{k}}0011^{\mathit{k}}XX1^{\mathit{k}}10} \mid k \in \N_0\}$
		\item $\{\mathrm{S011(011)^{\mathit{k}}00(110)^{\mathit{k}}XX(011)^{\mathit{k}}0} \mid k \in \N_0\}$
		\item $\{\mathrm{S(011)^{\mathit{k}}00(110)^{\mathit{k}}1X0X1(011)^{\mathit{k}}0} \mid k \in \N_0\}$
	\end{itemize}

    The odd-order long-refinement graphs $G$ with $\deg(G) = \{2,3\}$ are exactly the ones represented by strings contained in the following sets:
\begin{itemize}
	\setlength\itemsep{0em}
        \item $\{\mathrm{S1_211XX}\}$
	\item $\{\mathrm{S0X1X_2}\} \cup \{\mathrm{S1^{\mathit{k}}1011^{\mathit{k}}X1X1^{\mathit{k}}1_2} \mid k \in \N_0\}$
	\item $\{\mathrm{S110XX_2}\} \cup \{\mathrm{S111^{\mathit{k}}1011^{\mathit{k}}XX1^{\mathit{k}}1_2} \mid k \in \N_0\}$
	\item $\{\mathrm{S1^{\mathit{k}}01^{\mathit{k}}1XX1^{\mathit{k}}1_2} \mid k \in \N_0\}$
	\item $\{\mathrm{S(011)^{\mathit{k}}00(110)^{\mathit{k}}X1_2X(011)^{\mathit{k}}0} \mid k \in \N_0\}$
\end{itemize}
\end{theorem}

In particular, the theorem implies that the gaps that the families found in \cite{kiefer:phd} leave open cannot be covered with degrees $2$ and $3$.

\begin{corollary}\label{no-LRgraphs-n}
There is no long-refinement graph $G$ with $\deg(G) = \{2,3\}$ and order $n \neq 12$ such that $n \bmod 18\in\{6,12\}$.
\end{corollary}

Theorem \ref{thm:classification23} and Corollary \ref{no-LRgraphs-n} completely classify the long-refinement graphs $G$ with $\deg(G)=\{2,3\}$. Expanding to $\max\deg(G) \leq 4$, we obtain the following further parts of our classification of long-refinement graphs with small degrees from the results in section \ref{sec:max4}.

\begin{theorem}\label{thm:deg-12-or-14}
    There is no long-refinement graph $G$ with $\deg(G) = \{1,2\}$ or $\deg(G) = \{1,4\}$. 
\end{theorem}
\begin{theorem}\label{thm:deg-24}
    The long-refinement graphs $G$ with $\deg(G) = \{2,4\}$ are precisely the ones in Figures \ref{fig:d=l-deg24-17}, \ref{fig:d=l-deg24-14}, \ref{fig:d=l-deg24-12}, \ref{fig:d=l-deg24-13}, \ref{fig:d=l-deg24-16},  \ref{fig:d=l-deg24-18}, \ref{fig:d=l-deg24-21}, \ref{fig:d=l-deg24-27}, \ref{fig:d=l+1-singleton-adjXX-1},
    \ref{fig:d=l+1-singleton-adjXX-2},
 \ref{fig:deg24-14a}, and \ref{fig:deg24-24a}. 
\end{theorem}

\begin{theorem}\label{thm:deg-34}
    The long-refinement graphs $G$ with $\deg(G) = \{3,4\}$ are precisely the ones in Figures \ref{fig:d=l-deg34-11}, \ref{fig:d=l-deg34-10},  \ref{fig:d=l-deg34-12}, \ref{fig:d=l-deg34-13A}, \ref{fig:d=l-deg34-16}, \ref{fig:d=l-deg34-21}, \ref{fig:d=l-deg34-27}, 
    \ref{fig:d=l+2-deg34-15}
    as well as Tables \ref{tab:adj-list-d=l-34}, \ref{tab:adj-list2-d=l-34}, and \ref{tab:adj-list3-d=l-34}.
\end{theorem}

\begin{theorem}\label{thm:classification3}
    The graph in Figure \ref{fig:d=l+3l+4-graphs} is the only long-refinement graph $G$ with $\deg (G) = \{1,3\}$. 
\end{theorem}

While Theorems \ref{thm:deg-12-or-14} -- \ref{thm:classification3} are proved through an analysis of the refinement iterations in Section \ref{sec:max4}, we now present a simple alternative proof for Theorem \ref{thm:classification3}, which relies merely on Theorem \ref{thm:classification23}.

By \cite[Lemma 13]{KMcK20}, a long-refinement graph $G$ with $\deg(G) = \{1,3\}$ has at most two vertices of degree $1$. We start by treating the case where $G$ has exactly one vertex with degree $1$.

\begin{lemma}\label{lem:deg-13-singleton}
    The graph in Figure \ref{fig:d=l+3l+4-graphs} is the only long-refinement graph $G$ with $\deg(G) = \{1,3\}$ in which $|\{v \in V(G) : \deg(v) = 1\}| = 1$.
\end{lemma}

\begin{proof}
     Let $G$ be any long-refinement graph with $\deg(G)=\{1,3\}$ and the set of vertices of degree $1$ being a singleton $\{v_1\}$. Then the partition after the first Colour Refinement iteration on $G$ is $\pi_G^1 = \{\{v_1\}, V(G) \setminus \{v_1\}\}$. The graph $G' \coloneqq G[V(G) \setminus \{v_1\}]$, i.e.\ the subgraph of $G$ induced by $V(G) \setminus \{v_1\}$, satisfies $\deg(G') = \{2,3\}$ and $\{v \in V(G') : \deg(v) = 2\} = \{v_2\}$, where $\{v_2\} = N_G(v_1)$. Also, Colour Refinement takes $n-2$ iterations to terminate on $G'$, otherwise $G$ would not be a long-refinement graph. Hence, $G'$ is a long-refinement graph as well, and since $|V(G)|$ is even, $|V(G')|$ is odd. By our classification in Theorem \ref{thm:classification23}, $G'$ must be $G(\mathrm{S}1_211\mathrm{XX})$, since that is the only long-refinement graph with degrees $2$ and $3$ in which there is only one vertex of degree $2$. So, by degree arguments, $G$ must be the graph in Figure \ref{fig:d=l+3l+4-graphs}.
\end{proof}

Now we consider the situation in which there are exactly two vertices of degree $1$.

\begin{lemma}\label{lem:deg-13-pair}
    There is no long-refinement graph $G$ with $\deg(G) = \{1,3\}$ in which $|\{v \in V(G) : \deg(v) = 1\}| = 2$. 
\end{lemma}

\begin{proof}
     Suppose $G$ is a long-refinement graph with $\deg(G) = \{1,3\}$ for which $\{v \in V(G) : \deg(v) = 1\} = \{v_1, v'_1\}$ with $v_1 \neq v'_1$. We know that $\{v_1,v'_1\} \notin E(G)$, since $G$ is connected. However, now the graph $G' \coloneqq (V(G), E(G) \cup \{\{v_1,v'_1\}\}$ is a long-refinement graph with $\deg(G') = \{2,3\}$ and $\{v \in V(G') : \deg(v) = 2\} = \{v_1, v'_1\}$. The degree property is clear. To see that $G'$ is a long-refinement graph, note that inserting the new edge does not alter the partition into degrees, i.e.\ $\pi^1$. Also, since in both graphs, it holds that $|N(\{v_1,v'_1\})| \leq 2$, the class $\{v_1,v'_1\}$ must be the last class that splits (in both graphs). Furthermore, $N_G(v_1)\setminus\{v'_1\} = N_{G'}(v_1)\setminus\{v'_1\}$ and $N_G(v'_1)\setminus\{v_1\} = N_{G'}(v'_1)\setminus\{v_1\}$. This implies that $\pi_G^i = \pi_{G'}^i$ for every $i \in \N$, so $G'$ is a long-refinement graph if and only if $G$ is one. By double counting edges, $G$ (and hence also $G'$) has an even number of vertices and hence $G'$ is one of the graphs from the first list in Theorem \ref{thm:classification23}. The only graph from that list with exactly two vertices of degree $2$ is $G(\mathrm{S011XX})$, but these vertices are not adjacent. Hence, $G'$ does not exist, and therefore $G$ does not exist.
\end{proof}

This concludes the proof of Theorem \ref{thm:classification3}, which together with Theorem \ref{thm:deg-12-or-14} yields that the graph in Figure \ref{fig:d=l+3l+4-graphs} is the only long-refinement graph $G$ with $\min\deg (G) = 1$ and $\max\deg(G) \leq 4$.

\section{Distinguishability of Long-Refinement Graphs}\label{sec:long-ref-distinguish}

To conclude and to address another open question from \cite{KMcK20}, we consider the situation when Colour Refinement is applied to two graphs $G$ and $H$ in parallel, for example to check whether they are isomorphic. Formally, Colour Refinement is applied to the graph consisting of the disjoint union $U$ of $G$ and $H$ (ensuring that their vertex sets are disjoint). Recall that $\chi^i_U$ denotes that colouring that Colour Refinement has computed on $U$ after $i$ iterations. We say that Colour Refinement \emph{distinguishes} $G$ and $H$ in $i$ iterations if $\doubleBrackets{\chi^i_U(v) \mid v \in V(G)} \neq \doubleBrackets{\chi^i_U(v) \mid v \in V(H)}$. If there is no such $i$, we say that $G$ and $H$ are \emph{equivalent} with respect to Colour Refinement.

\begin{theorem}\label{thm:nolongdist}
    For $n \geq 3$, there is no pair of graphs $G$, $H$ with $|H| \leq |G| = n$ that Colour Refinement distinguishes only after $n - 1$ iterations.
\end{theorem}

\begin{proof}
 Graphs of different orders are distinguished within one iteration of Colour Refinement, so any counterexample to the theorem would have to consist of graphs $G$, $H$ with $|G| = |H|$ and both graphs would be long-refinement graphs. In particular, Colour Refinement computes the discrete partition on both graphs and, hence, the last colour class that splits in each graph is a pair $P_G$ and $P_H$, respectively. Let $U$ be the disjoint union of $G$ and $H$. Since $P_G$ and $P_H$ must have the same colour in iteration $n-2 \geq 1$, it holds that $P_G \in E(G)$ if and only if $P_H \in E(H)$. Let $\{v_1,v_2\} \coloneqq P_G$. Since in $\pi_U^{n-2}$, all vertices in $V(G) \setminus P_G$ and in $V(H) \setminus P_H$ form singleton colour classes, there has to be a vertex $w_1 \in P_H$ that is adjacent to exactly the same colour classes from iteration $n-2$ as $v_1$, and similarly, the second vertex $w_2 \in P_H$ must be adjacent to exactly the same singleton classes as $v_2$. However, then for each $C \in \pi_U^{n-2}$ and for $i \in \{1,2\}$, it holds that $\deg_C(v_i) = \deg_C(w_i)$. But then $\chi_U^{n-1}(v_i) = \chi_U^{n-1}(w_i)$. Thus, $G$ and $H$ are equivalent with respect to Colour Refinement and in particular not distinguished after $n-1$ iterations.
\end{proof}

The theorem gives a negative answer to the question whether there are pairs of graphs that are only distinguished after the maximum possible number of iterations of Colour Refinement. The above argument might be pushed a little further to exclude other strong lower bounds on the number of iterations for distinguishability. Note, however, that it may still be possible that there exists a constant $c$ such that for all $n \in \N$ (from some start value on), there are pairs of graphs $G$ and $H$ with $|G| = |H| = n$ which are only distinguished in $n-c$ iterations of Colour Refinement.

\section{Conclusion}

Long-refinement graphs are graphs which take the maximum theoretically possible number of Colour Refinement iterations to stabilise. Even though the Colour Refinement procedure is simple and its power and mechanics are well-understood (see, for example, \cite{kiefer:phd}), the fact that long-refinement graphs actually exist is surprising, given that this contrasts the situation for higher dimensions of the Weisfeiler--Leman algorithm. The original approaches to show their existence were empirical, i.e.\ to test candidate graphs \cite{Goedicke,KMcK20}. A theoretical analysis of the results then yielded infinite families of long-refinement graphs with small degrees, but also left infinitely many gaps with respect to graph sizes.

In this work, we have shown that for minimal degrees, those gaps cannot be closed. More precisely, we have proved that Theorems 18 and 24 from \cite{KMcK20} are actually a classification, meaning those graphs are \emph{the only long-refinement graphs with degrees $\{2,3\}$ that exist}, for arbitrarily large $n$. We have extended this to the first classification of long-refinement graphs of maximum degree $4$. While previous searches for long-refinement graphs started from computational data, our reverse-engineering method is a purely theoretically-motivated approach. Our analysis shows that \emph{all} families of long-refinement graphs with minimal degrees (i.e.\ with maximum degree at most $3$) can be captured with string representations. The characterisation that falls out of this insight yields a simple proof that there is only one long-refinement graph with minimum degree $1$ and maximum degree at most $4$. 

As witnessed in Lemma \ref{lem:complement}, the graph obtained from a long-refinement graph by taking the complement of the edge set is again a long-refinement graph. Thus, our classification of long-refinement graphs with maximum degree at most $4$ also immediately yields a classification of long-refinement graphs with minimum degree at least $n-5$. Extending the analyses to maximum degree larger than $4$ and minimum degree smaller than $n-5$ remains an interesting project for future work. The analysis in Section \ref{sec:rev-eng} describes at least the final five Colour Refinement iterations on long-refinement graphs. For the long-refinement graphs with maximum degree at most $4$, it accounts for at least the final half of all iterations. Previous study of long-refinement graphs with general degree has consisted in computing possible \emph{split procedures} in a brute-force-like fashion \cite{Goedicke}. Assuming a relatively small number of singletons (or equivalently, a relatively high number of pairs) in the first iteration in which every class has size at most $2$, our structural results have the potential to significantly reduce the search space for split procedures, as they provide additional constraints that the previous approach did not exploit. This would potentially allow us to compute larger long-refinement graphs of arbitrary degrees.

Given the existence of infinite families of long-refinement graphs, in the light of the practical applications of Colour Refinement, it is an intriguing question to ask for pairs of long-refinement graphs that are only distinguished (i.e.\ receive different colourings) in the last Colour Refinement iteration on them. This search was initiated in \cite{KMcK20} and has been ongoing since then. In Section \ref{sec:long-ref-distinguish}, we close it by showing that such graphs do not exist. Our technique to capture long-refinement graphs via string representations opens the door to generalisations, e.g.\ to the construction of long-refinement graphs which take long to be distinguished by Colour Refinement. Without the requirement to be long-refinement graphs, the currently best known lower bound is $n - 8 \sqrt{n}$ \cite[Theorem 4.6]{krever15}. The task to improve on this lower bound to $n - c$ for a constant $c$ (which must be larger than $1$, due to Theorem \ref{thm:nolongdist}) is particularly captivating given that such pairs of graphs do not exist for higher dimensions of the Weisfeiler--Leman algorithm \cite{GroheLN23}. It is plausible that the bound of $n-1$ in the statement of Theorem \ref{thm:nolongdist} could be strengthened to $n-c$ for some small $c$. For such graphs, there can be at most $c-1$ -- hence very few -- iterations in which the number of colours increases by more than $1$, i.e.\ which differ from splittings in long-refinement graphs. Since the number of iterations is crucial for the parallelisability of the algorithm, improved upper bounds mean an improved complexity of a parallel implementations.

\bibliography{refs}

\begin{thebibliography}{10}

\bibitem{AndersS21}
Markus Anders and Pascal Schweitzer.
\newblock Engineering a fast probabilistic isomorphism test.
\newblock In Martin Farach{-}Colton and Sabine Storandt, editors, {\em
  Proceedings of the Symposium on Algorithm Engineering and Experiments,
  {ALENEX} 2021, Virtual Conference, January 10-11, 2021}, pages 73--84.
  {SIAM}, 2021.

\bibitem{ArvindKRV17}
Vikraman Arvind, Johannes K{\"{o}}bler, Gaurav Rattan, and Oleg Verbitsky.
\newblock Graph isomorphism, color refinement, and compactness.
\newblock {\em Computational Complexity}, 26(3):627--685, 2017.

\bibitem{AtseriasM13}
Albert Atserias and Elitza~N. Maneva.
\newblock {S}herali-{A}dams relaxations and indistinguishability in counting
  logics.
\newblock {\em {SIAM} J. Comput.}, 42(1):112--137, 2013.

\bibitem{AtseriasO18}
Albert Atserias and Joanna Ochremiak.
\newblock Definable ellipsoid method, sums-of-squares proofs, and the
  isomorphism problem.
\newblock In Anuj Dawar and Erich Gr{\"{a}}del, editors, {\em Proceedings of
  the 33rd Annual {ACM/IEEE} Symposium on Logic in Computer Science, {LICS}
  2018, Oxford, UK, July 09-12, 2018}, pages 66--75. {ACM}, 2018.

\bibitem{Babai16}
L{\'{a}}szl{\'{o}} Babai.
\newblock Graph isomorphism in quasipolynomial time [extended abstract].
\newblock In Daniel Wichs and Yishay Mansour, editors, {\em Proceedings of the
  48th Annual {ACM} {SIGACT} Symposium on Theory of Computing, {STOC} 2016,
  Cambridge, MA, USA, June 18-21, 2016}, pages 684--697. {ACM}, 2016.

\bibitem{BabErdSelSta80}
L\'{a}szl\'{o} Babai, Paul Erd\H{o}s, and Stanley~M. Selkow.
\newblock Random graph isomorphism.
\newblock {\em {SIAM} Journal on Computing}, 9(3):628--635, 1980.

\bibitem{BerkholzN16}
Christoph Berkholz and Jakob Nordstr{\"{o}}m.
\newblock Near-optimal lower bounds on quantifier depth and
  {W}eisfeiler-{L}eman refinement steps.
\newblock In {\em Proceedings of the 31st Annual {ACM/IEEE} Symposium on Logic
  in Computer Science, {LICS} '16, New York, NY, USA, July 5-8, 2016}, pages
  267--276, 2016.

\bibitem{cfi}
Jin-Yi Cai, Martin F\"{u}rer, and Neil Immerman.
\newblock An optimal lower bound on the number of variables for graph
  identification.
\newblock {\em Combinatorica}, 12(4):389--410, 1992.

\bibitem{carcro82}
Alain Cardon and Maxime Crochemore.
\newblock Partitioning a graph in {$O(|A| \log |A|)$}.
\newblock {\em Theoretical Computer Science}, 19:85--98, 1982.

\bibitem{DargaLSM04}
Paul~T. Darga, Mark~H. Liffiton, Karem~A. Sakallah, and Igor~L. Markov.
\newblock Exploiting structure in symmetry detection for {CNF}.
\newblock In {\em Proceedings of the 41th Design Automation Conference, {DAC}
  2004, San Diego, CA, USA, June 7-11, 2004}, pages 530--534. {ACM}, 2004.

\bibitem{DellGR18}
Holger Dell, Martin Grohe, and Gaurav Rattan.
\newblock Lov{\'{a}}sz meets {Weisfeiler and Leman}.
\newblock In Ioannis Chatzigiannakis, Christos Kaklamanis, D{\'{a}}niel Marx,
  and Donald Sannella, editors, {\em 45th International Colloquium on Automata,
  Languages, and Programming, {ICALP} 2018, July 9-13, 2018, Prague, Czech
  Republic}, volume 107 of {\em LIPIcs}, pages 40:1--40:14. Schloss Dagstuhl -
  Leibniz-Zentrum f{\"{u}}r Informatik, 2018.

\bibitem{Dvorak10}
Zdenek Dvor{\'{a}}k.
\newblock On recognizing graphs by numbers of homomorphisms.
\newblock {\em J. Graph Theory}, 64(4):330--342, 2010.

\bibitem{Furer01}
Martin F{\"{u}}rer.
\newblock {W}eisfeiler-{L}ehman refinement requires at least a linear number of
  iterations.
\newblock In {\em Automata, Languages and Programming, 28th International
  Colloquium, {ICALP} 2001, Crete, Greece, July 8-12, 2001, Proceedings},
  volume 2076 of {\em Lecture Notes in Computer Science}, pages 322--333.
  Springer, 2001.

\bibitem{GR24}
Andreas G{\"{o}}bel, Leslie~Ann Goldberg, and Marc Roth.
\newblock The {Weisfeiler-Leman} dimension of conjunctive queries.
\newblock {\em Proc. {ACM} Manag. Data}, 2(2):86, 2024.

\bibitem{Goedicke}
Maximilian G{\"{o}}dicke.
\newblock The iteration number of the {W}eisfeiler-{L}ehman-algo\-rithm.
\newblock Master's thesis, RWTH Aachen University, 2015.

\bibitem{god97}
Christopher~D. Godsil.
\newblock Compact graphs and equitable partitions.
\newblock {\em Linear Algebra Appl.}, 255:259--266, 1997.

\bibitem{GroheK21}
Martin Grohe and Sandra Kiefer.
\newblock Logarithmic {Weisfeiler-Leman} identifies all planar graphs.
\newblock In Nikhil Bansal, Emanuela Merelli, and James Worrell, editors, {\em
  48th International Colloquium on Automata, Languages, and Programming,
  {ICALP} 2021, July 12-16, 2021, Glasgow, Scotland (Virtual Conference)},
  volume 198 of {\em LIPIcs}, pages 134:1--134:20. Schloss Dagstuhl -
  Leibniz-Zentrum f{\"{u}}r Informatik, 2021.

\bibitem{GroheLN23}
Martin Grohe, Moritz Lichter, and Daniel Neuen.
\newblock The iteration number of the {Weisfeiler-Leman} algorithm.
\newblock {\em {ACM} Trans. Comput. Log.}, 26(1):6:1--6:31, 2025.

\bibitem{GroheLNS23com}
Martin Grohe, Moritz Lichter, Daniel Neuen, and Pascal Schweitzer.
\newblock Compressing {CFI} graphs and lower bounds for the {Weisfeiler-Leman}
  refinements.
\newblock {\em J. {ACM}}, 72(3):21:1--21:27, 2025.

\bibitem{groverb06}
Martin Grohe and Oleg Verbitsky.
\newblock Testing graph isomorphism in parallel by playing a game.
\newblock In {\em Automata, Languages and Programming, 33rd International
  Colloquium, {ICALP} 2006, Venice, Italy, July 10--14, 2006, Proceedings, Part
  {I}}, pages 3--14, 2006.

\bibitem{immlan90}
Neil Immerman and Eric Lander.
\newblock Describing graphs: A first-order approach to graph canonization.
\newblock In Alan~L. Selman, editor, {\em Complexity theory retrospective},
  pages 59--81. Springer, 1990.

\bibitem{JunttilaK07}
Tommi~A. Junttila and Petteri Kaski.
\newblock Engineering an efficient canonical labeling tool for large and sparse
  graphs.
\newblock In {\em Proceedings of the Nine Workshop on Algorithm Engineering and
  Experiments, {ALENEX} 2007, New Orleans, Louisiana, USA, January 6, 2007}.
  {SIAM}, 2007.

\bibitem{JunttilaK11}
Tommi~A. Junttila and Petteri Kaski.
\newblock Conflict propagation and component recursion for canonical labeling.
\newblock In Alberto Marchetti{-}Spaccamela and Michael Segal, editors, {\em
  Theory and Practice of Algorithms in (Computer) Systems - First International
  {ICST} Conference, {TAPAS} 2011, Rome, Italy, April 18-20, 2011.
  Proceedings}, volume 6595 of {\em Lecture Notes in Computer Science}, pages
  151--162. Springer, 2011.

\bibitem{kiefer:phd}
Sandra Kiefer.
\newblock {\em Power and Limits of the {W}eisfeiler-{L}eman Algorithm}.
\newblock PhD thesis, RWTH Aachen University, Aachen, 2020.

\bibitem{KMcK20}
Sandra Kiefer and Brendan~D. McKay.
\newblock {The Iteration Number of Colour Refinement}.
\newblock In Artur Czumaj, Anuj Dawar, and Emanuela Merelli, editors, {\em 47th
  International Colloquium on Automata, Languages, and Programming (ICALP
  2020)}, volume 168 of {\em Leibniz International Proceedings in Informatics
  (LIPIcs)}, pages 73:1--73:19, Dagstuhl, Germany, 2020. Schloss Dagstuhl --
  Leibniz-Zentrum f{\"u}r Informatik.

\bibitem{kieschw19}
Sandra Kiefer and Pascal Schweitzer.
\newblock Upper bounds on the quantifier depth for graph differentiation in
  first-order logic.
\newblock {\em Logical Methods in Computer Science}, 15(2), 2019.

\bibitem{kieschweisel22}
Sandra Kiefer, Pascal Schweitzer, and Erkal Selman.
\newblock Graphs identified by logics with counting.
\newblock {\em {ACM} Trans. Comput. Log.}, 23(1):1:1--1:31, 2022.

\bibitem{KoblerV08}
Johannes K{\"{o}}bler and Oleg Verbitsky.
\newblock From invariants to canonization in parallel.
\newblock In {\em Computer Science - Theory and Applications, Third
  International Computer Science Symposium in Russia, {CSR} 2008, Moscow,
  Russia, June 7-12, 2008, Proceedings}, volume 5010 of {\em Lecture Notes in
  Computer Science}, pages 216--227. Springer, 2008.

\bibitem{krever15}
Andreas Krebs and Oleg Verbitsky.
\newblock Universal covers, color refinement, and two-variable counting logic:
  Lower bounds for the depth.
\newblock In {\em Proceedings of the Thirtieth Annual {ACM/IEEE} Symposium on
  Logic in Computer Science}, pages 689--700, 2015.

\bibitem{LiSSSS16}
Wenchao Li, Hossein Saidi, Huascar Sanchez, Martin Sch{\"{a}}f, and Pascal
  Schweitzer.
\newblock Detecting similar programs via the {W}eisfeiler-{L}eman graph kernel.
\newblock In {\em Software Reuse: Bridging with Social-Awareness - 15th
  International Conference, {ICSR} 2016, Limassol, Cyprus, June 5-7, 2016,
  Proceedings}, pages 315--330, 2016.

\bibitem{lichponschwei19}
Moritz Lichter, Ilia~N. Ponomarenko, and Pascal Schweitzer.
\newblock Walk refinement, walk logic, and the iteration number of the
  {W}eisfeiler-{L}eman algorithm.
\newblock In {\em Proceedings of the Thirty-Fourth Annual {ACM/IEEE} Symposium
  on Logic in Computer Science}, pages 1--13, June 2019.

\bibitem{McKay81}
Brendan~D. McKay.
\newblock Practical graph isomorphism.
\newblock {\em Congressus Numerantium}, 30:45--87, 1981.

\bibitem{mckaypip14}
Brendan~D. McKay and Adolfo Piperno.
\newblock Practical graph isomorphism, {II}.
\newblock {\em Journal of Symbolic Computation}, 60:94--112, 2014.

\bibitem{morritfey+19}
Christopher Morris, Martin Ritzert, Matthias Fey, William~L. Hamilton, Jan~E.
  Lenssen, Gaurav Rattan, and Martin Grohe.
\newblock {Weisfeiler} and {Leman} go neural: Higher-order graph neural
  networks.
\newblock In {\em Proceedings of the Thirty-Third AAAI Conference on Artificial
  Intelligence}, January 2019.

\bibitem{ramscheiull94}
Motakuri~V. Ramana, Edward~R. Scheinerman, and Daniel Ullman.
\newblock Fractional isomorphism of graphs.
\newblock {\em Discrete Mathematics}, 132(1--3):247--265, 1994.

\bibitem{ShervashidzeSLMB11}
Nino Shervashidze, Pascal Schweitzer, Erik~Jan van Leeuwen, Kurt Mehlhorn, and
  Karsten~M. Borgwardt.
\newblock {Weisfeiler-Lehman} graph kernels.
\newblock {\em Journal of Machine Learning Research}, 12:2539--2561, 2011.

\bibitem{tin91}
Gottfried Tinhofer.
\newblock A note on compact graphs.
\newblock {\em Discrete Applied Mathematics}, 30(2--3):253--264, 1991.

\bibitem{vBGKO23}
Steffen van Bergerem, Martin Grohe, Sandra Kiefer, and Luca Oeljeklaus.
\newblock Simulating logspace-recursion with logarithmic quantifier depth.
\newblock In {\em 38th Annual {ACM/IEEE} Symposium on Logic in Computer
  Science, {LICS} 2023, Boston, MA, USA, June 26-29, 2023}, pages 1--13.
  {IEEE}, 2023.

\bibitem{WeisfeilerL68}
Boris Weisfeiler and Andrei Leman.
\newblock The reduction of a graph to canonical form and the algebra which
  appears therein.
\newblock {\em NTI, Series 2}, 1968.
\newblock English translation by G.~Ryabov available at
  \url{https://www.iti.zcu.cz/wl2018/pdf/wl_paper_translation.pdf}.

\end{thebibliography}
\end{document}